%% file: terrain_vis_icde.tex
\documentclass[conference]{IEEEtran}
\usepackage{pdfpages}
\usepackage{balance}  
\usepackage{algorithm}
\usepackage[noend]{algorithmic}
\usepackage{algorithmic}
\usepackage{epstopdf}
\usepackage{comment}
\usepackage[tight,scriptsize]{subfigure}
\usepackage{amsmath}
\usepackage{times,epsfig}
\usepackage{multirow}
\usepackage{pbox}

\ifCLASSINFOpdf
\else
\fi
\begin{document}

\newtheorem{claim}{Claim}
\newtheorem{definition}{Definition}
\newtheorem{lemma}{Lemma}
\newtheorem{theorem}{Theorem}
\newtheorem{proposition}{Proposition}

%
\title{Analyzing and Visualizing Scalar Fields on Graphs}

\author{\IEEEauthorblockN{Yang Zhang}
\IEEEauthorblockA{Department of Computer Science\\
and Engineering\\
The Ohio State University\\
Columbus, Ohio, USA, 43210\\
Email: zhang.863@osu.edu}
\and
\IEEEauthorblockN{Yusu Wang}
\IEEEauthorblockA{Department of Computer Science\\
and Engineering\\
The Ohio State University\\
Columbus, Ohio, USA, 43210\\
Email: yusu@cse.ohio-state.edu}
\and
\IEEEauthorblockN{Srinivasan Parthasarathy}
\IEEEauthorblockA{Department of Computer Science\\
and Engineering\\
The Ohio State University\\
Columbus, Ohio, USA, 43210\\
Email: srini@cse.ohio-state.edu
}}


%


\maketitle
\baselineskip=0.90\normalbaselineskip

\begin{abstract}
The value proposition of a dataset often resides in the implicit interconnections or explicit
relationships (patterns) among individual entities, and is often modeled as a graph.
Effective visualization of such graphs can lead to key insights uncovering such value.
In this article we propose a visualization method to explore graphs with numerical
attributes associated with nodes (or edges) -- referred to as scalar graphs.
Such numerical attributes can represent raw content information, similarities, or
derived information reflecting important network measures such
as triangle density and centrality.
The proposed visualization strategy
seeks to simultaneously uncover the relationship between attribute values and graph topology,
 and relies on transforming the network to generate a terrain map.
A key objective here is to ensure that the terrain map
reveals the overall distribution of
components-of-interest (e.g. dense subgraphs, k-cores) and the relationships among them
while being sensitive to the attribute values over the graph.
We also design extensions that can capture the  relationship across multiple
numerical attributes (scalars).
We demonstrate the efficacy of our method on several real-world data science tasks
while scaling to large graphs with millions of nodes.
\end{abstract}


%
\IEEEpeerreviewmaketitle

\input{Introduction}
\input{SingleScalar}
\input{MultipleScalar}

\input{Evaluation}

\input{Discussion}
\input{UserStudy}

\vspace{-0.10in}

\section{Conclusion}
In this paper, we propose a visualization method to analyze scalar graphs.
We demonstrate that our visualization method can reveal important information
such as the overall distribution of attribute values over a graph,
while simultaneously highlighting components-of-interest (dense subgraphs, social communities, etc.).
It can also be extended to analyze relationship between multiple attributes.
We evaluate our system on a range of real-world data and
demonstrated its effectiveness and scalability on a range of data science tasks. 
One interesting avenue of ongoing research is to incorporate these ideas
within a database -- where one can view the result of a query as an attributed graph (attributed by similarities among query result tuples). We envision that such visualizations can aid and abet users in creating more natural interfaces to interrogate and understand their data.

\vspace{-0.10in}





%
\bibliographystyle{IEEEtran}
\bibliography{icde_ref}

\end{document}

%% file: Introduction.tex
\section{Introduction}
Our ability to produce and store data has far outstripped our ability to analyze and utilize this data to derive actionable insight.
Many phenomena and  problems from all walks of human endeavor can often be represented in graph or network form, where nodes represent entities-of-interest and edges represent interactions or relationships among them. Examples abound across the physical, biological, business, technological, sociological and health domains.
A fundamental challenge is the ability to visualize such interconnections at scale while working within the pixel limits of modern displays.
This in turn has led to several recent advances in the database systems~\cite{csv,clus_model,wang:dngraph,tkcore12}, network science\cite{Vlad:kcore,kcore_vis}, geometry\cite{isosurface,multi_field,landscapes,terrain10}, and information visualization~\cite{graphvis1,measurevis1} communities.

However, the data scientist is often interested in uncovering patterns that go beyond layout
and visualization designs -- for instance accounting for measures
or attributes defined on both nodes and edges of the graph~\cite{clus_attr, clus_model}.
Such measures often encode information about connectivity -- both locally (e.g. triangle density,
clustering coefficient, K-Core, K-Truss) as well as globally (e.g.
various centrality measures such as closeness, betweenness and harmonic,
and also measures of importance such as PageRank and Influence).
Beyond just measures from the topology, such measures may also incorporate heterogeneous information related to content (e.g. sequence information of a protein etc.).
Visualizing the measure information in such graphs (where each node or edge has one or more scalar measures associated with it) further exacerbates the challenge.

In this article we propose a novel visualization method to explore graphs with numerical
measures associated with nodes (or edges) -- referred to as \emph{scalar graphs}, and the measure values are referred to as \emph{scalar values}.
Each scalar value could either be a natural attribute or a derived attribute summarizing information from
multiple natural attributes (e.g. triangle density, centralities, cliquishness)~\cite{tkcore12,csv}.
We model the problem
as a continuous function $f: X \to \mathcal{R}$: the domain $X$ is a simplicial complex whose vertex set is the set of input graph nodes, and its topology is determined by the input graph topology.
We call such a representation of a graph a \emph{graph-induced scalar field}.
We then leverage a powerful ``terrain metaphor" idea for visualizing scalar fields~\cite{landscapes, terrain10} and adapt it to our context.
Our visualization model {\it naturally encodes both topological and numerical measure information together},
and is capable of handling large graphs with millions of nodes/edges.
Empirically we demonstrate the use of our methodology on a range of data science tasks while providing a comparative assessment
against state-of-the-art alternatives from the perspectives of efficacy, efficiency and usability.

Figure~\ref{TerrainExample} previews our methodology for the tasks of:  i)
visualizing dense subgraphs (K-Cores, K-Trusses, etc); and ii)
visualizing communities in social networks.
In Figure~\ref{GrQc_3d_intro}, we use K-Core number~\cite{Vlad:kcore} as a measure on each node and use it to visualize K-Cores in a collaboration network (a fundamental
network science task), where the top part of high peaks contains dense K-Cores
(Clicking on the circled part of the high peak will show the details of a dense K-Core in the red box).
In Figure~\ref{dblp4comm_intro}, we use community score~\cite{overlap13} as a measure to visualize the four communities in a DBLP network,
where each major peak is a community, and the sub-peaks in a major peak indicate sub-communities.
The top part of a peak contains key members of the community.
Our proposed tool can also color the terrain using a second measure,
so one can analyze the correlation between two different measures on the graph and furthermore allows for various rotation/transformation operators which allow an end-user
to interrogate the data from multiple perspectives.
Further details of our evaluation is detailed in Section~\ref{evaluation}.
In summary, our contributions are:

\begin{figure}[h]
\centering
\subfigure[Visualizing K-Cores]
{
\includegraphics[width=0.23\textwidth]{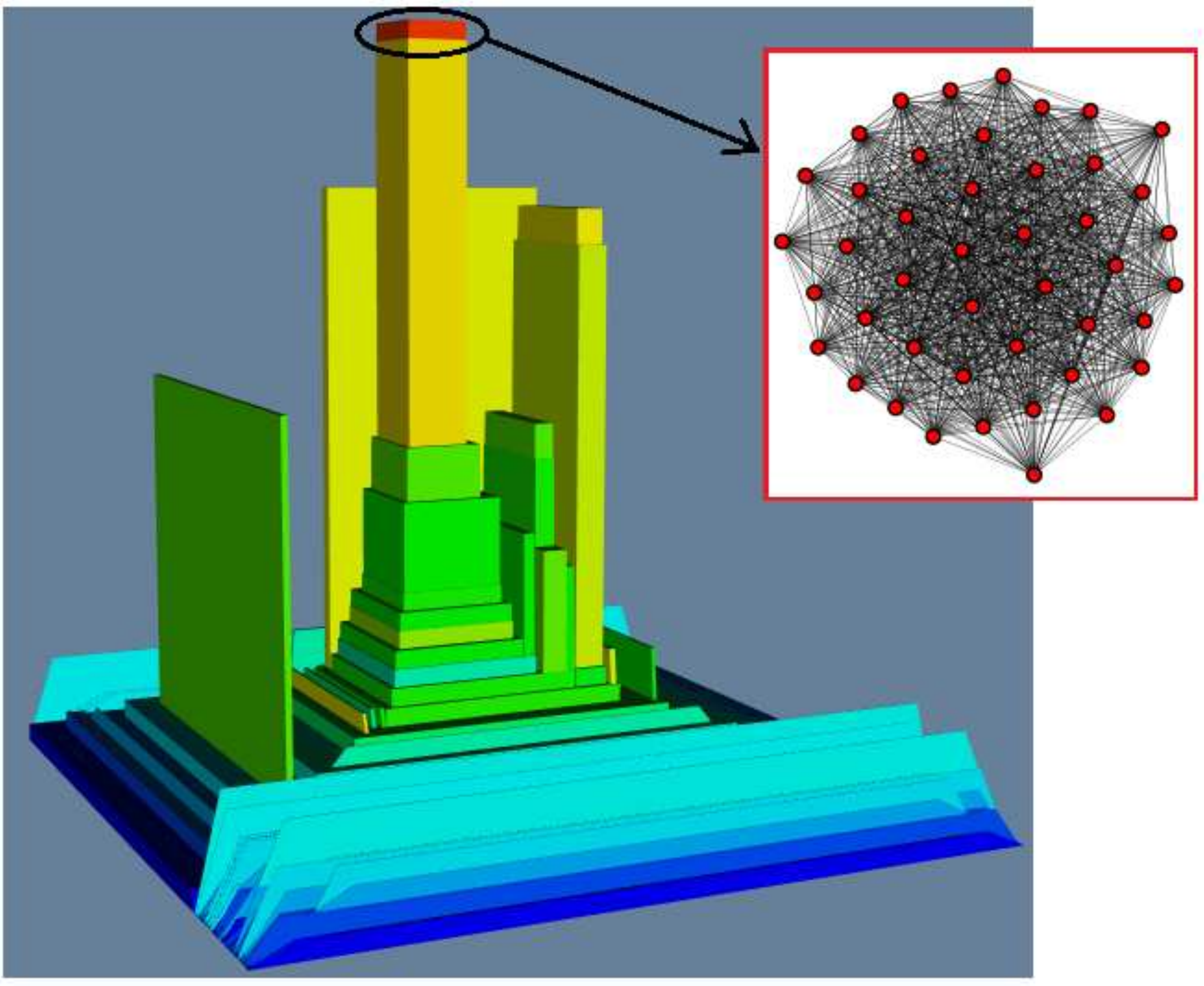}
\label{GrQc_3d_intro}
}
 \quad
\subfigure[Visualizing Communities]
{
\includegraphics[width=0.195\textwidth]{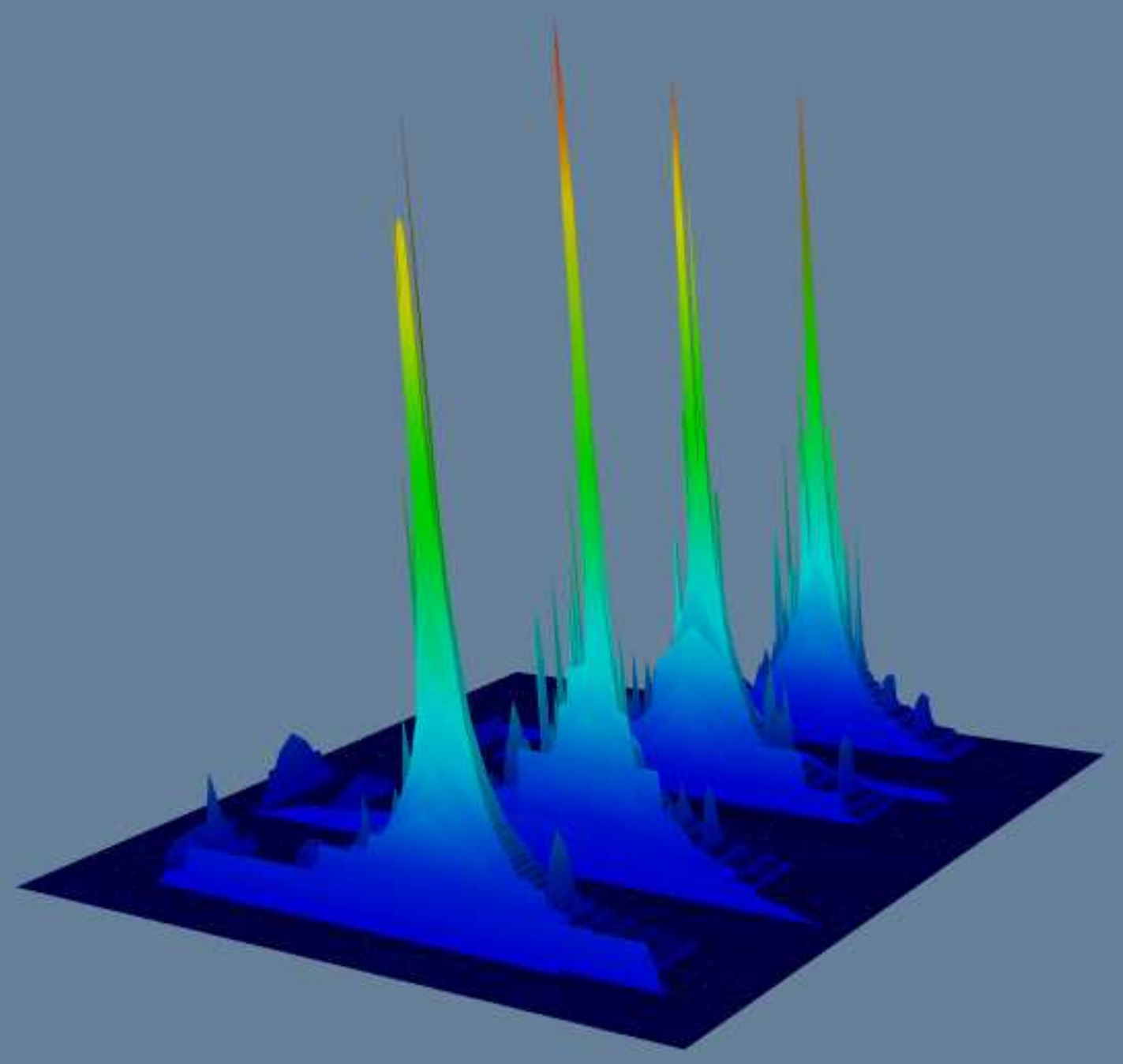}
\label{dblp4comm_intro}
}

\vspace{-0.10in}
\caption[Optional caption for list of figures]{Examples of Terrain Visualization}
\label{TerrainExample}
\end{figure}

%

\begin{itemize}
  \item We advocate a novel visualization method to visualize a graph whose vertices/edges are associated with numerical measures. Central to our approach is to embed the graph
        in {\it terrain}  space.
  The visualization method enables users to quickly capture how the numerical measure value is distributed over the graph.
  \item We propose the concept of maximal $\alpha$-connected component to represent various graph patterns (dense subgraphs, communities, K-Cores, K-Truss etc.), and show that our visualization method not only captures the distribution of graph patterns, but also their global relationships.
  \item Based on our visualization method, we propose methods to analyze the relation between multiple numerical measures on one graph.

\item Finally, we empirically demonstrate that the visualization method is both general-purpose and scalable and can handle graphs with millions of nodes and edges. A key part of our evaluation strategy includes a user-study following best practices (Section~\ref{user_sec}).

\end{itemize}

%% file: SingleScalar.tex
\vspace{-0.05in}
\section{Scalar Graph Visualization}

\noindent {\bf Notation:}
We define a vertex-based scalar graph \emph{G(V, E)} as a graph comprising edges \emph{E} and vertices \emph{V},
where each vertex $v$ has one scalar value associated with it, denoted as \emph{v.scalar}.
In the following we refer to vertex-based scalar graph as a {\it scalar graph} for expository
simplicity. We also assume the following notion in the rest of this section.
For two subgraphs $G'(V', E')$ and $G''(V'', E'')$ of scalar graph G,
$G'$ is the same as $G''$ (denoted as $G' = G'') $ if $V'=V''$ and $E'=E''$,
$G'$ is a subgraph of $G''$ (denoted as $G' \subseteq G''$) if $V' \subseteq V''$ and $E' \subseteq E''$,
$G'$ is connected to $G''$ (denoted as $G' \leftrightarrow G''$) if there is a vertex $v'\in V'$ and a vertex $v''\in V''$ such that $v'$ is connected to $v''$.

\subsection{Theoretical Insights}
We now describe the key insights underpinning our visualization strategy.
We first define \textbf{maximal $\alpha$-connected component} as follows:

\noindent \begin{definition}
A connected component $C(V_{C}, E_{C})$ of scalar graph $G(V,E)$ is a \textbf{maximal $\alpha$-connected component} if it satisfies following conditions: \\
(1) for every vertex $v \in V_{C}$, $v.scalar >= \alpha$.\\
(2) for any vertex $v \in V_{C}$, if $v'$ is connected to $v$ and $v' \not\in V_{C}$, then $v'.scalar < \alpha$.\\
(3) for any edge $e(v_{1}, v_{2}) \in E$, if $v_{1} \in V_{C}$ and $v_{2} \in V_{C}$, then $e(v_{1}, v_{2}) \in E_{C}$.
\end{definition}

Maximal $\alpha$-connected components are an important aspect of our visualization strategy.
We note that the definition accommodates both connectivity (topology) and scalar value information.
The scalar values may be inherited from the domain or derived (e.g. local density or k-core values\cite{Vlad:kcore})). The
distribution of maximal $\alpha$-connected components in a graph would allow the data scientist to
interrogate the distribution and inter-connectivity of such structures within the graph (e.g.
distribution and relationships of K-cores).




\noindent \begin{definition}
For vertex $v$, we define \textbf{MCC($v$)}
as the maximal $v.scalar$-connected component containing $v$ (i.e. $\alpha = v.scalar$).
\end{definition}

\noindent \begin{theorem}
For any maximal $\alpha$-connected component $C$ in $G$, there is a vertex $v$ in $G$ such that $MCC(v)= C$.
\label{theorem_sg1}
\end{theorem}

\noindent \emph{\textbf{Proof Sketch:}}
It is easy to show that $v$ in Theorem~\ref{theorem_sg1} is the vertex with minimum scalar value in $C$.
The proof trivially follows from this observation.

\noindent \begin{theorem}
For two vertices $v$ and $v'$, if $v.scalar = v'.scalar$ and $MCC(v)$ contains $v'$, then $MCC(v) = MCC(v')$.
\label{theorem_sg2}
\end{theorem}

\noindent \emph{\textbf{Proof Sketch:}} For every vertex $v_{i}$ in $MCC(v)$,
we can show that in $MCC(v)$ there is a path $v_{i}\rightarrow...\rightarrow v\rightarrow...\rightarrow v'$, which starts with $v_{i}$, passes through $v$, and ends at $v'$,
and all the vertices on the path have scalar value greater than or equal to $v'.scalar$, so $v_{i}$ is in $MCC(v')$. Similarly, we can show for every vertex $v_{j}$ in $MCC(v')$, $v_{j}$ is in $MCC(v)$. Thus $MCC(v) = MCC(v')$.

\noindent \begin{theorem}
For a maximal $\alpha_{1}$-connected component $C_{1}$ and another maximal $\alpha_{2}$-connected component $C_{2}$,
if $C_{1} \leftrightarrow C_{2}$ then $C_{1} \subseteq C_{2}$ or $C_{2} \subseteq C_{1}$ .
\label{theorem_sg3}
\end{theorem}

\noindent \emph{\textbf{Proof Sketch:}} Based on Theorem \ref{theorem_sg1}, there are two vertices $v_{1}$ and $v_{2}$ such that $MCC(v_{1})= C_{1}$ and
$MCC(v_{2})=C_{2}$. Assume $v_{1}.scalar \leq v_{2}.scalar$ (w.l.g).
Since $MCC(v_{1})\leftrightarrow MCC(v_{2})$, for any vertex $v'$ in $MCC(v_{2})$,
we can find a path $v_{1}\rightarrow...\rightarrow v'$ connecting $v_{1}$ and $v'$, and
all the vertices on the path have scalar value greater than or equal to $v_{1}.scalar$, so $v'$ is in $MCC(v_{1})$. Thus $MCC(v_{2}) \subseteq MCC(v_{1})$.

Theorems \ref{theorem_sg1} and \ref{theorem_sg2} bound the
number of distinct maximal $\alpha$-connected components that need to be processed by any
algorithm (number of vertices in the graph).  Theorem \ref{theorem_sg3} highlights an
important hierarchical relationship between different maximal $\alpha$-connected components which
will be exploited by our visualization strategy.

\subsection{Vertex Scalar Tree}
In this section, we describe
the vertex scalar tree (\emph{scalar tree} for short) to analyze a scalar graph.
We note that if one views the graph as a 1-dimensional simplicial complex,
the vertex scalar values induce a piecewise-linear function on this domain,
then the maximal $\alpha$-connected components we defined previously are
akin to level sets or contours for each $\alpha$~\cite{contour00, isosurface}.
We also note that this perspective also helps us easily extend the notion of scalar tree for edge-based scalar graphs (in Section 2.3).

A scalar tree is a tree in which every node is associated with a scalar value,
and the scalar tree T of scalar graph G has the following properties:
\begin{enumerate}
  \item Every node in T corresponds to a vertex in G with the same scalar value, and vice versa (i.e. one-to-one correspondence).
  \item Every maximal $\alpha$-connected component in G corresponds to a subtree in T, and vice versa (i.e. one-to-one correspondence).
  \item Assume a maximal $\alpha_{1}$-connected component $C_{1}$ corresponds to subtree $T_{1}$ in T, and another maximal $\alpha_{2}$-connected component $C_{2}$ corresponds to subtree $T_{2}$ in T,
      then $C_{1}$ is a subgraph of $C_{2}$ if and only if $T_{1}$ is subtree of $T_{2}$.
  \item Assume a maximal $\alpha_{1}$-connected component $C_{1}$ corresponds to subtree $T_{1}$ in T, and another maximal $\alpha_{2}$-connected component $C_{2}$ corresponds to subtree $T_{2}$ in T,
      then $C_{1}$ and $C_{2}$ are not connected if and only if $T_{1}$ and $T_{2}$ are not connected.
\end{enumerate}

\smallskip
\noindent \textbf{Notation:} In the following text,
the scalar tree node corresponding to vertex \emph{v} is denoted as \emph{n(v)},
and the vertex corresponding to scalar tree node \emph{n} is denoted as \emph{v(n)}.
The parent of tree node $n$ is denoted as $parent(n)$.
The subtree corresponding to a maximal $\alpha$-connected component \emph{C} is denoted as \emph{ST(C)}.
A subtree \emph{ST} containing nodes $n_{1}, n_{2} ... n_{k}$ is denoted as $ST(n_{1}, n_{2}, ... , n_{k})$,
and a connected component C containing vertices $v_{1}, v_{2}, ... v_{k}$ is denoted as $C(v_{1}, v_{2}, ... ,v_{k})$.

\smallskip
\noindent {\bf Example:} In Figure~\ref{ScalarExample} we use an example to illustrate the properties of scalar tree.
Figure~\ref{example_graph} is a scalar graph G, in which every vertex's label is its scalar value.
Figure~\ref{example_tree} is a correspondent scalar tree T of Figure~\ref{example_graph},
rooted at node $n_{9}$.
Node $n_{i}$ in Figure~\ref{example_tree} corresponds to vertex $v_{i}$ in Figure~\ref{example_graph} (Property 1).
In Figure~\ref{example_graph_cut}, we extract all the maximal 2.5-connected components of G:\\
$C_{1}(v_{1}, v_{2}, v_{3}, v_{5})$ and $C_{2}(v_{4}, v_{6})$,
and in Figure~\ref{example_tree},
their correspondent subtrees are: $ST(C_{1}) = ST(n_{1}, n_{2}, n_{3}, n_{5})$ and $ST(C_{2}) = ST(n_{4}, n_{6})$,
this satisfies Property 2.
We notice that $C_{1}$ and $C_{2}$ are not connected, and $ST(C_{1})$ and $ST(C_{2})$ are not connected either, this satisfies Property 4.
We also observe that $C_{1}$ is a subgraph of a maximal 2-connected component $C_{3}(v_{1}, v_{2}, v_{3}, v_{4}, v_{5}, v_{6}, v_{7})$,
and $ST(C_{1})$ is a subtree of $ST(C_{3}) = ST(n_{1}, n_{2}, n_{3}, n_{4}, n_{5}, n_{6}, n_{7})$,
this satisfies Property 3.

\begin{figure}[h]
\centering
\subfigure[Scalar Graph G]
{
\includegraphics[width=0.13\textwidth]{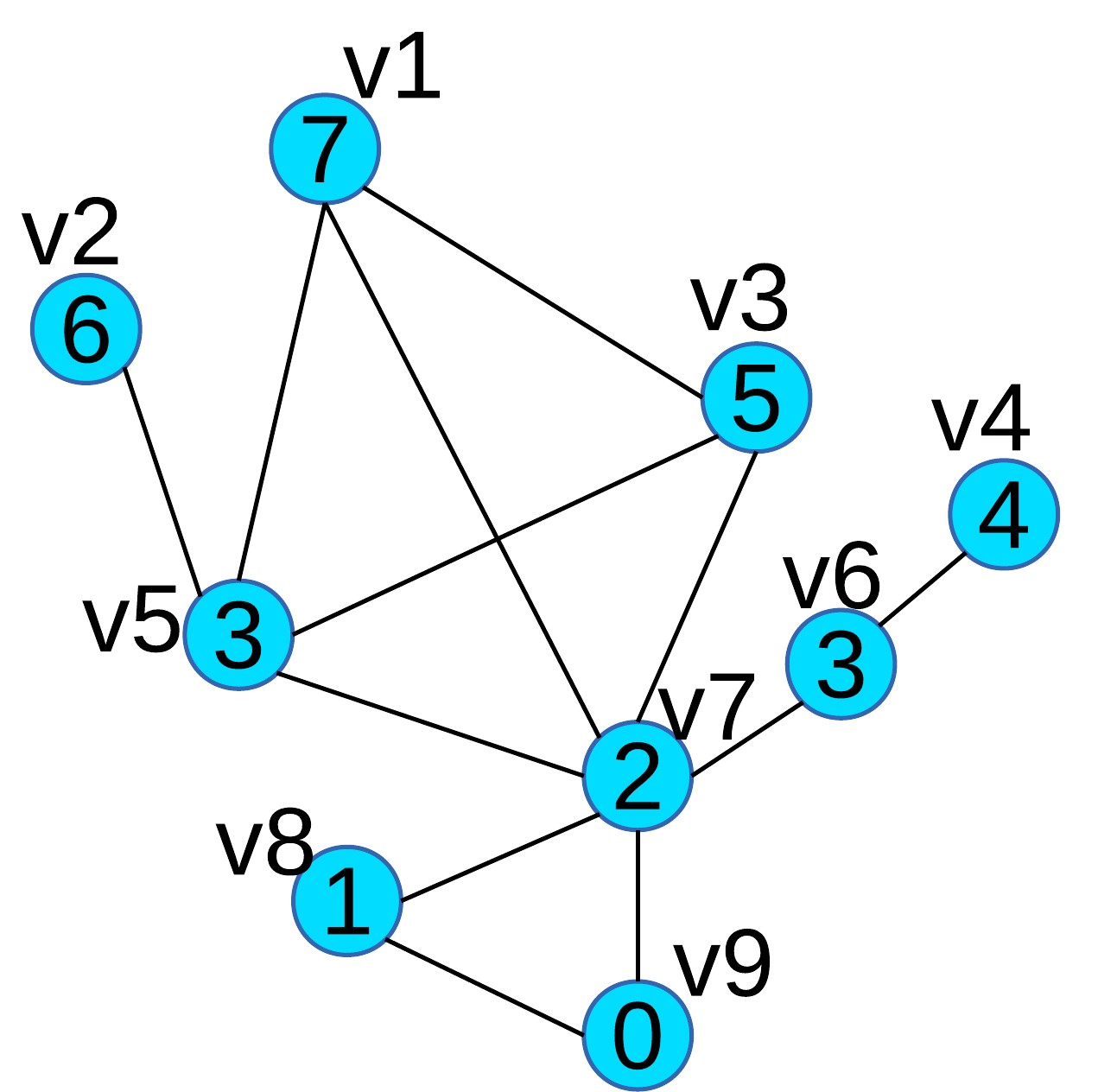}
\label{example_graph}
}
 \quad
\subfigure[Scalar Tree T]
{
\includegraphics[width=0.11\textwidth]{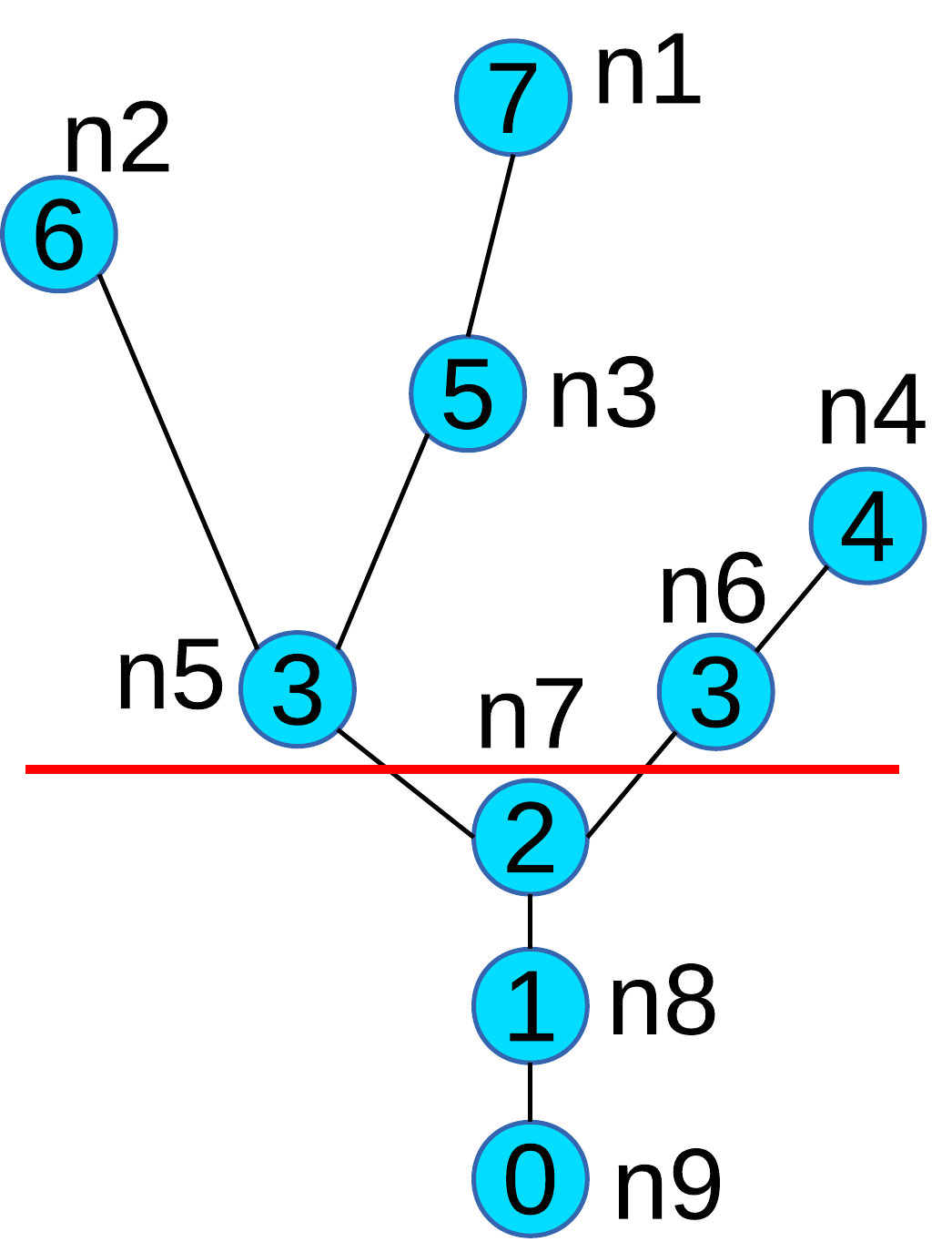}
\label{example_tree}
}
 \quad
\subfigure[maximal 2.5-connected components]
{
\includegraphics[width=0.13\textwidth]{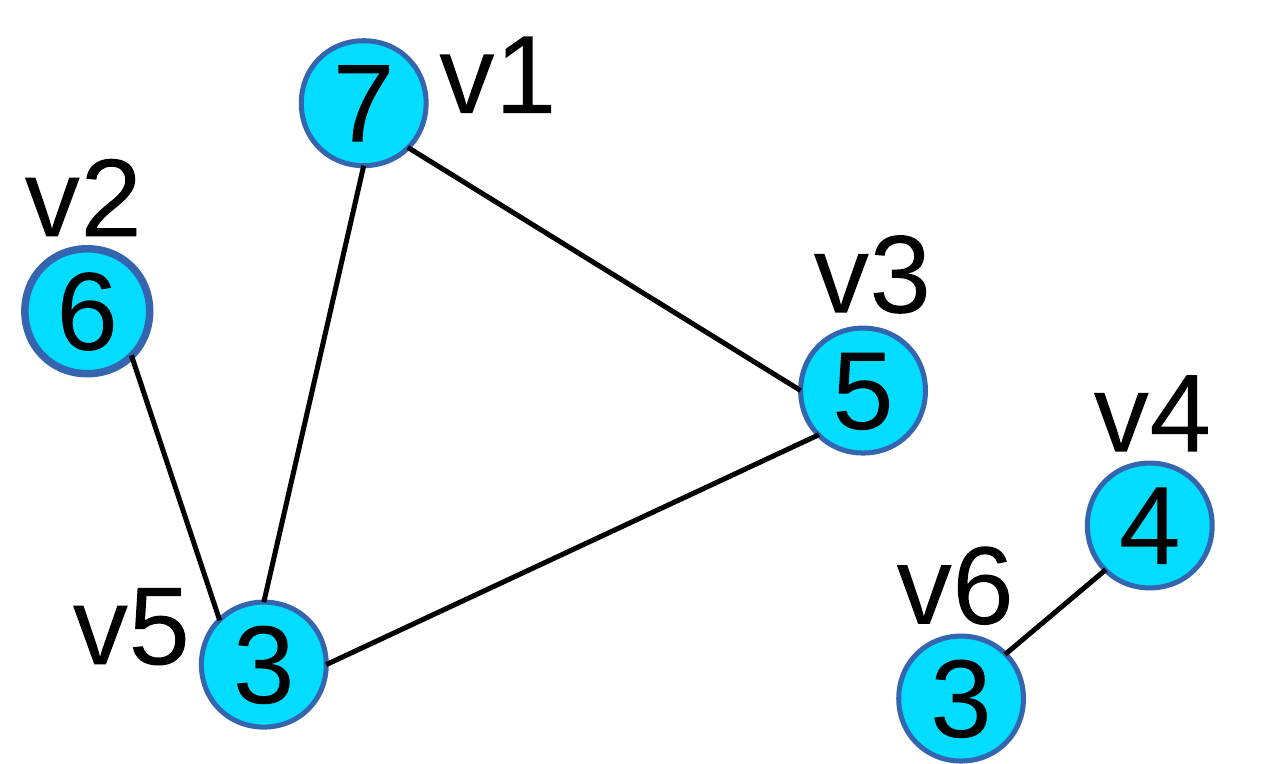}
\label{example_graph_cut}
}

\vspace{-0.10in}
\caption[Optional caption for list of figures]{Scalar Graph and Scalar Tree}
\label{ScalarExample}
\end{figure}

From the properties above, we can see that scalar tree captures the distribution and relationship of
all maximal $\alpha$-connected components within a scalar graph. This allows one
to analyze the scalar tree instead of the scalar graph and leverage it for our visualization
strategy.

\smallskip
\noindent{\bf Constructing the Vertex Scalar Tree:}
Algorithm~\ref{ScalarTree} constructs the scalar tree by leveraging
the observation connecting our problem with level sets and
contour trees\cite{contour00, isosurface}.
Algorithm~\ref{ScalarTree} processes all the vertices in decreasing order of scalar values (line 1-3).
If the current vertex $v_{i}$ is connected to a previously processed vertex $v_{j}$,
but $n(v_{i})$ is not in the current subtree of $n(v_{j})$ (line 4-5),
then we connect $n(v_{i})$ with the root of the current subtree of $n(v_{j})$ (line 6).
Here $root(n(v_{j}))$ denotes the root node of the current subtree containing $n(v_{j})$.
Note that now $n(v_{i})$ is parent of $root(n(v_{j}))$, so $n(v_{i})$ becomes the new root of the current subtree containing $n(v_{j})$.

\begin{algorithm}[]
\fontsize{9}{10}\selectfont
\caption{Constructing Scalar Tree}
\begin{algorithmic}[1]
\REQUIRE A scalar graph G(V, E).
\ENSURE The scalar tree T of G.
\STATE Sort vertices in decreasing order of scalar values, the sorted vertices are $v_{1}, v_{2}, ... v_{n}$;
\STATE Create a tree node $n(v_{i})$ for each vertex $v_{i}$;
\FOR{$i=1 $ to $n$}
\FOR { every neighbor $v_{j}$ of $v_{i}$}
\IF {$j < i$ and currently $n(v_{i})$ and $n(v_{j})$ are not in the same subtree}
\STATE Connect $n(v_{i})$ to $root(n(v_{j}))$; //$n(v_{i})$ is parent
\ENDIF
\ENDFOR
\ENDFOR
\end{algorithmic}
\label{ScalarTree}
\end{algorithm}

The running time of line 1 is $O(|V|*log|V|)$.
An efficient implementation of line 5 uses the ``Union Find'' data structure, which compares $root(n(v_{i}))$ and $root(n(v_{j}))$.
The amortized time for ``Union Find'' (line 5) is $O(\alpha(n))$ per operation, where $\alpha(n)$ is  inverse of Ackermann function,
and is usually a small constant. So the total worst-case time cost of Algorithm~\ref{ScalarTree} is $O(|E|*\alpha(n) + |V|*log|V|)$

In the scalar tree T generated by Algorithm~\ref{ScalarTree},
every node's scalar value is greater than or equal to its parent's scalar value.
If we layout a scalar tree T in such a way that the height of each node is the scalar value of the node,
then we could get all the maximal $\alpha$-connected components for a particular $\alpha$ in a simple way:
draw a line with $height = \alpha$ to cross T,
and each of the subtrees above the line corresponds to a maximal $\alpha$-connected component.
For example, as Figure~\ref{example_tree} shows,
the two subtrees above red line $height = 2.5$ correspond to all the maximal $2.5$-connected components.
When every vertex in the input scalar graph G has a distinct scalar value,
the scalar tree T generated by Algorithm~\ref{ScalarTree} has the following property:
\begin{proposition}
For every vertex v in G, assume it corresponds to n(v) in T,
then the subtree rooted at n(v) (denoted as $ST$) in T corresponds to the MCC(v).
\label{tree_mcc}
\end{proposition}
\begin{proof}
Obviously $ST$ corresponds to a connected component $SG$.
Since $n(v).scalar$ is the minimum scalar value in $ST$,
every vertex in $SG$ has scalar value greater than or equal to $v.scalar$.
If there is a vertex $v_{i}$ that connects to a vertex $v_{j}$ in $SG$, and $v_{i}.scalar > v.scalar$,
then due to Algorithm~\ref{ScalarTree},
$n(v_{i})$ and $n(v_{j})$ will be in the same subtree $ST$,
which indicates that $v_{i}$ is in $SG$.
So $SG$ is a maximal $v.scalar$-connected component, and it is MCC(v).
\end{proof}

Based on Theorem \ref{theorem_sg1}, Theorem \ref{theorem_sg3} and Proposition \ref{tree_mcc},
it follows that when every vertex in the input scalar graph has a {\it distinct value},
the tree generated by Algorithm~\ref{ScalarTree} has the four properties of the scalar tree.
When some vertices in scalar graph have the same scalar value,
Algorithm~\ref{ScalarTree} may not guarantee Property 2,
and we need to do some additional postprocessing, described next.

\smallskip
\noindent {\bf Postprocessing the Vertex Scalar Tree:}
If some vertices in G have the same scalar values,
the scalar tree generated by Algorithm~\ref{ScalarTree} may not satisfy Property 2
 -- a subtree may not correspond to a maximal $\alpha$-connected component.
For example, Figure~\ref{example_super_graph} is a scalar graph,
and Figure~\ref{example_super_tree} is the scalar tree generated by Algorithm~\ref{ScalarTree},
in which $n_{i}$ corresponds to $v_{i}$ in the scalar graph.
The subtree rooted at $n_{3}$ is $ST(n_{1}, n_{3})$ corresponds to connected component $C(v_{1}, v_{3})$,
however, $C(v_{1}, v_{3})$ is not a maximal $\alpha$-connected component.

\begin{figure}[h]
\centering
\subfigure[Scalar Graph G]
{
\includegraphics[width=0.13\textwidth]{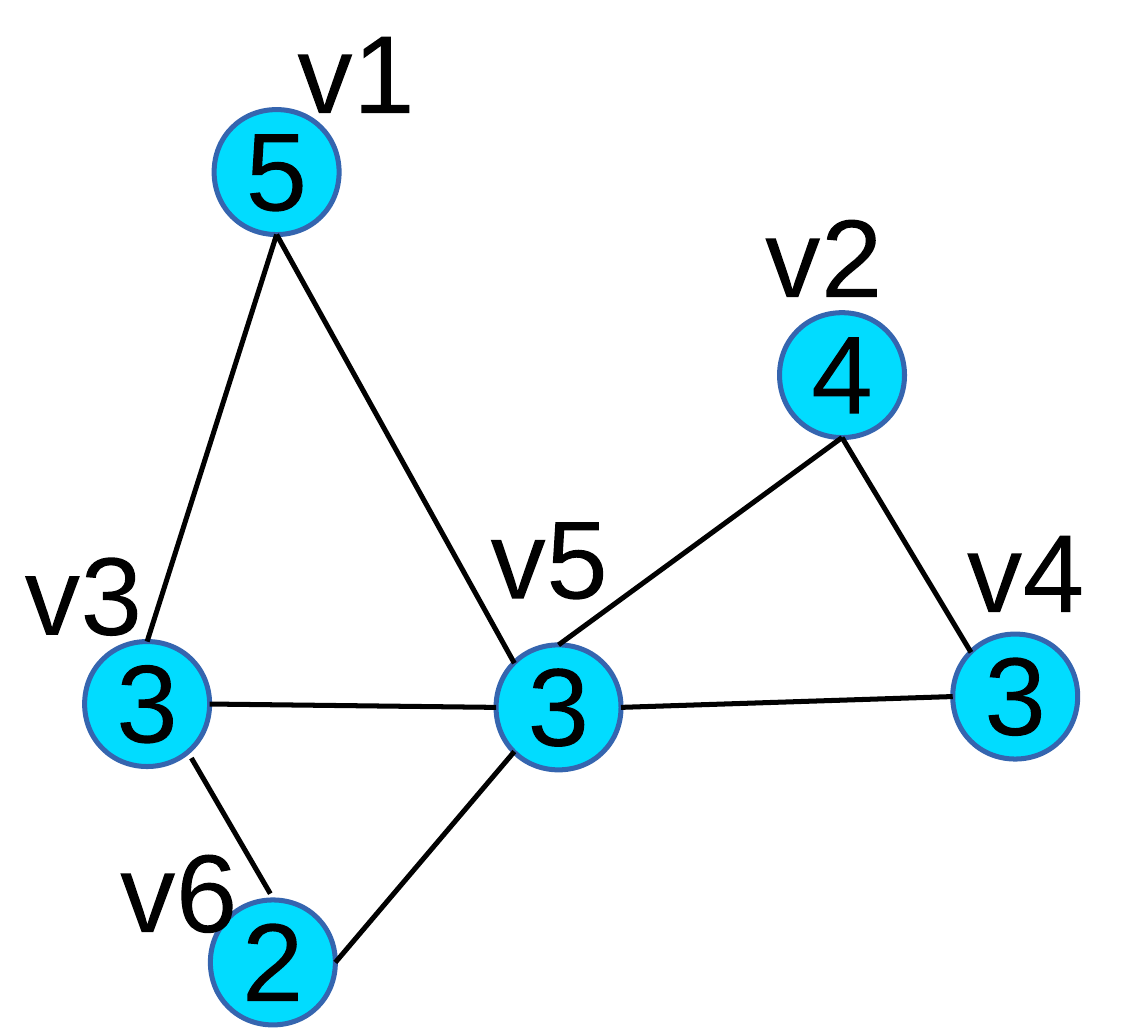}
\label{example_super_graph}
}
 \quad
\subfigure[Scalar Tree T]
{
\includegraphics[width=0.11\textwidth]{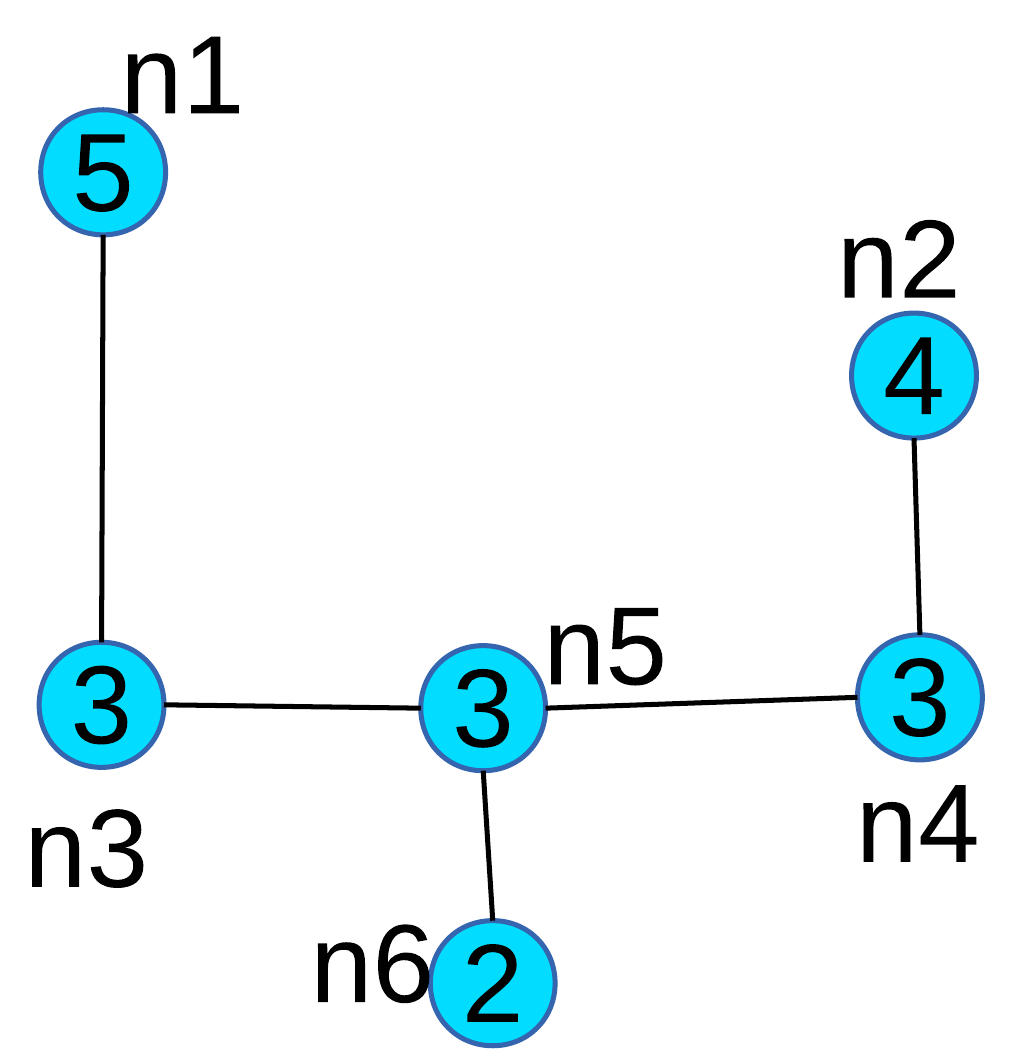}
\label{example_super_tree}
}
 \quad
\subfigure[Super Scalar Tree]
{
\includegraphics[width=0.13\textwidth]{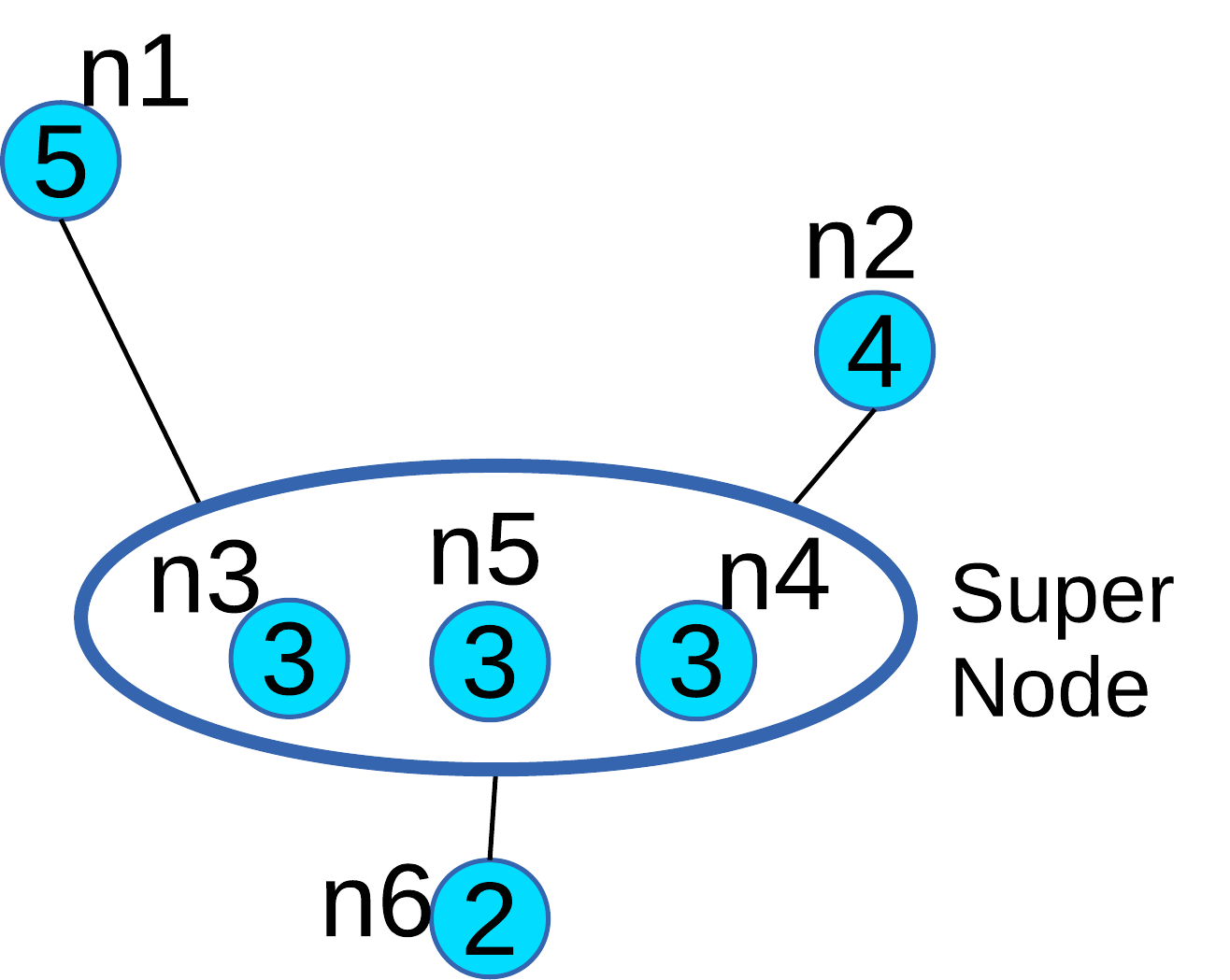}
\label{example_super_tree_merge}
}
\vspace{-0.10in}
\caption[Optional caption for list of figures]{Postprocessing Scalar Tree}
\label{SuperScalarExample}
\end{figure}

To solve the problem, we use Algorithm~\ref{PostScalarTree} to postprocess the scalar tree T generated by Algorithm~\ref{ScalarTree}:
we merge the ancestor node with all its descendants with the same scalar value into a super node.
The correctness of the algorithm is based on the following proposition:

\begin{proposition}
For any tree node $n$ in T, assume it has an ancestor $n_{anc}$ in T, such that
$n_{anc}.scalar = n.scalar$, and $parent(n_{anc})$ is null or $parent(n_{anc}).scalar < n_{anc}.scalar$,
then the subtree rooted at $n_{anc}$ (denoted as $ST_{anc}$) corresponds to the MCC(v(n)).
\label{anc_mcc}
\end{proposition}

In the example in Figure~\ref{example_super_tree}, Algorithm~\ref{PostScalarTree} will merge the nodes $n_{3}, n_{4}, n_{5}$ into a super node,
and build a super tree (Figure~\ref{example_super_tree_merge}).
Every subtree of the super tree will correspond to a maximal $\alpha$-connected component.
After postprocessing, the scalar tree will still satisfy Property 2, 3, 4, but may not satisfy Property 1,
because a super node may correspond to multiple vertices, however, this does not affect the further analysis.

Algorithm~\ref{PostScalarTree} only needs one pass of the scalar tree T, so the time complexity is $O(|V|)$,
 where $|V|$ is the number of vertices in the original scalar graph.

\begin{algorithm}[]
\fontsize{9}{10}\selectfont
\caption{Postprocess Scalar Tree}
\begin{algorithmic}[1]
\REQUIRE The scalar tree $T$ generated by Algorithm 1.
\ENSURE The super scalar tree $T_{super}$.
\STATE Create an array $ancestors$;
\STATE Assume the root of $T$ is $n_{r}$, $n_{r}.super\_node =$ new $Node()$;
\STATE $ancestors.add(n_{r})$;
\STATE $T_{super}.add(n_{r}.super\_node$);
\FOR{each node $n_{anc}$ in $ancestors$}
\STATE Create a queue Q;
\STATE Q.push($n_{anc}$);
\WHILE {!Q.empty()}
\STATE $n_{q} = Q.pop()$;
\STATE $n_{anc}.super\_node.members.add(n_{q})$;
\FOR {every child $n_{c}$ of $n_{q}$}
\IF {$n_{c}.scalar = n_{q}.scalar$}
\STATE Q.push($n_{c}$);
\ELSE
\STATE $n_{c}.super\_node =$ new $Node()$;
\STATE $ancestors.add(n_{c})$;
\STATE $T_{super}.add(n_{c}.super\_node)$;
\STATE Connect $n_{anc}.super\_node$ with $n_{c}.super\_node$;
\ENDIF
\ENDFOR
\ENDWHILE
\ENDFOR

\end{algorithmic}
\label{PostScalarTree}
\end{algorithm}

\subsection{Edge-based Scalar Graph}
 Here we describe a {\it novel} approach for modeling scalar values on edges.
We define an edge-based scalar graph $G(V, E)$ as a graph comprising edges $E$ and vertices $V$,
where each edge \emph{e} has one scalar value associated with it, denoted as \emph{e.scalar}.
Similarly, the \textbf{maximal $\alpha$-edge connected component} is defined below:
\begin{definition}
A connected component $C(V_{C}, E_{C})$ of scalar graph $G(V,E)$ is a \textbf{maximal $\alpha$-edge connected component} if it satisfies following conditions: \\
(1) for every edge $e \in E_{C}$, $e.scalar >= \alpha$.\\
(2) for any edge $e \in E_{C}$, if edge $e'$ shares a common vertex with $e$ and $e' \not\in E_{C}$, then $e'.scalar < \alpha$.\\
(3) for any edge $e(v_{1}, v_{2}) \in E_{C}$, we have $v_{1} \in V_{C}$ and $v_{2} \in V_{C}$.
\end{definition}

The edge scalar tree T of edge scalar graph G has the following properties:
\begin{enumerate}
  \item Every node in T corresponds to an edge in G with the same scalar value, and vice versa (i.e. one-to-one correspondence).
  \item Every maximal $\alpha$-edge connected component in G corresponds to a subtree in T, and vice versa (i.e. one-to-one correspondence).
  \item Assume a maximal $\alpha_{1}$-edge connected component $C_{1}$ corresponds to subtree $T_{1}$ in T, and another maximal $\alpha_{2}$-edge connected component $C_{2}$ corresponds to subtree $T_{2}$ in T, then $C_{1}$ is a subgraph of $C_{2}$ if and only if $T_{1}$ is subtree of $T_{2}$.
  \item Assume a maximal $\alpha_{1}$-edge connected component $C_{1}$ corresponds to subtree $T_{1}$ in T, and another maximal $\alpha_{2}$-edge connected component $C_{2}$ corresponds to subtree $T_{2}$ in T, then $C_{1}$ and $C_{2}$ are not connected if and only if $T_{1}$ and $T_{2}$ are not connected.
\end{enumerate}

\smallskip
\noindent \textbf{\emph{Naive Method:}} a naive way to build edge scalar tree is to
first convert the edge-based scalar graph $G(V, E)$ to a dual graph $G_{d}(V_{d}, E_{d})$--
where every edge in $G$ is converted to be a vertex in $G_{d}$,
and if two edges in $G$ share a common vertex, their correspondent vertices in $G_{d}$ are connected.
We then apply Algorithm 1 to $G_{d}$ -- the generated tree is an edge scalar tree of $G$.
The time complexity of the method is $O(|E_{d}|*log|V_{d}| + |V_{d}|*log|V_{d}|)$.
In the dual graph $G_{d}$, we have $|V_{d}| = |E|$ and $|E_{d}| = O(\sum\limits_{v \in V} degree(v)^{2})$,
so the time complexity is actually $O(\sum\limits_{v \in V} degree(v)^{2}*log|E| + |E|*log|E|)$.
The time cost is high because the bottleneck $\sum\limits_{v \in V} degree(v)^{2}$ could be $|V|^{3}$ in the worst case.

\smallskip
\noindent \textbf{\emph{Optimized Method:}} We propose a novel, more efficient
method (Algorithm~\ref{EdgeScalarTree}) to construct edge scalar tree from the edge scalar graph, and the time complexity is reduced to be $O(|E|*log|E|)$.
In line 1, we sorted all the edges in descending order of scalar value.
In line 2-3, we select the $min\_id\_edge$  on vertex v that has the minimum index.
In line 6-9, we process edge $e_{i}$.
Instead of checking all $e_{i}$'s neighbor edges (edges which share a common vertex with $e_{i}$) ,
we just need to check the $min\_id\_edge$s of $e_{i}$'s two vertices (line 6-7).
This is based on Proposition \ref{min_id_edge}:

\begin{proposition}
If edge $e_{j}$ is a neighbor edge of $e_{i}$ ($i > j$), and they share the same vertex $v$,
when processing $e_{i}$ in line 6-9 of Algorithm~\ref{EdgeScalarTree}, root($n(e_{j})$) is the same as root($n(v.min\_id\_edge)$).
\label{min_id_edge}
\end{proposition}

\begin{proof}
Since $i > j$, $e_{j}$ will be processed before $e_{i}$.
When processing $e_{j}$ in line 6-9 of Algorithm~\ref{EdgeScalarTree},
$n(e_{j})$ will be connected to root($n(v.min\_id\_edge)$),
which means $n(e_{j})$ and $n(v.min\_id\_edge)$ will be in the same subtree thereafter.
So when processing $e_{i}$, root($n(e_{j})$) is the same as root($n(v.min\_id\_edge)$).
\end{proof}

\begin{algorithm}[]
\fontsize{9}{10}\selectfont
\caption{Constructing Edge Scalar Tree}
\begin{algorithmic}[1]
\REQUIRE An edge scalar graph G(V, E).
\ENSURE The edge scalar tree T of G.
\STATE Sort edges in decreasing order of scalar values, the sorted edges are $e_{1}, e_{2}, ... e_{n}$.
\FOR{each vertex v in G}
\STATE set $v.min\_id\_edge$ to be the edge on v that has the minimum index.
\ENDFOR
\STATE Create a tree node $n(e)$ for each edge $e$.
\FOR{$i=1 $ to $n$}
\STATE Assume $e_{i}$ has two vertices $v_{1}$ and $v_{2}$, create an array: \\$min\_neighbors[2]=\{v_{1}.min\_id\_edge, v_{2}.min\_id\_edge\}$;
\FOR{each edge $e_{m}$ in $min\_neighbors$}
\IF {$m < i$ and currently $n(e_{i})$ and $n(e_{m})$ are not in the same subtree}
\STATE connect $n(e_{i})$ to root($n(e_{m})$) //$n(e_{i})$ is parent
\ENDIF
\ENDFOR
\ENDFOR
\end{algorithmic}
\label{EdgeScalarTree}
\end{algorithm}

The time complexity of line 1 in Algorithm~\ref{EdgeScalarTree} is $O(|E|*log|E|)$.
For each edge e, line 8 is executed $O(1)$ times,
so line 8 is executed a total of $O(|E|)$ times, and the total running time of line 5-9 is $O(|E|*log|E|$.
The worst case running time of Algorithm~\ref{EdgeScalarTree} is $O(|E|*log|E|)$.
Similar to the case described in the previous section,
if some edges have the same scalar value, we can
use Algorithm~\ref{PostScalarTree} to postprocess the edge scalar tree.

\subsection{Relationship between maximal $\alpha$-(edge) connected component and Dense Subgraph}
\label{dense_def}
A dense subgraph is a connected subgraph in which every vertex is heavily connected to other vertices in the subgraph.
K-Core~\cite{Vlad:kcore} and K-Truss~\cite{james:truss, ktruss_cohen}(also called Triangle K-Core in~\cite{tkcore12}, DN-graph in~\cite{wang:dngraph} ) are two common dense subgraph patterns that draw much attention in recent works.
The definitions of K-Core and K-Truss are as follows:
\begin{definition}
A K-Core is a subgraph in which each vertex participates in at least K edges within the subgraph.
The maximal K-Core of a vertex v is the K-Core containing v that has the maximum K value,
and the K value of maximal K-Core of v is denoted as KC(v).
\end{definition}

\begin{definition}
A K-Truss is a subgraph in which each edge participates in at least K triangles within the subgraph.
The maximal K-Truss of an edge e is the K-Truss containing e that has the maximum K value,
and the K value of maximal K-Truss of e is denoted as KT(e).
\end{definition}
Now we prove the relationship between maximal $\alpha$-(edge) connected component and K-Core, K-Truss below.
\begin{proposition}
\label{relation_kcore}
In a scalar graph G where for any vertex \emph{v}, \emph{v.scalar = KC(v)},
a maximal $\alpha$-connected component in G is a K-Core where $K = \alpha$.
\end{proposition}

\begin{proof}
Assume in a maximal $\alpha$-connected component $C$, vertex v has the minimum scalar value.
Based on the definition of K-Core, for every vertex \emph{v'} in the maximal K-Core of \emph{v}, $KC(v') >= KC(v)$,
so the maximal K-Core of \emph{v} is a subgraph of \emph{C}.
Since \emph{v} is connected to at least $KC(v)$ vertices in its maximal K-Core,
\emph{v} is connected to at least $KC(v)$ vertices in \emph{C}.
For every other vertex \emph{v'} in \emph{C},
similarly we can get that \emph{v'} is also connected to at least $KC(v)$ vertices in \emph{C}.
So \emph{C} is a K-Core where $K = KC(v)$,
since $KC(v) = v.scalar >= \alpha$,
C is also a K-Core where $K = \alpha$.
\end{proof}

\begin{proposition}
In an edge scalar graph G where for any edge \emph{e}, $e.scalar = KT(e)$,
a maximal $\alpha$-edge connected component in G is a K-Truss where $K = \alpha$.
\end{proposition}
The proof is similar to the proof of Proposition~\ref{relation_kcore} and is omitted
in the interests of space.
Note that when we define the scalar value of each vertex/edge to be KC(v)/KT(e),
the (edge) scalar tree will capture the distribution and relationships among K-Cores or
K-Trusses in the graph.
\label{relation_dense}

\subsection{Visualization via Terrain Metaphor}
Scalar trees are usually not easy to visually interpret,
especially when the size of the tree is too large.
We adapt the terrain metaphor -- topological landscape
visualization technique defined on scalar-valued functions~\cite{terrain10} --
to visualize scalar graphs.

In Figure~\ref{Scalar3D} we use an example to illustrate
how to convert the scalar tree in Figure~\ref{example_tree2} to terrain visualization in Figure~\ref{example_3dvis}.
In Figure~\ref{example_layout} we first layout all the tree nodes of Figure~\ref{example_tree2} in a 2D plane,
every node $n_{i}$ is represented by a boundary $b_{i}$ in the 2D plane,
and the area enclosed by the boundary $b_{i}$ is proportional to the number of nodes in subtree (not including $n_{i}$) rooted at $n_{i}$.
To generate the 2D layout, we start traversing the tree from the root(bottom) node $n_{9}$,
draw the outermost boundary $b_{9}$ to represent it.
Then we move to $n_{8}$, and draw a boundary $b_{8}$ inside $b_{9}$.
When we reach $n_{7}$,
and draw boundary $b_{7}$, we find there are two subtrees rooted at $n_{7}$,
so we split the area inside $b_{7}$ into 2 areas,
and recursively layout each subtree in each area.
When we reach leaf nodes $n_{1}, n_{2}, n_{4}$, since the size of their subtrees is 0,
their correspondent boundaries degenerate to be points.

To convert the 2D layout (Figure~\ref{example_layout}) into a terrain visualization in 3D space (Figure~\ref{example_3dvis}),
we first escalate each boundary $b_{i}$ in Figure~\ref{example_layout} to the height of $n_{i}.scalar$,
and then draw a ``wall'' between neighboring boundaries.
Finally we generate a terrain in Figure~\ref{example_3dvis}.
We can color the terrain by assigning each vertex a color value, and since each ``wall'' is confined by two boundaries $b_{i}$ and $b_{j}$,
we color the wall based on the color value of the vertex corresponding to $b_{i}$ or $b_{j}$.
\footnote{The use of 3D rather than 2D was a conscious one.
First we found that the 3D abstraction
better matched users' mental-map of the ``terrain'' concept as well as the
hierarchical relationships amongst components-of-interest.
While 3D-layouts potentially
pose a problem with respect to occlusion, the ability to interactively rotate the
point-of-view along with the ability to link 2D-layouts of regions-of-interest as discussed below, alleviates this issue. }

To identify a subtree of the scalar tree in the terrain visualization,
we locate the correspondent boundary $b_{r}$ of the subtree root $n_{r}$,
and the terrain area within the boundary $b_{r}$ corresponds to the subtree.
In our paper,
we define a \textbf{$peak_{\alpha}$} in terrain as below:
\begin{definition}
A $peak_{\alpha}$ is the terrain area within a boundary whose height is $\alpha$.
\end{definition}
Since each $peak_{\alpha}$ corresponds to a subtree in scalar tree,
we can easily get that every $peak_{\alpha}$ corresponds to a maximal $\alpha$-connected component.
Also, $peak_{\alpha}$s preserve the containment/connection relationship of maximal $\alpha$-connected components.
For example, the red peak in Figure~\ref{example_cc2} is a $peak_{5}$,
which corresponds to the maximal 5-connected component (red nodes) in Figure~\ref{example_graph_cc2},
and the red peak in Figure~\ref{example_cc4} is a $peak_{3}$,
which corresponds to the maximal 3-connected component (red nodes) in Figure~\ref{example_graph_cc4}.
One $peak_{\alpha}$ may contain some sub-peaks,
 which indicates its maximal $\alpha$-connected component contains other maximal $\alpha'$-connected components.
For example, $peak_{5}$ in Figure~\ref{example_cc2} is contained in $peak_{3}$ in Figure~\ref{example_cc4},
this indicates that the correspondent maximal 5-connected component in Figure~\ref{example_graph_cc2} is a subgraph
of the maximal 3-connected component in Figure~\ref{example_graph_cc4}.
In a $peak_{\alpha}$, the area of its bottom boundary indicates the number of vertices in its correspondent maximal $\alpha$-connected component.

To get all the maximal $\alpha$-connected components for a particular $\alpha$ value,
we can use a 2D plane with height = $\alpha$ to cross the terrain in 3D space,
and all the $peak_{\alpha}$s above the plane correspond to all maximal $\alpha$-connected components.
\emph{The benefit of using terrain visualization is,
it captures the overall information of all maximal $\alpha$-connected components in one picture.
Also, we could encode more information in the terrain by using colors to the terrain.}


\begin{figure}[!h]
\centering
\subfigure[Scalar Tree]
{
\hspace*{-0.40cm}\includegraphics[width=0.115\textwidth]{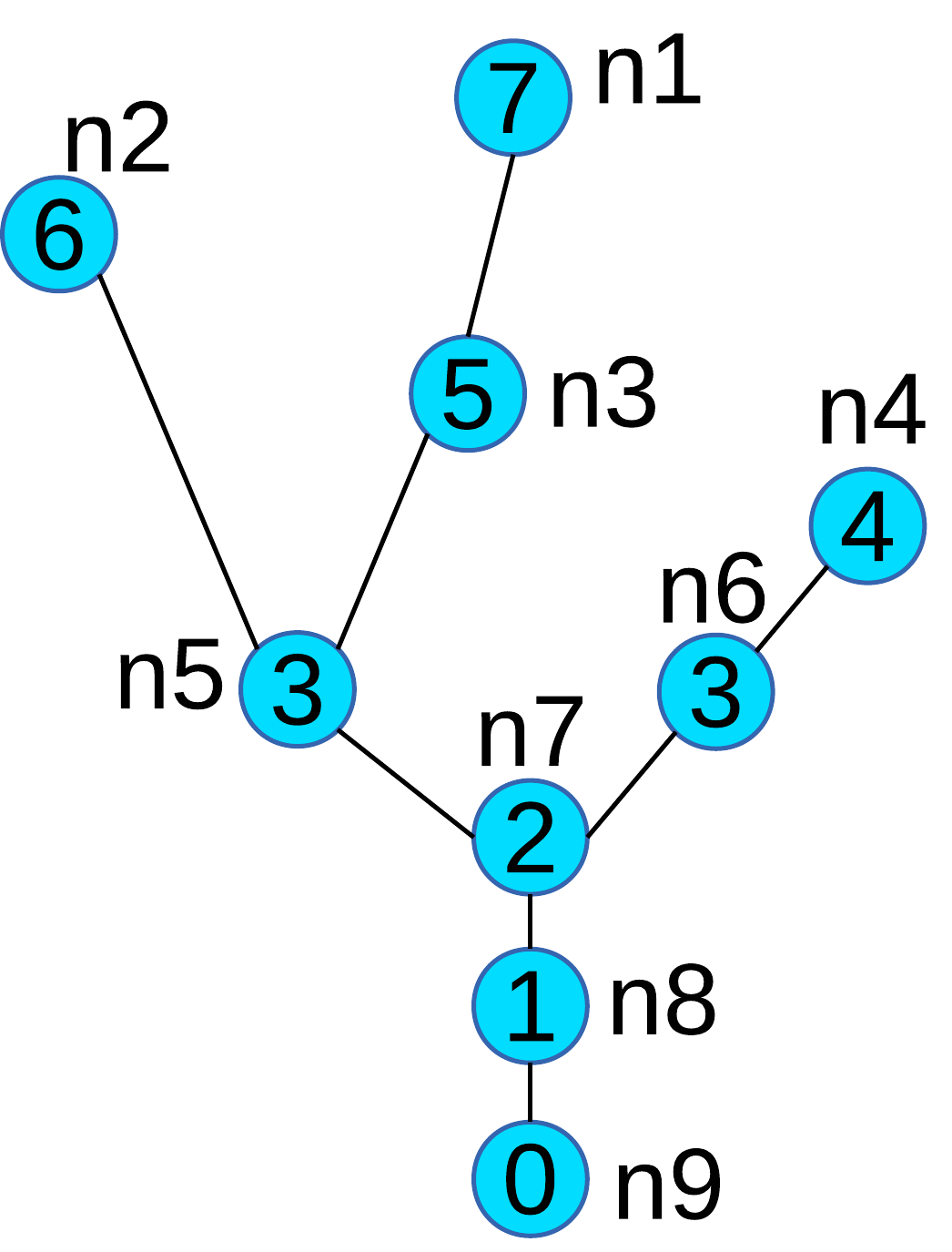}
\label{example_tree2}
}
 \quad
\subfigure[2D Nodes Layout]
{
\hspace*{-0.20cm}\includegraphics[width=0.20\textwidth]{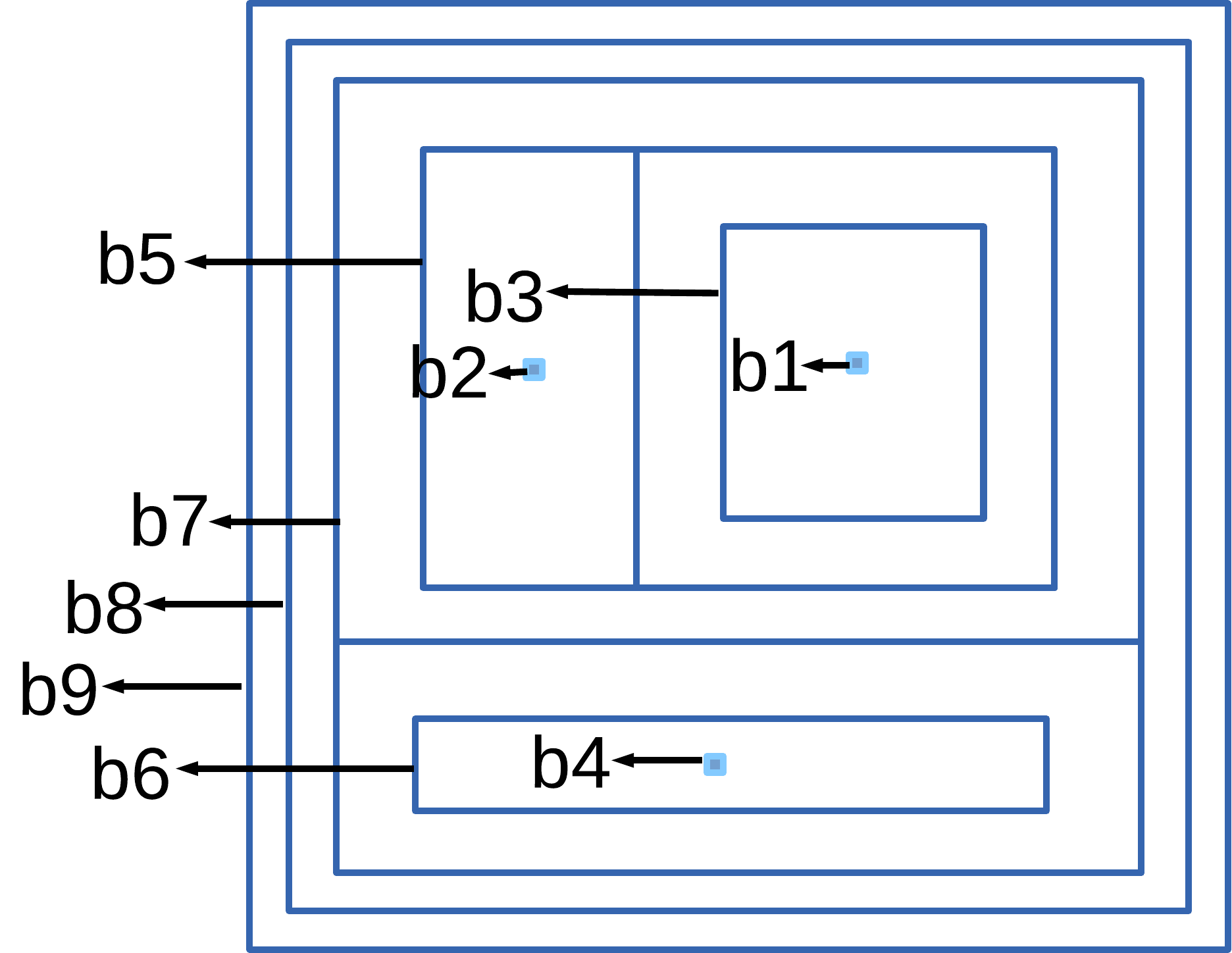}
\label{example_layout}
}
 \quad
\subfigure[3D Terrain]
{
\hspace*{-0.0cm}\includegraphics[width=0.10\textwidth]{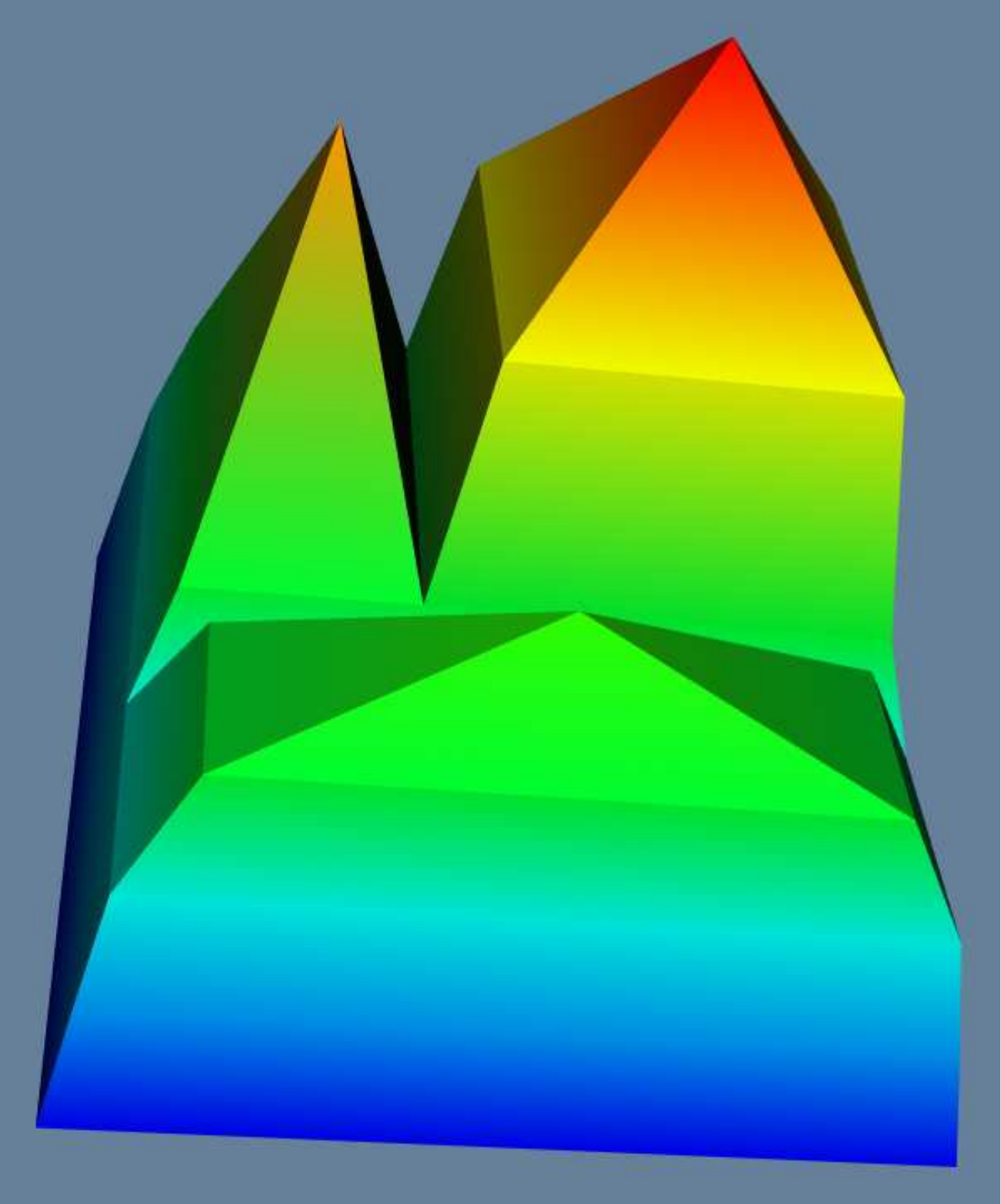}
\label{example_3dvis}
}

\subfigure[Maximal 5-Connected Component (red)]
{
\hspace*{-0.0cm}\includegraphics[width=0.15\textwidth]{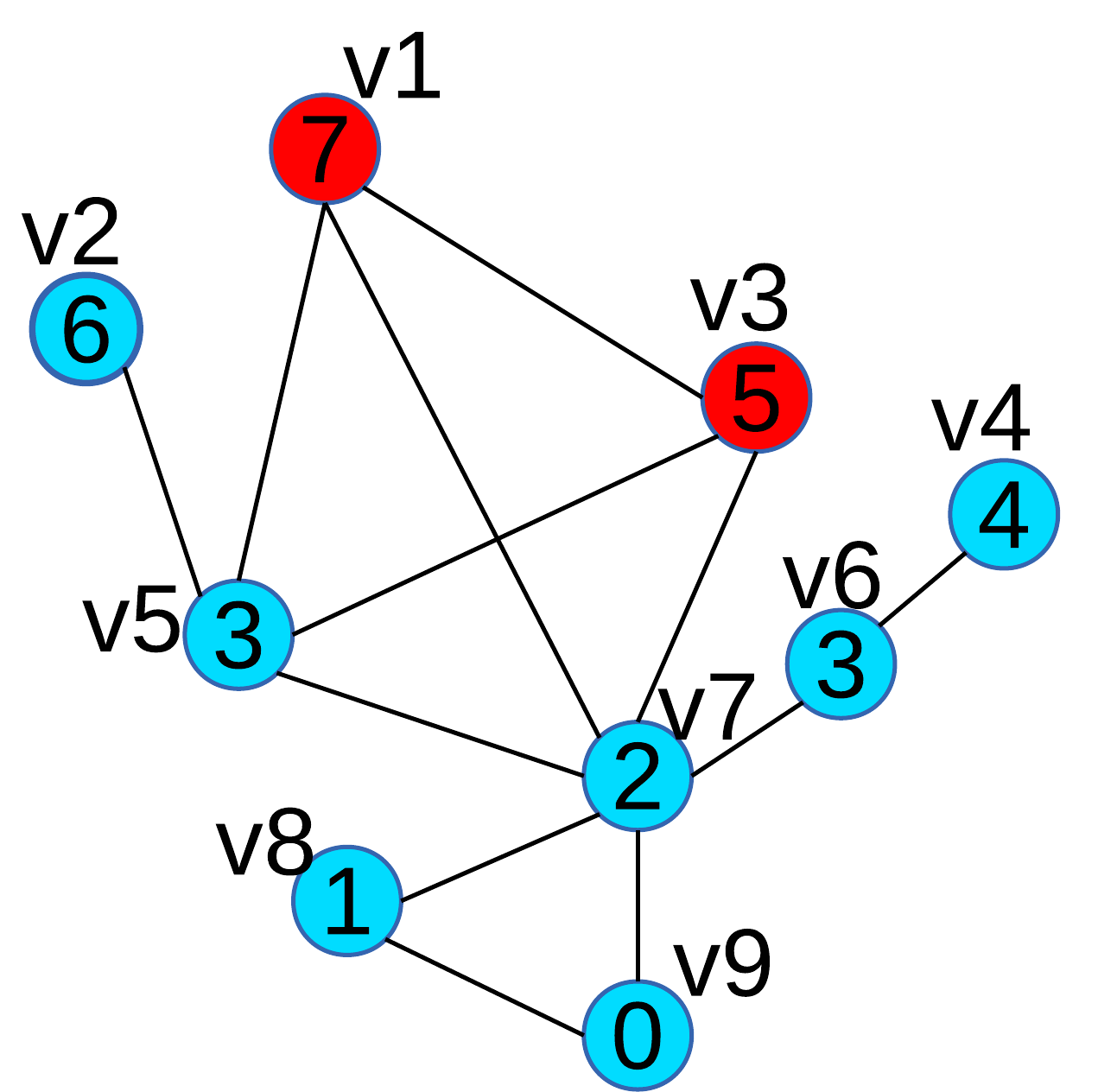}
\label{example_graph_cc2}
}
 \quad
\subfigure[Correspondent Subtree (red)]
{
\hspace*{-0.0cm}\includegraphics[width=0.12\textwidth]{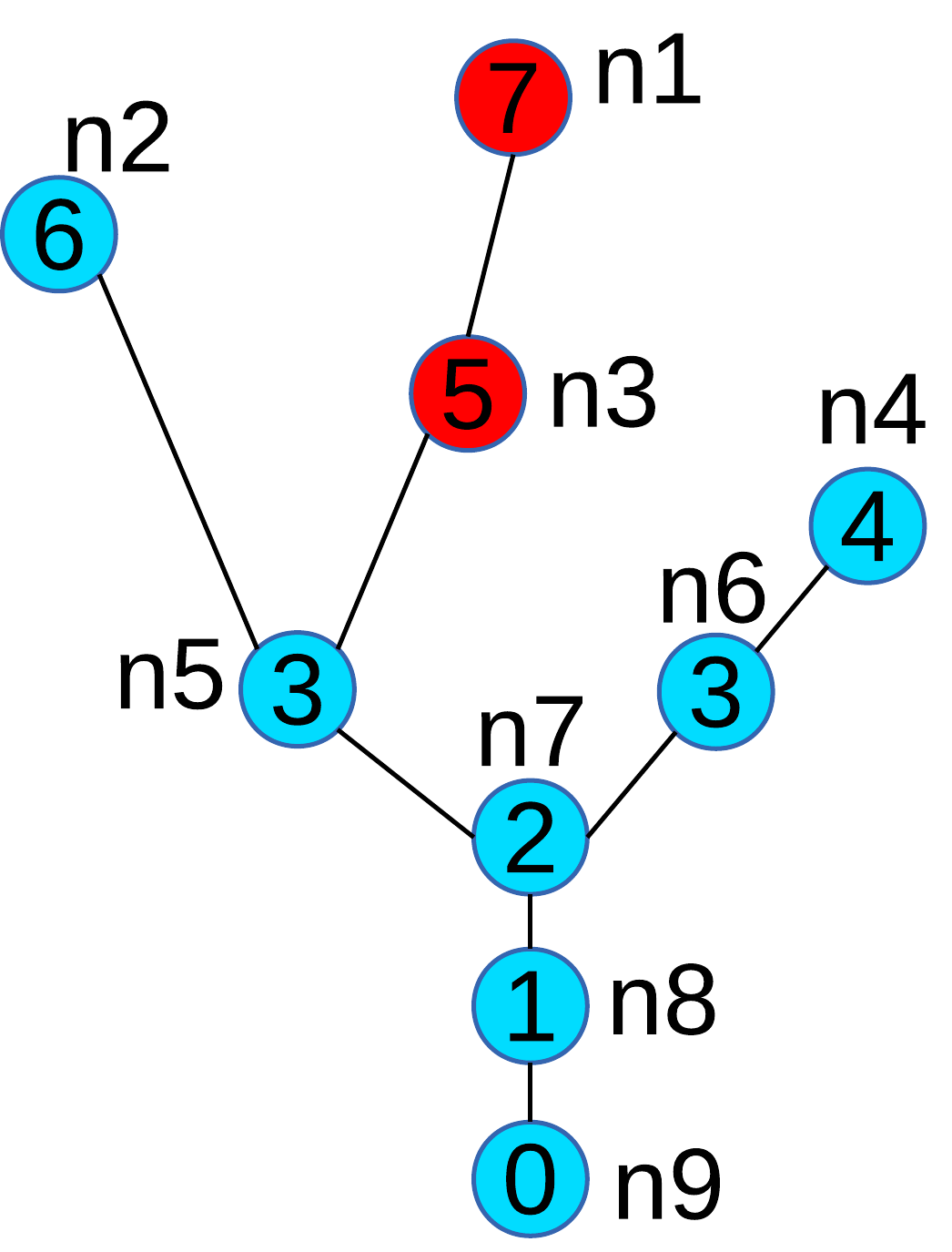}
\label{example_tree_cc2}
}
 \quad
\subfigure[$Peak_{5}$ (red)]
{
\hspace*{-0.0cm}\includegraphics[width=0.11\textwidth]{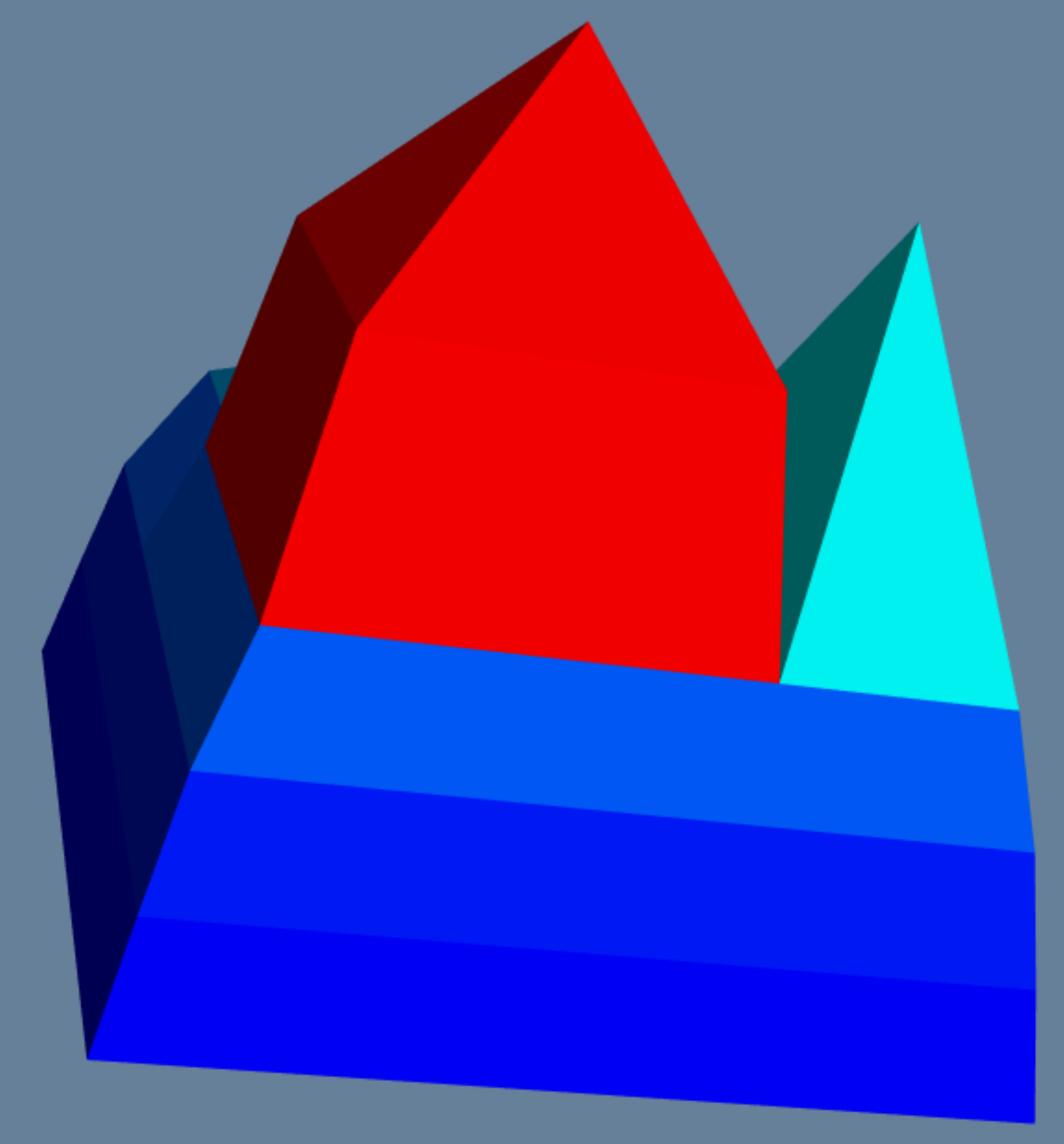}
\label{example_cc2}
}

\subfigure[Maximal 3-Connected Component (red)]
{
\hspace*{-0.0cm}\includegraphics[width=0.15\textwidth]{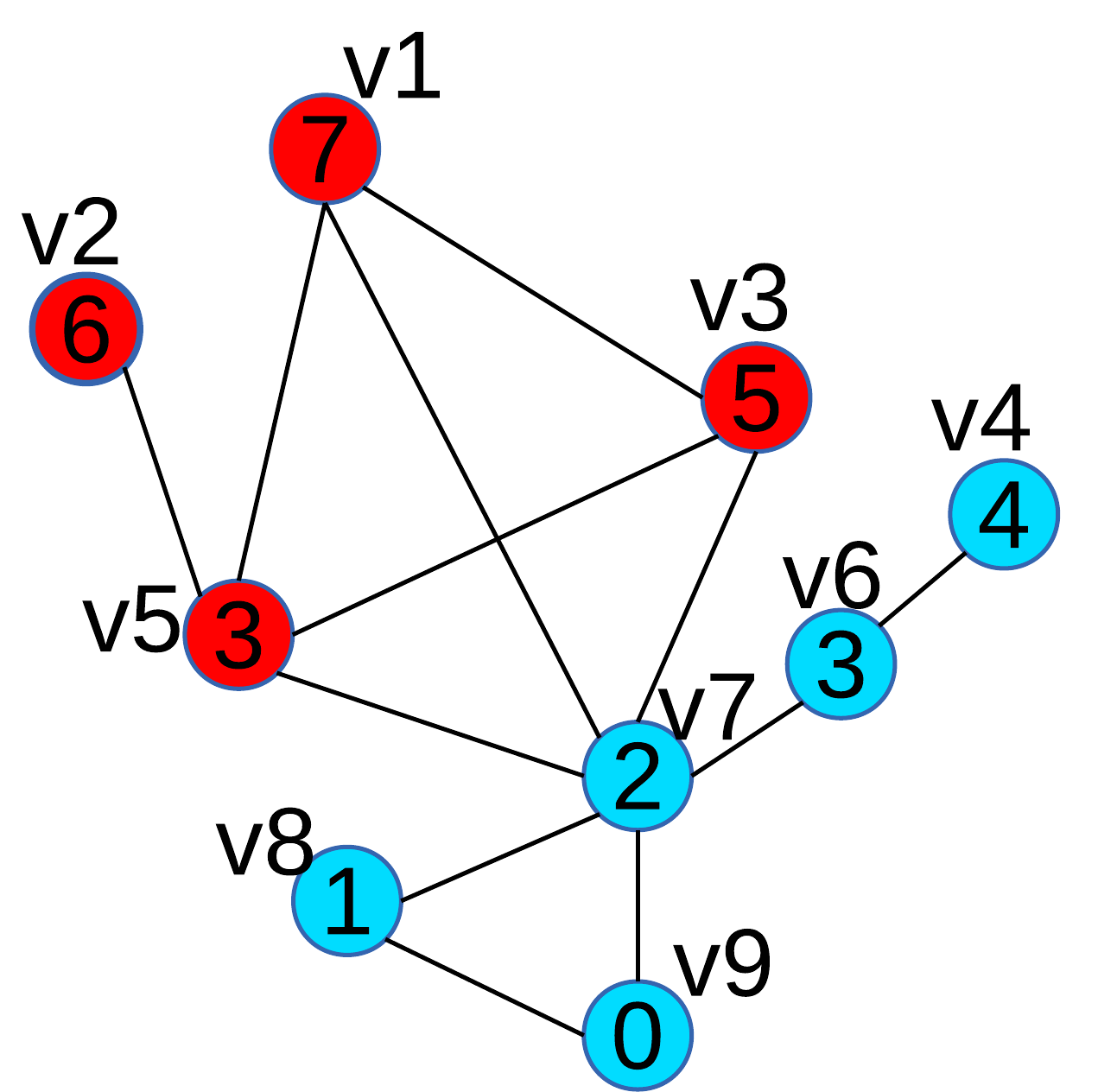}
\label{example_graph_cc4}
}
 \quad
\subfigure[Correspondent Subtree (red)]
{
\hspace*{-0.0cm}\includegraphics[width=0.12\textwidth]{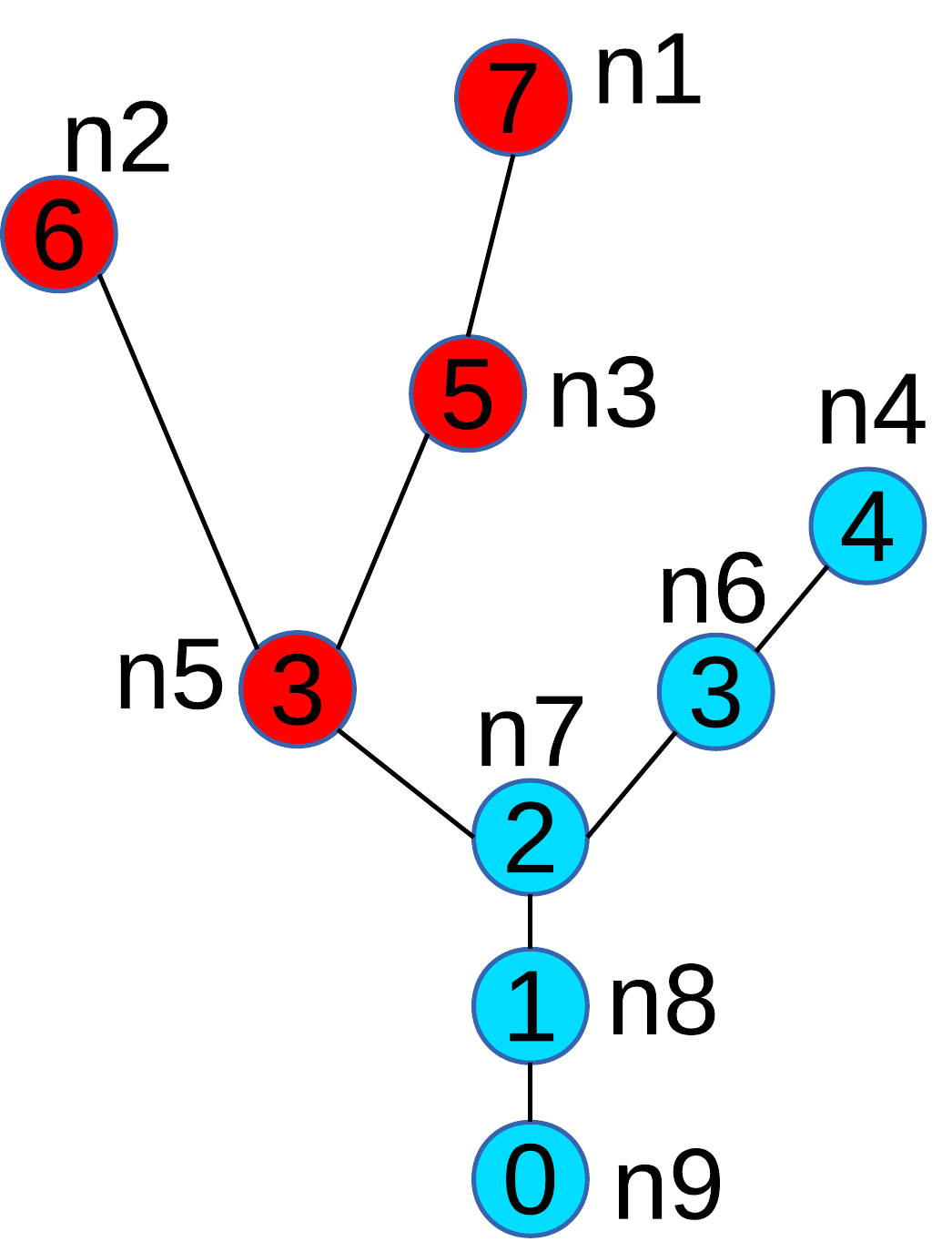}
\label{example_tree_cc4}
}
 \quad
\subfigure[$Peak_{3}$ (red)]
{
\hspace*{-0.0cm}\includegraphics[width=0.11\textwidth]{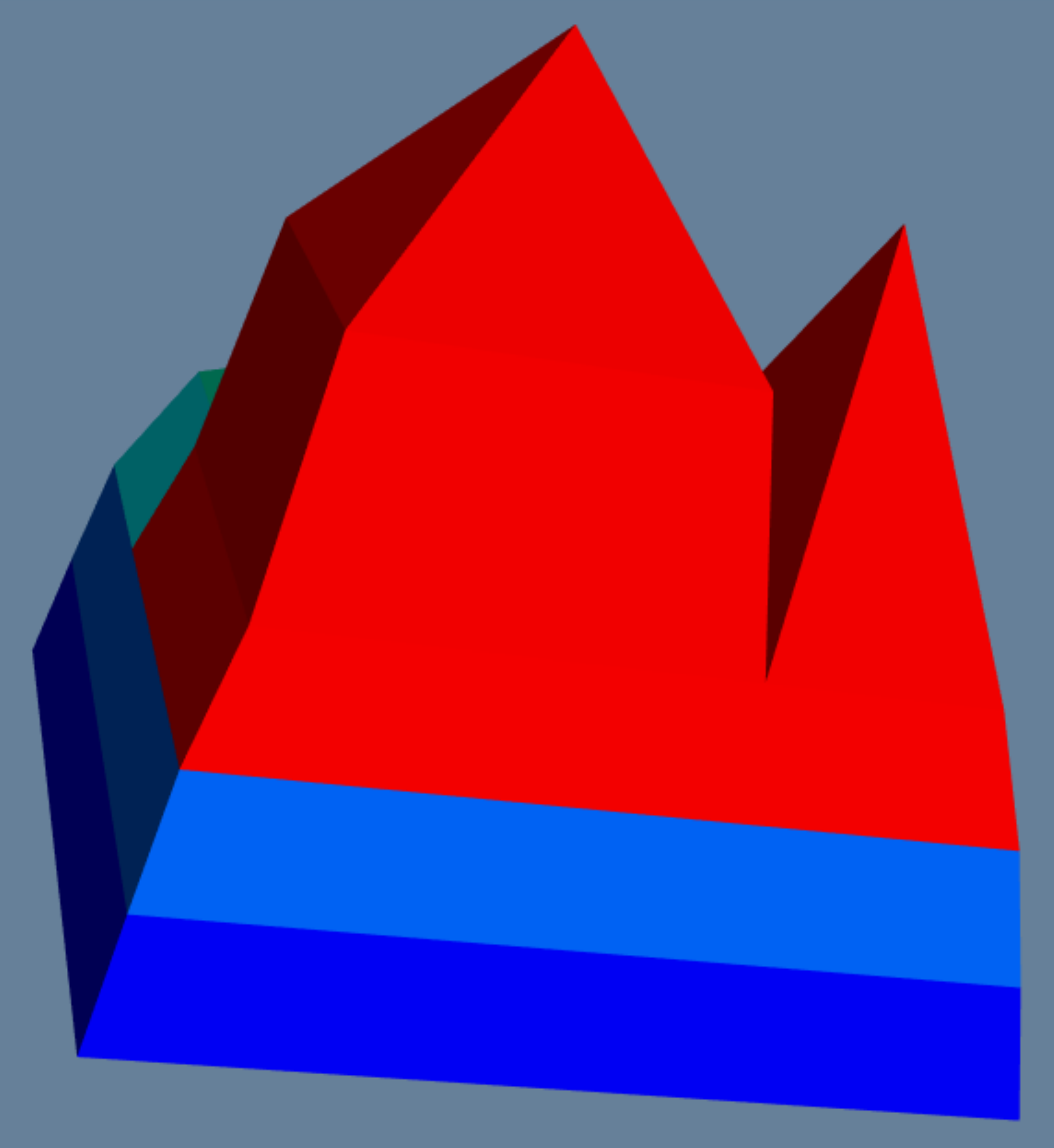}
\label{example_cc4}
}
\caption[Optional caption for list of figures]{Terrain Visualization of a Simple Scalar Tree}
\label{Scalar3D}
\end{figure}

\smallskip
\noindent
\textbf{\bf{User Interaction:}}
Our terrain
visualization tool provides following features to help users interact with the terrain.

\smallskip
\noindent \textbf{\emph{Rotate:}} the user could rotate the terrain to look at it from different angles. For example, Figure~\ref{example_3dvis} and Figure~\ref{example_cc2} show the same terrain from two different viewpoints.

\smallskip
\noindent \textbf{\emph{Zoom in/out:}} the user can zoom in/out to see the details/overview of the terrain. For example, in Figure~\ref{dblp4areacomm1allcircle}, we zoom into the terrain in the left picture, and get a clear picture of the two peaks in the right picture.

\smallskip
\noindent \emph{\textbf{Simplification: }} When visualizing a scalar tree with too many nodes,
the rendering and interaction speed might be slow,
we simplify the tree to make the visualization faster as follows.
We discretize the scalar values,
so similar scalar values will be approximately represented by the same value,
 and then we can use Algorithm~\ref{PostScalarTree} to build an approximate super tree with far
fewer tree nodes.

\smallskip
\noindent \textbf{\emph{Linked-2D-Displays:}} Our tool allows the user to select any region of the terrain, and invoke a ``callback'' function to visualize the selected region using other visualization method. For example, in Figure~\ref{grqc_terrain}, we select the region in the white dashed line box,
and 2D-linked
 spring layout visualization method to draw the selected region in the red box beside it.
Expert users optionally can use the ``callback'' function to integrate our terrain visualization
with other customized visualization methods.
We can also link a 2D treemap of the scalar graph by setting the height of all boundaries
to 0 and (optionally) using colors -- red/yellow/green/blue -- to indicate highest/high/low/lowest value --
so the red/yellow blocks in the 2D treemap indicate the subgraph areas with high scalar values
(see Figure~\ref{GrQc_2d} and Figure~\ref{GrQc_3d}).
The 2D visualization clearly shows how the red/yellow/green blocks are distributed over the graph
while in the  3D visualization the user will need to rotate her point-of-view to get a similar map.
That said, using color to encode scalar value has a drawback.
For example, peak 1 and peak 2 in Figure~\ref{GrQc_3d} correspond to block 1 and block 2 in Figure~\ref{GrQc_2d},
from the height we can see that the peak 2 is a little higher than the peak 1,
but from the color we cannot tell the difference between block 1 and block 2.
Finally, note that using height allows us to color the terrain based on a different attribute
(from the one that generates the terrain),
which helps to get a picture of the correlation across multiple scalars, as we shall discuss next.
\vspace{-0.10in}
\begin{figure}[h]
\centering
\subfigure[GrQc 2D visualization]
{
\includegraphics[width=0.20\textwidth]{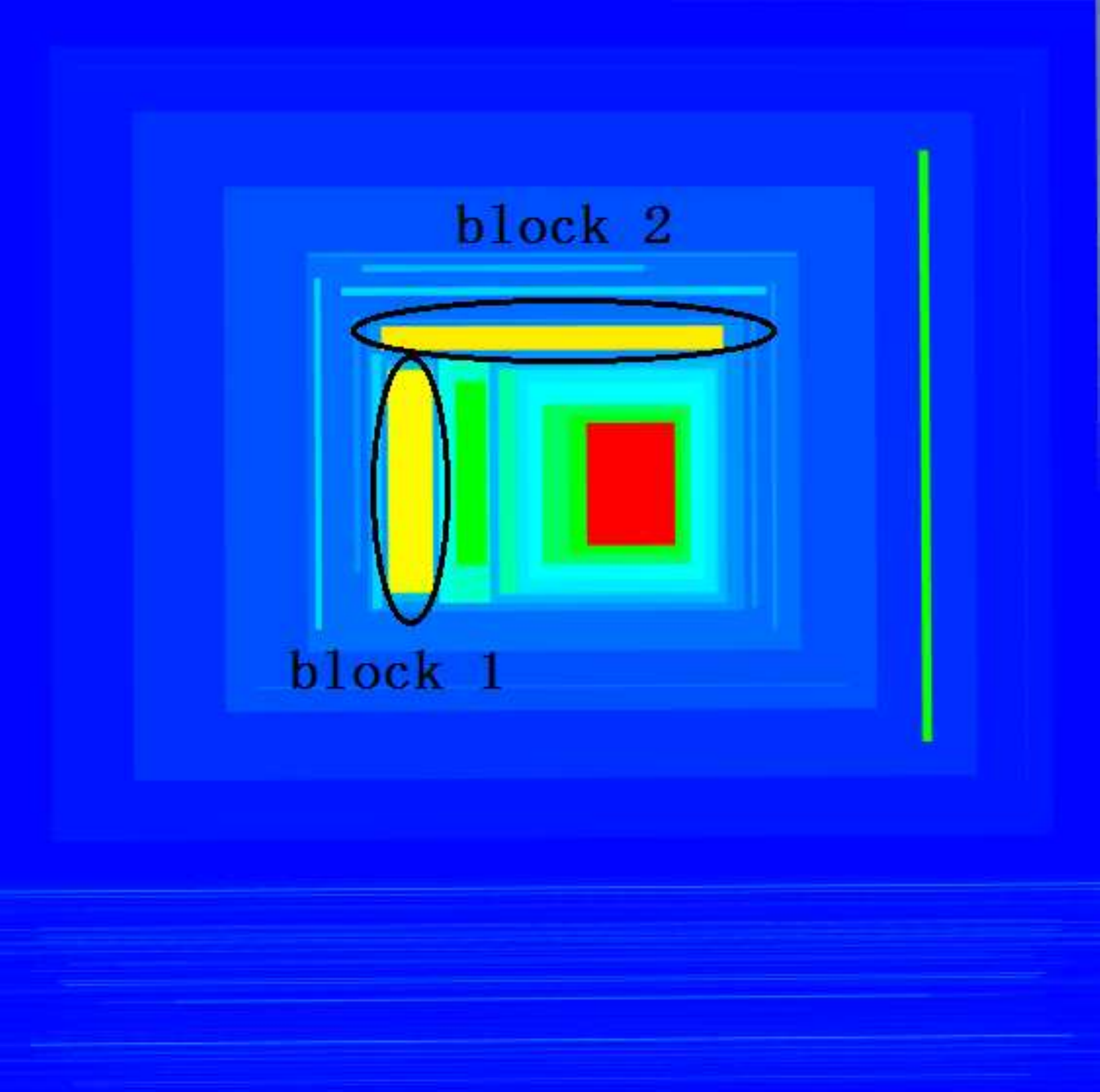}
\label{GrQc_2d}
}
 \quad
\subfigure[GrQc 3D visualization]
{
\includegraphics[width=0.20\textwidth]{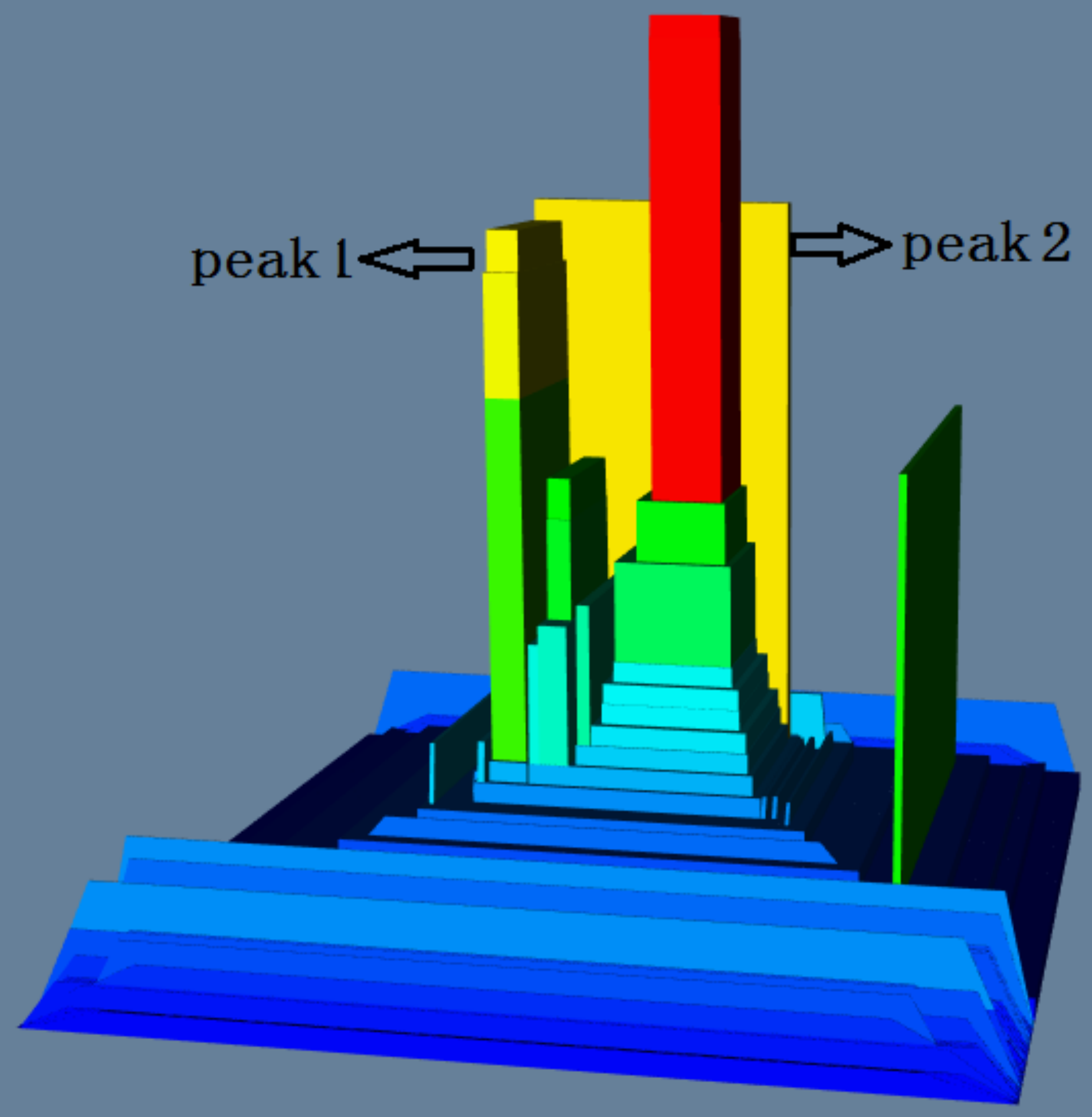}
\label{GrQc_3d}
}

\vspace{-0.15in}
\caption[Optional caption for list of figures]{2D Treemap vs. 3D Terrain}
\label{2dvs3d}
\end{figure}

%% file: MultipleScalar.tex
\subsection{Handling Multiple Scalar Fields}
On some graphs, multiple scalar fields can be defined, which means there are multiple scalar values defined on each vertex.
For example, a vertex v has degree value and KC(v) value.
Users might be interested in how the multiple scalar fields correlate with each other -- are their changing trends the same or the opposite on the graph?
In this section,
we propose two indexes to measure the correlation of two scalar fields on a graph.
Then we can use terrain visualization to analyze the relationship between two scalar fields.

\smallskip
\noindent \textbf{\emph{Local Correlation Index:}}
Assume we have two scalar fields, $S_{i}$ and $S_{j}$, defined on a graph.
Each vertex v has scalar values $v.scalar_{i}$ and $v.scalar_{j}$ in the two scalar fields.
Some previous work ~\cite{multi_field} proposed a measure to compute the correlation of multiple scalar fields in continuous domain,
we adapt their method and propose Local Correlation Index to measure the correlation of two scalar fields on local areas of a graph.
Here local area is defined as k-hop neighborhood of each vertex v(denoted as N(v)),
for all experiments we limit this to be 1-hop.
The Local Correlation Index of $S_{i}$ and $S_{j}$ on N(v) is denoted as $LCI_{Si, Sj}(v)$ .
For each vertex v, we compute $LCI_{Si, Sj}(v)$  as follows.

{\small
\begin{flalign*}
&\overline{v.scalar_{i}} = \frac{\sum\limits_{u\in N(v)}u.scalar_{i}}{|N(v)|} & \\
&Cov_{ij}(v) = \frac{\sum\limits_{u\in N(v)} (u.scalar_{i} - \overline{v.scalar_{i}}) * (u.scalar_{j} - \overline{v.scalar_{j}})}{|N(v)|} & \\
&LCI_{Si, Sj}(v) = \frac{Cov_{ij}(v)}{\sqrt{Cov_{ii}(v)} * \sqrt{Cov_{jj}(v)}}&
\end{flalign*}
}

$LCI_{Si, Sj}(v)$ is actually the correlation of the scalar values of $S_{i}$ and $S_{j}$ on v's k-hop neighborhood.
A positive/negative $LCI_{Si, Sj}(v)$ indicates the changing trends of $S_{i}$ and $S_{j}$ are consistent/inconsistent on v's k-hop neighborhood.
This method can easily be adapted to analyze edge-based scalar graphs.

\smallskip
\noindent \textbf{\emph{Global Correlation Index:}}
We can compute the Global Correlation Index (GCI) of scalar fields $S_{i}$ and $S_{j}$ on a graph by averaging the Local Correlation Indexes of all neighborhoods.

\vspace{-0.10in}
\begin{flalign*}
GCI_{S_{i}, S_{j}}(G) = \sum\limits_{v\in V} LCI_{S_{i}, S_{j}}(v)/|V|
\end{flalign*}
By comparing the Global Correlation Index and Local Correlation Index,
we may identify some outlier neighborhood on which the correlation of scalar fields $S_{i}$ and $S_{j}$ is different from the overall correlation.

\smallskip
\noindent \textbf{\emph{Terrain Visualization:}}
To visualize the local correlation between two scalar fields,
we can use $LCI_{Si, Sj}(v)$ as a scalar field to draw the terrain.
This will show us the overall distribution of $LCI_{Si, Sj}(v)$ over the graph,
and help us identify the area of the graph where the two scalar fields are positively/negatively correlated.

We can also visually capture the global correlation of scalar fields $S_{i}$ and $S_{j}$ through coloring terrain visualization.
We use one scalar field $S_{i}$ to draw the terrain visualization,
and use the other scalar field $S_{j}$ to color the terrain (see Figure~\ref{GrQc_3d_intro}).
Please note that this method can also
 visualize the relationship between a numerical attribute and a nominal attribute of vertices,
by coloring the terrain based on the value of nominal attributes. 

\input{Discussion}

%% file: Discussion.tex
\subsection{Related Context}

We are now in a position to briefly place the proposed visualization strategy
in the context of related work.
Visualizing graph data is an important problem (See ~\cite{L12} for a recent survey -- due to interests of space here we focus only on the most relevant).
Gronemann et al.~\cite{L1} (similarly Athenstadt et al.~\cite{L4}) use topographic maps
to visualize clustering structure within a graph --
each mountain corresponds to a cluster.
van Liere et al.\cite{L10} propose the GraphSplatting
method to visualize a graph as a 2D splat field.
Telea et al.\cite{L9} generate a concise representation of graph by clustering edges and bundling similar edges together and subsequently visualize the graph.
While effective on small scale datasets for displaying overall cluster structure, these methods simply
do not scale to large data with millions of edges nor do they account for attributed graphs (scalar
values).
Bezerianos et al.~\cite{L11} and Wattenberg~\cite{wattenberg2006visual} propose interactive visual system to let users explore networks with multiple node or edge attributes. However, these methods do not consider
hierarchical relationships among components-of-interest and graph attributes simultaneously.

Sariyuce et al.~\cite{dense_hierarchy} proposed (r, s)-nucleus to denote a dense subgraph comprised by cliques, and use forest of nucleus to represent hierarchical structure of a graph. The difference is, their definition of (r, s)-nucleus focuses on density of subgraph, while our definition of \textbf{\emph{maximal $\alpha$-connected components}} focuses on relation between scalar values and graph topology, and not just limited to nucleus motifs. Moreover, their effort does not consider visualization as an objective -- a primary focus of our effort.
Some visualization methods (such as LaNet-vi~\cite{kcore_vis}) are proposed to visualize K-Cores, but they are not general enough to handle other vertex/edge attributes. Martin et al.~\cite{open_ord} proposed OpenOrd to visualize large scale graphs in multi-level way, but do not effectively highlight components-of-interest.
Both LaNet-vi and OpenOrd share some of our objectives w.r.t network visualization, and we compare and contrast with these efforts in Section~\ref{user_sec}.

In summary, a major difference between our work and prior art is that
we propose to analyze the graph through the hierarchical structure (scalar tree) induced by the \textbf{\emph{maximal $\alpha$-connected components}}.
The benefit is that it {\it naturally} encodes the relationship between clustering structure and
scalar values -- it highlights how scalar values evolve from high values to low values over the graph.
This is particularly useful for a data scientist who wishes to understand
how a community is expanded from its core members to peripheral
members(see Figure~\ref{dblp4area_communities})\cite{Bajaj97thecontour, 1263272, 769927}.
Furthermore, 
based on different attributes,
the \textbf{\emph{maximal $\alpha$-connected component}} can represent different subgraph patterns, such as K-core, K-truss, subcommunity,
which has attracted much interest within the database community~\cite{csv,Vlad:kcore, james:truss},
to reveal the topological relationship (containment, connection) among components of interest
(e.g. K-Cores, K-Trusses, communities). We examine these issues next.

%% file: Evaluation.tex
\section{Experimental Evaluation}
\label{evaluation}
We seek to evaluate the effectiveness (qualitative) and efficiency
of the our interactive network visualization method in this section.
We leverage a wide range of datasets from
the network science community (some are from ~\cite{snapnets}) as noted in Table~\ref{datasets}.
Heights in our terrain visualization represent
scalar measures of input scalar graph while color represents intensity of the same measure (unless otherwise
noted). The color ranges from red (most intense);
yellow (intense); green (less intense); blue (least intense).
All experiments are evaluated on a 3.4GHz CPU, 16G RAM Linux-based desktop.

\begin{table}[!h]
\fontsize{8}{10}\selectfont
\centering
\caption{Dataset Properties}
\begin{tabular}{|c|c|c|c|}\hline
  Dataset & \# Nodes & \# Edges & Context \\\hline
  GrQc   &  5242     & 14496  & \pbox{16cm}{Coauthorship in General Relativity\\ and Quantum Cosmology}\\\hline
  Wikivote & 7115    & 103689 & \pbox{16cm} {Who-votes-on-whom relationship\\ between Wikipedia users} \\\hline
  Wikipedia & 1,815,914 & 34,022,831 & \pbox{16cm} {Links between Wikipedia pages} \\\hline
  PPI & 4741 & 15147 & \pbox{20cm} {Protein Protein Interaction network} \\\hline
  Cit-Patent & 3,774,768 & 16,518,947 & \pbox{16cm} { Citations made by patents granted\\ between 1975 and 1999} \\\hline
  Amazon &  334863 & 925872  & \pbox{16cm} {Co-Purchase relationship\\ between products in Amazon }\\\hline
  Astro & 17903 & 196972  & \pbox{16cm}{Coauthorship between authors in \\
                                            Astro Physics} \\\hline
  DBLP & 27199 & 66832 & \pbox{16cm}{Coauthorship between authors in\\
                                   (Database, Data Mining, Machine\\ Learning, Information Retrieval)} \\\hline

\end{tabular}
\label{datasets}
\end{table}

\subsection{Visualizing Dense Subgraphs}

\noindent{\bf Effectiveness:}
The visualization of dense subgraphs within graphs has been of much interest within the database and information visualization.
Examples abound and include CSV plots~\cite{csv}, K-Core~\cite{kcore_vis} and Triangle K-Core (K-Truss) ~\cite{tkcore12} plots.
Here we use our terrain visualization to visualize K-Cores and K-Trusses and compare with the previous methods. 

We consider two datasets (GrQc, Wikivote) for this illustration.
In Figure~\ref{grqc_spring} and Figure~\ref{wikivote_spring}, we use the traditional spring layout Algorithm~\cite{graphdraw} to draw both networks --
it is hard to say anything about the distribution of dense subgraphs using such a plot.
Following the discussion in Section~\ref{relation_dense},
we use KC(v) as scalar value,
and generate the terrain visualization of both networks in
Figure~\ref{grqc_terrain} and Figure~\ref{wikivote_terrain}.
Recall that if the scalar value of every vertex v is defined to be KC(v),
each maximal $\alpha$-connected component is a K-Core where $K = \alpha$.
Thus in the terrain, each $peak_{\alpha}$ is a K-Core where $K = \alpha$.
The distribution of K-Cores in the two datasets is obviously different.
Figure~\ref{wikivote_terrain} shows that there is one single high peak,
which means the network has one densest K-Core,
and K-Core density gradually decreases to the neighboring vertices.
Figure~\ref{grqc_terrain} shows that there are several high peaks within the GrQc network,
which means there are several disconnected K-Cores with high K values (dense K-Cores).

\begin{figure}[!h]
\centering
\subfigure[GrQc (spring layout)]
{
\includegraphics[width=0.14\textwidth]{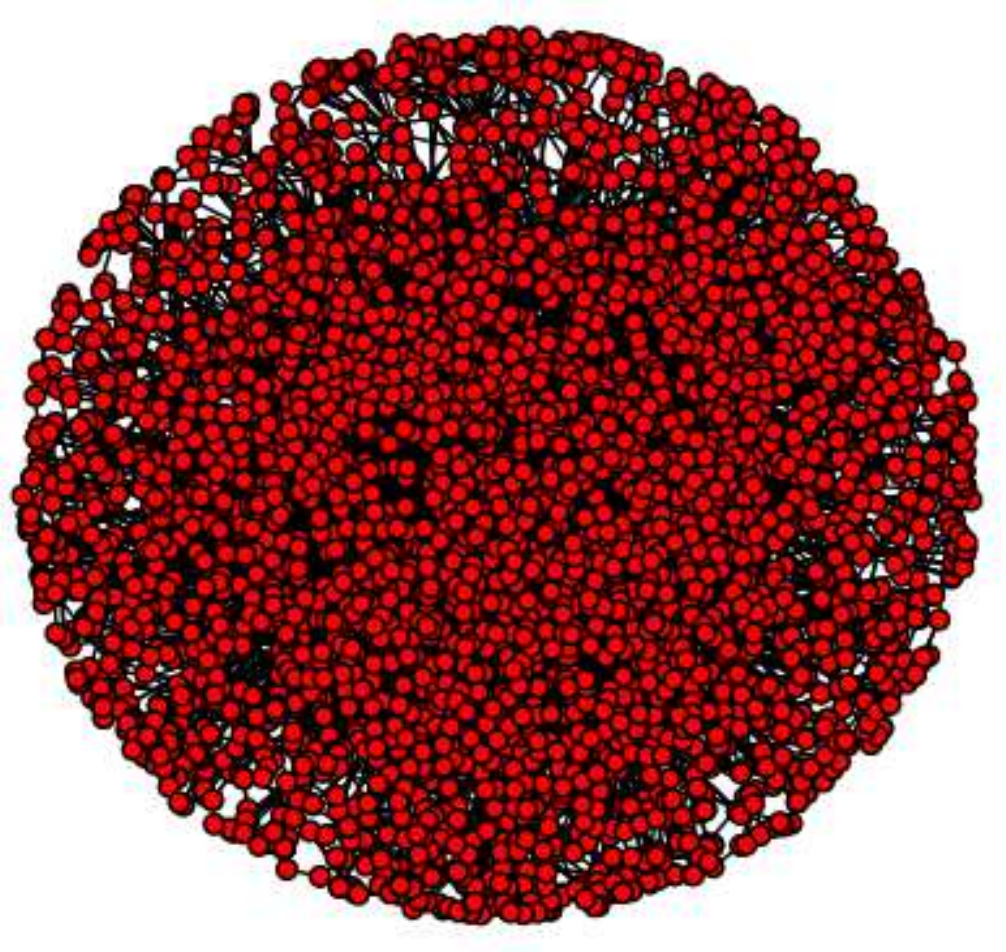}
\label{grqc_spring}
}
 \quad
\subfigure[wikiVote (spring layout)]
{
\includegraphics[width=0.22\textwidth]{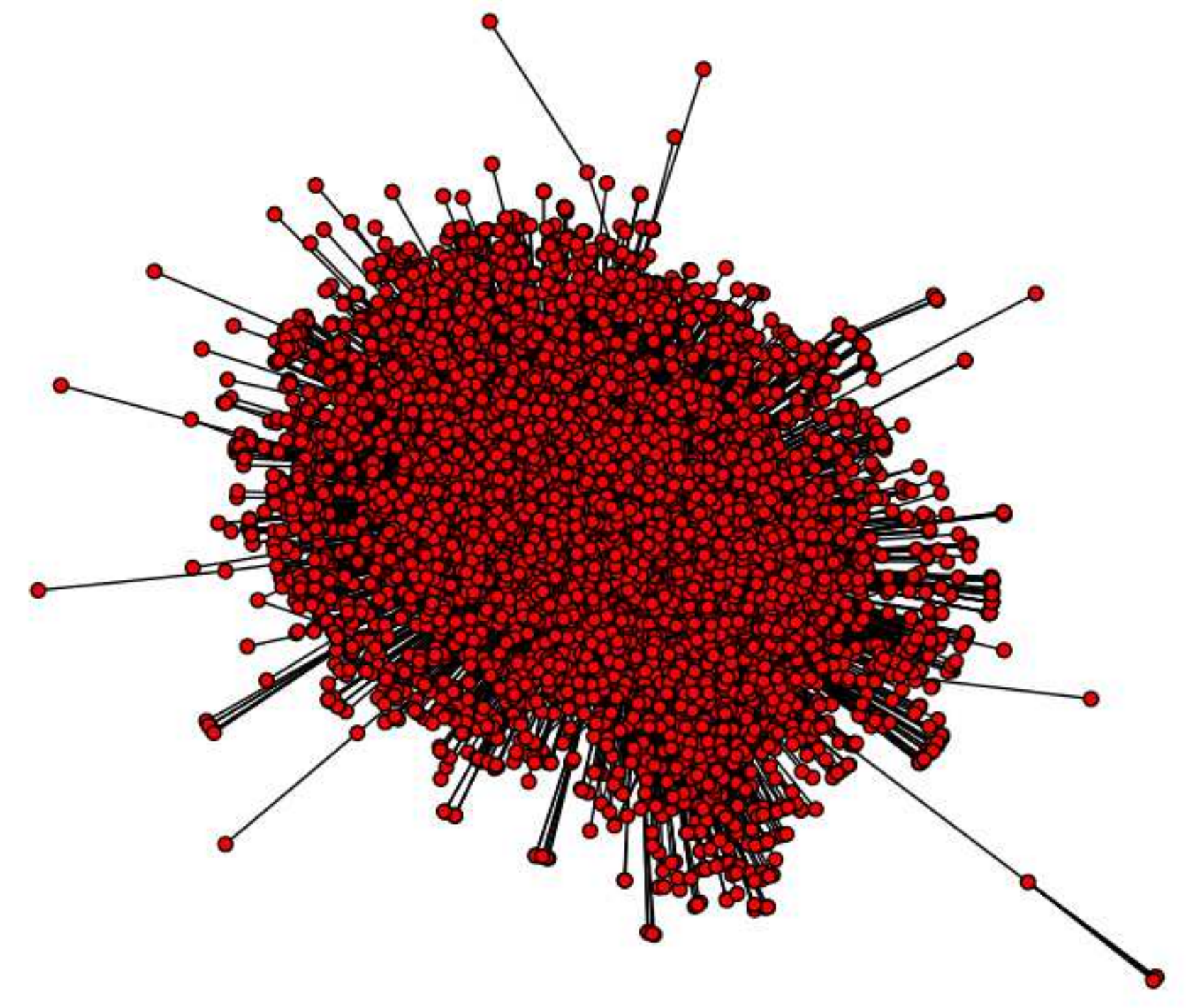}
\label{wikivote_spring}
}
\subfigure[GrQc (K-Core)]
{
\includegraphics[width=0.38\textwidth]{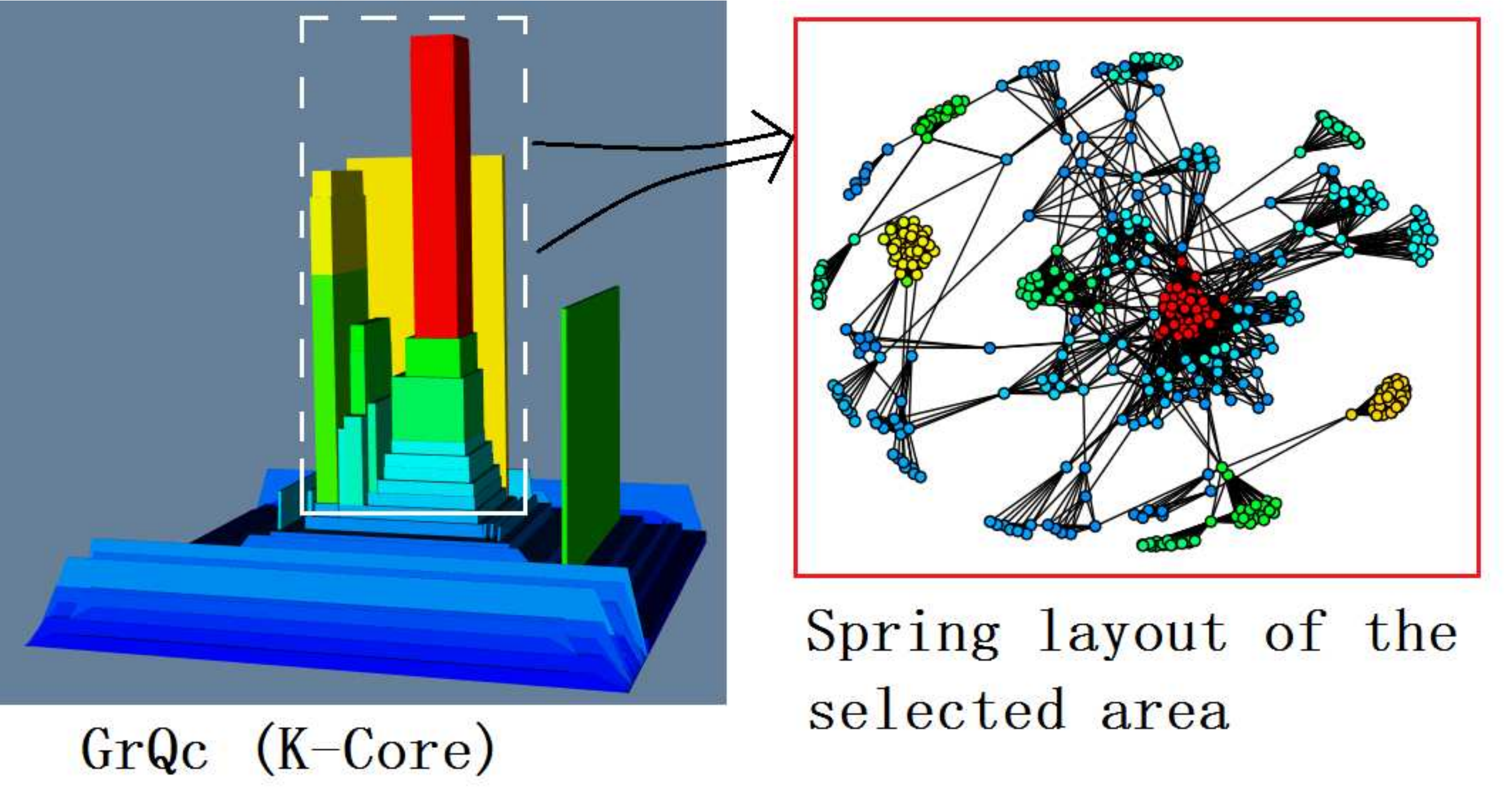}
\label{grqc_terrain}
}


\subfigure[wikiVote (K-Core)]
{
\includegraphics[width=0.18\textwidth]{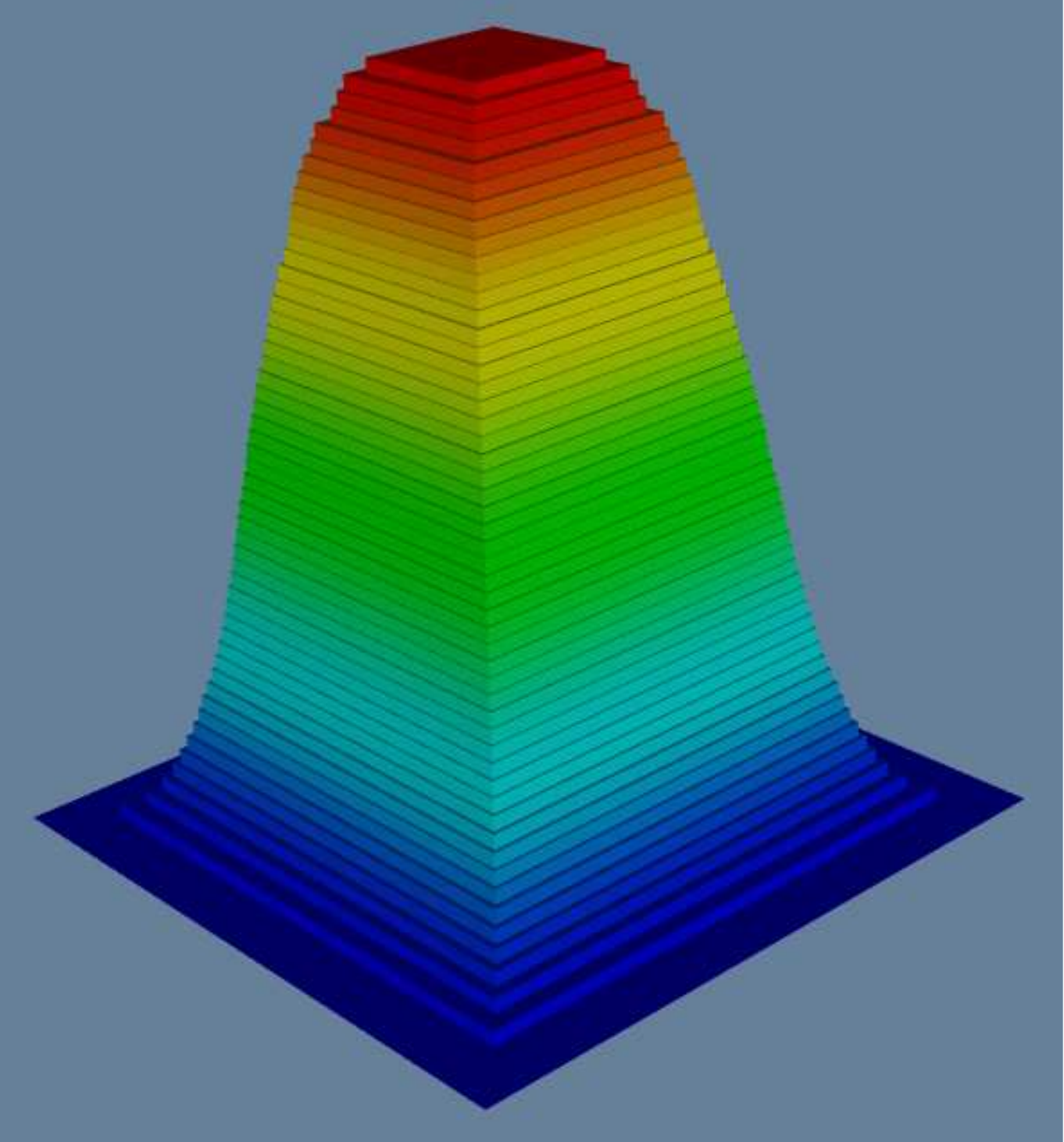}
\label{wikivote_terrain}
}
\quad
\subfigure[GrQc (K-Truss)]
{
\includegraphics[width=0.18\textwidth]{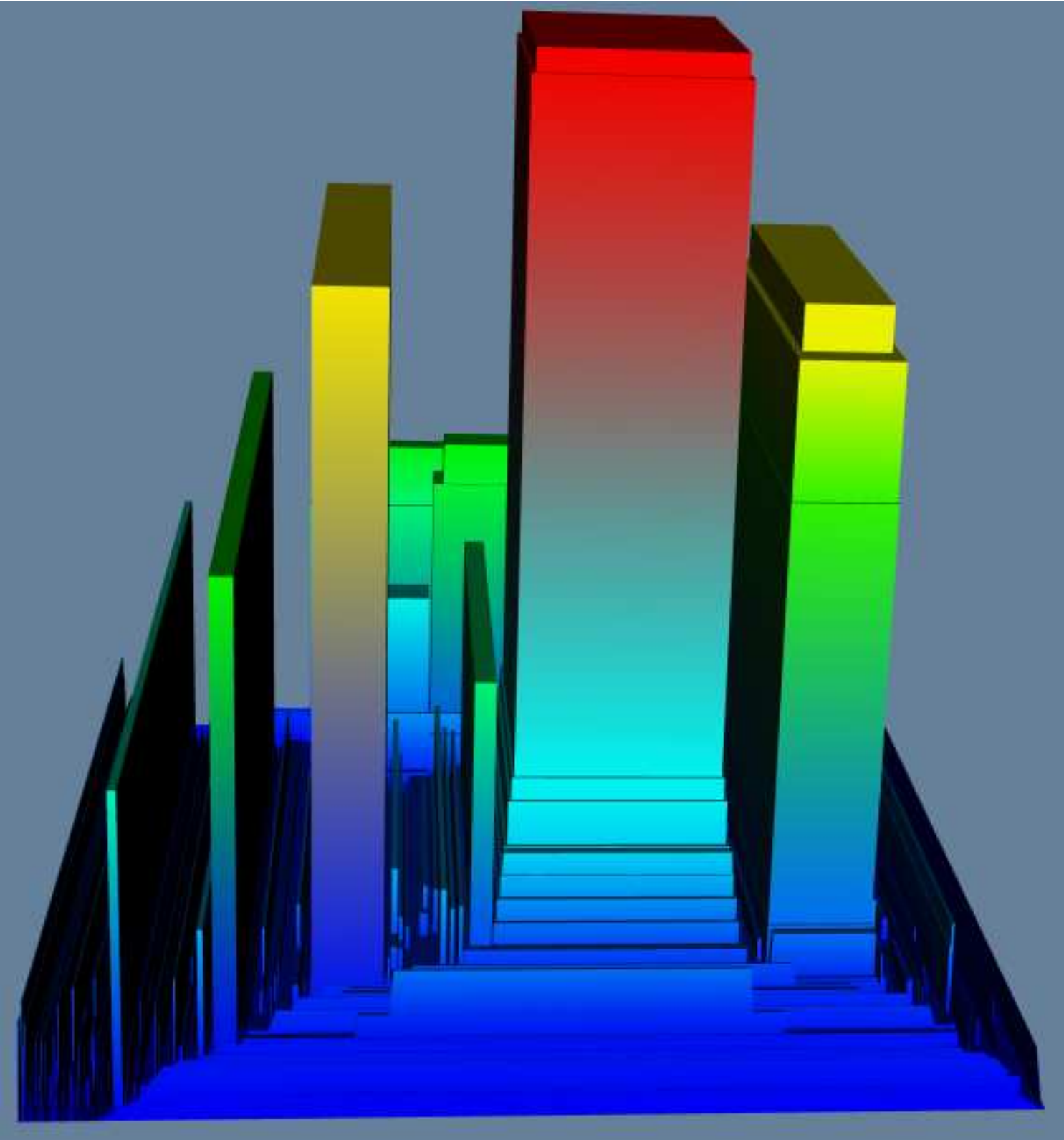}
\label{GrQc_tkcore_3d}
}

\subfigure[2D Visualization of K-Cores in GrQc]
{
\includegraphics[width=0.18\textwidth]{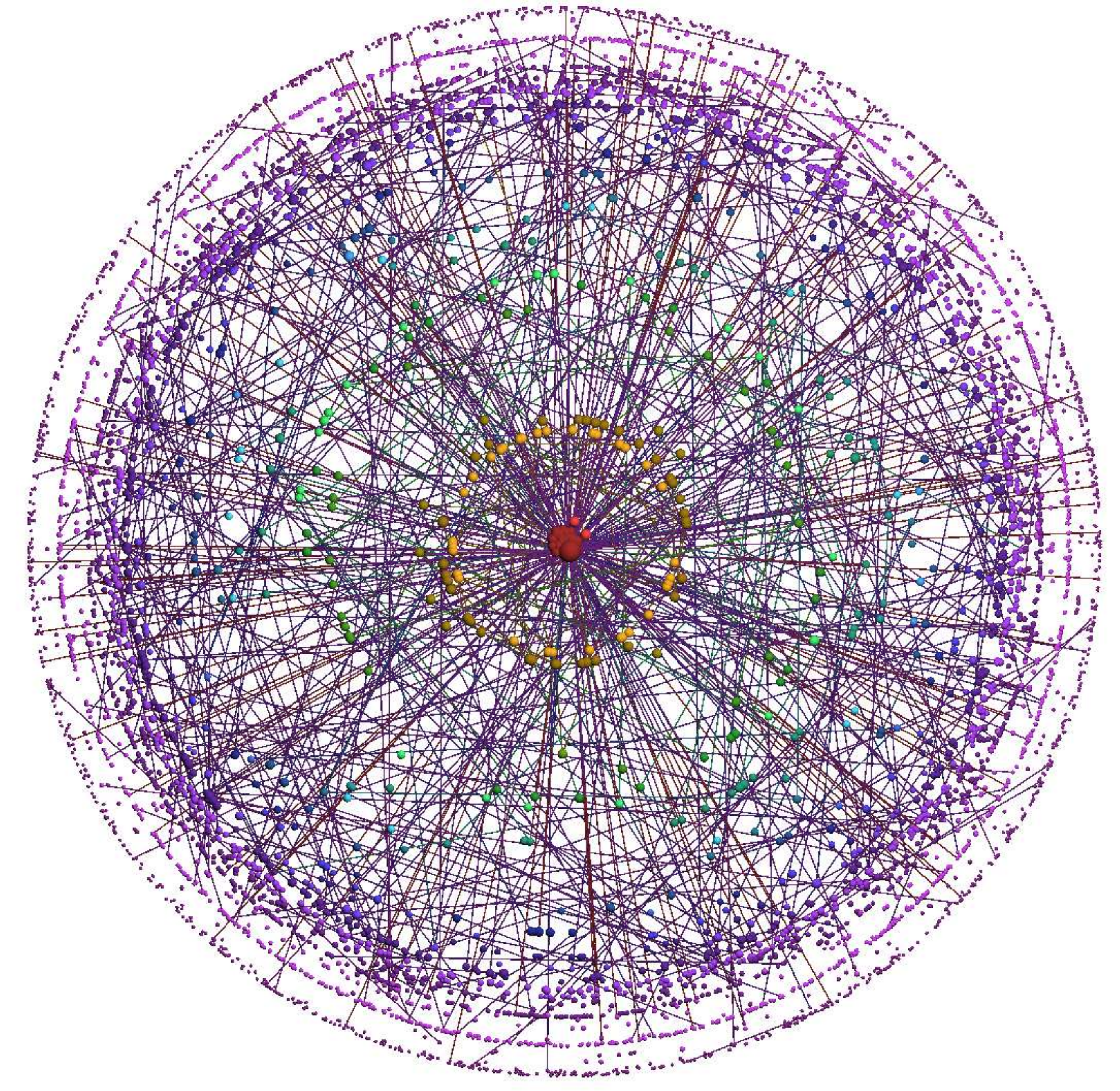}
\label{GrQc_kcore_others}
}
 \quad
\subfigure[2D plot of K-Trusses in GrQc]
{
\includegraphics[width=0.19\textwidth]{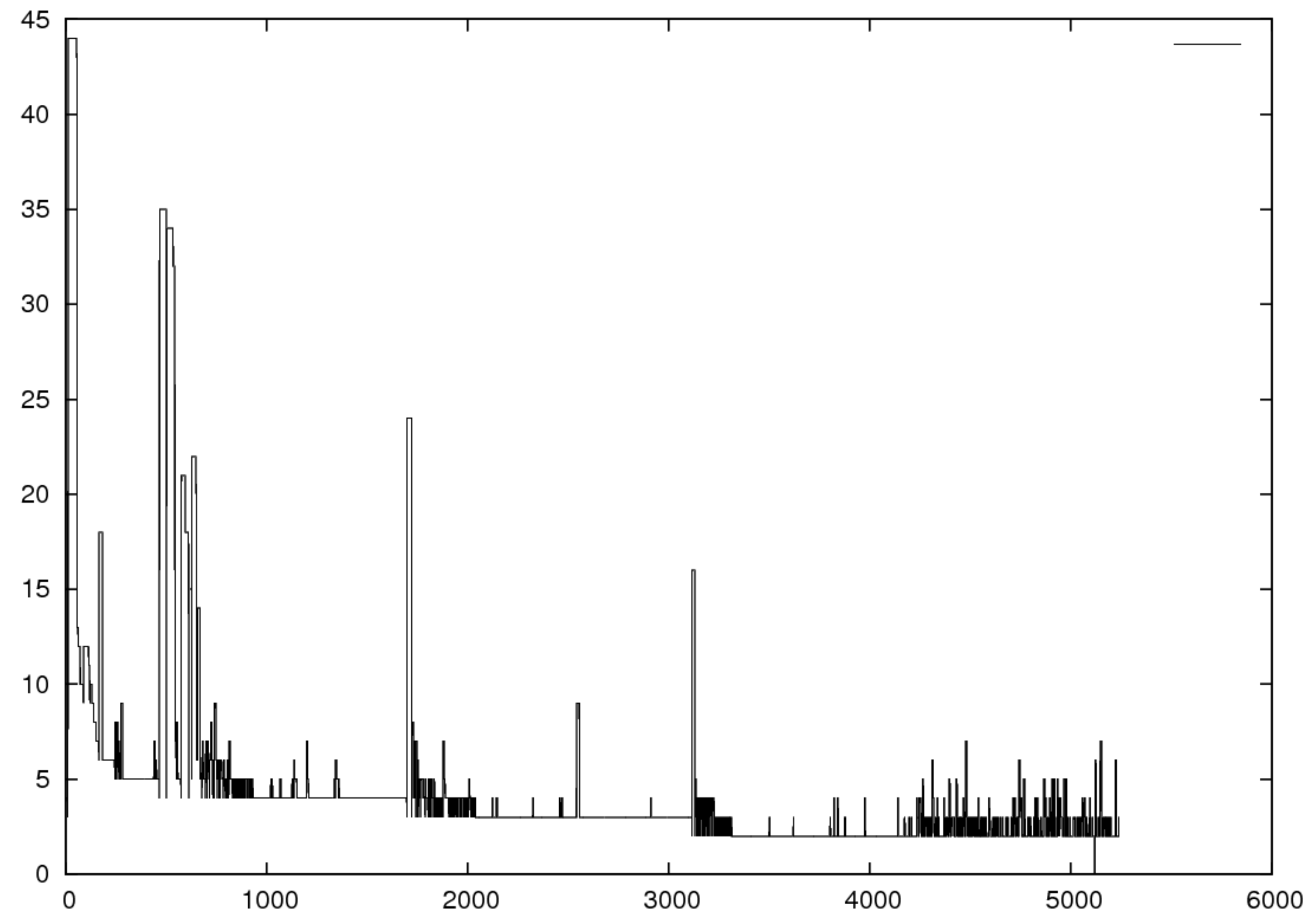}
\label{GrQc_csv_plot}
}
\caption[Optional caption for list of figures]{Visualizing Dense Subgraphs the Network}
\label{examplekcore}
\end{figure}

Moreover, Figure~\ref{grqc_terrain} clearly illustrates the
hierarchical relationship among K-Cores.
In the selected terrain area (the terrain area in dashed line), the red peak is placed on green and blue foundation,
which means the dense K-Core is contained in some less dense K-Cores.
This can be verified by drawing spring layout of the selected region in the red box (with our tool, a user can select a region,
and call other visualization methods to draw the selected region), the red dense K-Core is surrounded by some green and blue vertices.
The visualization of hierarchy is important, as it allows an analyst to derive high level insights on the connectivity that is not immediately obvious even in state of the art K-Core plots as shown in Figure~\ref{GrQc_kcore_others}~\cite{kcore_vis} for
the GrQc network. We give more detailed comparison between terrain visualization and other visualization methods in User Study section.

Also we can color the terrain using a second measure.
In Figure~\ref{GrQc_3d_intro}, we color the terrain based on vertex degree
(red/yellow/green/blue area indicates vertices with highest/high/low/lowest degrees),
we can see that generally KC(v) is positively correlated with degree -- vertices in dense K-Cores have high degrees,
except a few outlier vertices that have relatively high degree but low KC(v) values (the yellow area at the bottom of the terrain).
They are usually local hub nodes with sparse neighborhood.

We can illustrate the same principle when visualizing K-Trusses (used to understand triangle density) instead of K-Cores.
Here each edge uses KT(e) as the scalar measure, and we use the
edge-based scalar graph for
visualizing the K-Trusses in GrQc dataset.
The terrain visualization is in Figure~\ref{GrQc_tkcore_3d} where high peaks indicate dense K-Trusses.
To contrast, Figure~\ref{GrQc_csv_plot} depicts
a CSV plot, a state-of-the-art density plot leveraged within the database
community ~\cite{csv, tkcore12}.
Again such visualization strategies do not reveal important hierarchical relationships (e.g. contains) among different K-Trusses.
Also we note that our visualization method is a common and flexible framework which can render plots based on different scalar measures,
and the ability to rotate, filter and extract details on demand (allowing the analyst to quickly identify regions of interest) will help users
 understand the graph data better.

\noindent{\bf Scalability:}
We next examine the efficiency of Algorithm~\ref{ScalarTree}, \ref{PostScalarTree} and \ref{EdgeScalarTree}.
Every dataset has duplicate scalar values,
so the generated trees are all super trees.
We test our methods on datasets of various sizes,
and list the number of nodes in the final super (edge) scalar tree($N_{t}$), time cost to construct the tree (tc) and visualize the tree (tv) in Table \ref{wiki_cit_info}.
The time cost to construct the tree (tc) includes the time cost to construct the tree ( Algorithm~\ref{ScalarTree} or \ref{EdgeScalarTree}) and postprocess the tree (Algorithm~\ref{PostScalarTree}).
The time cost to visualize the tree (tv) is the time cost for the visualization software to read the scalar tree and render the terrain visualization.

\begin{table}[!h]
\fontsize{8}{10}\selectfont
\centering
\caption{Terrain Visualization Time Cost(sec)}
\begin{tabular}{|c|c|c|c|c|c|c|c|c|}\hline
  Dataset & Scalar & \pbox{20cm}{$N_{t}$} & $tc$ & $te$ & $tv$ \\\hline
  GrQc & KC(v) &  869 & 0.0018 &  & $<$1\\\hline
  GrQc & KT(e) &  728&  0.0039& 0.0072 & $<$1 \\\hline
  WikiVote & KC(v) & 106&  0.0037 &  & $<$1 \\\hline
  WikiVote & KT(e) & 44& 0.053& 0.69& $<$1 \\\hline
  Wikipedia & KC(v) &  230 & 6.9 &  & 2 \\\hline
  Wikipedia & KT(e) &  1,903& 49.3 &16334  & 22\\\hline
  Cit-Patent & KC(v) & 1,059&  7.1 &  & 2\\\hline
  Cit-Patent & KT(e) &  110,412& 27.7 & 65.3 & 13 \\\hline
\end{tabular}
\label{wiki_cit_info}
\end{table}

 Also we list the time cost of the naive method (using dual graph) to build edge-scalar tree (te) in Table~\ref{wiki_cit_info}.
We can see the improved method (tc) is much faster than the naive method (te), especially on the Wikipedia dataset,
the improved method is more than 300 times faster than the naive method.

Additionally,
in Figure~\ref{wiki_cit_vis},
we show the terrain visualizations of (edge) scalar tree of wikipedia and cit-Patent datasets.
The peaks in Figure~\ref{wiki_kcore_3d} \ref{wiki_tkcore_3d_peak} \ref{cit_kcore_3d_peak} \ref{cit_tkcore_3d}
indicate dense K-Cores and K-Trusses in the network,
we highlight the highest peaks in Figure~\ref{wiki_tkcore_3d_peak} and Figure~\ref{cit_kcore_3d_peak},
and draw the details in Figure~\ref{wiki_tkcore_peak_2d} and Figure~\ref{cit_kcore_peak_2d}.
They are a K-Truss with K = 86 and a K-Core with K = 64.

\begin{figure}[!h]
\centering
\subfigure[Wikipedia Network (K-Core)]
{
\includegraphics[width=0.20\textwidth]{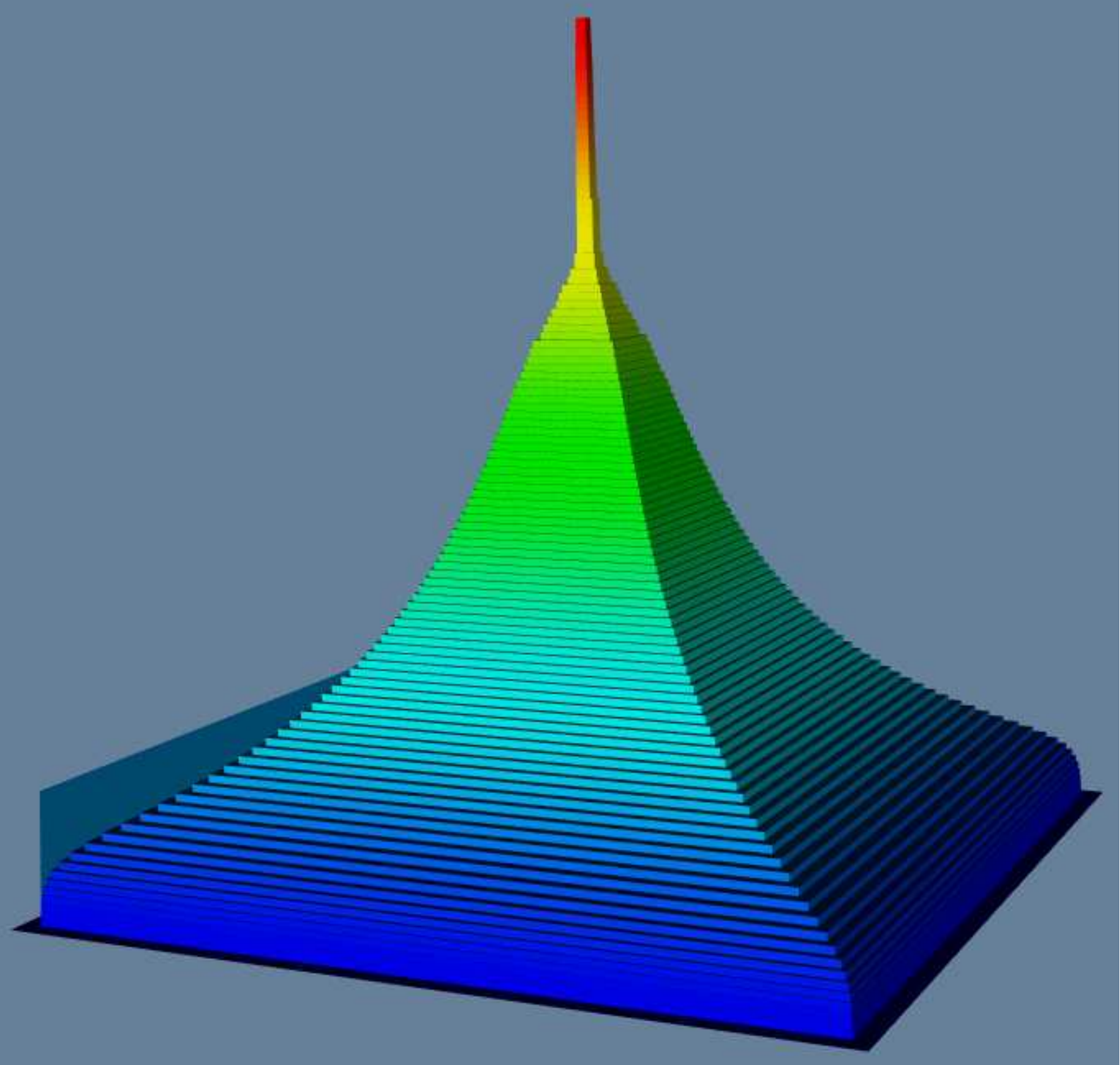}
\label{wiki_kcore_3d}
}
 \quad
\subfigure[Wikipedia Network (K-Truss)]
{
\includegraphics[width=0.20\textwidth]{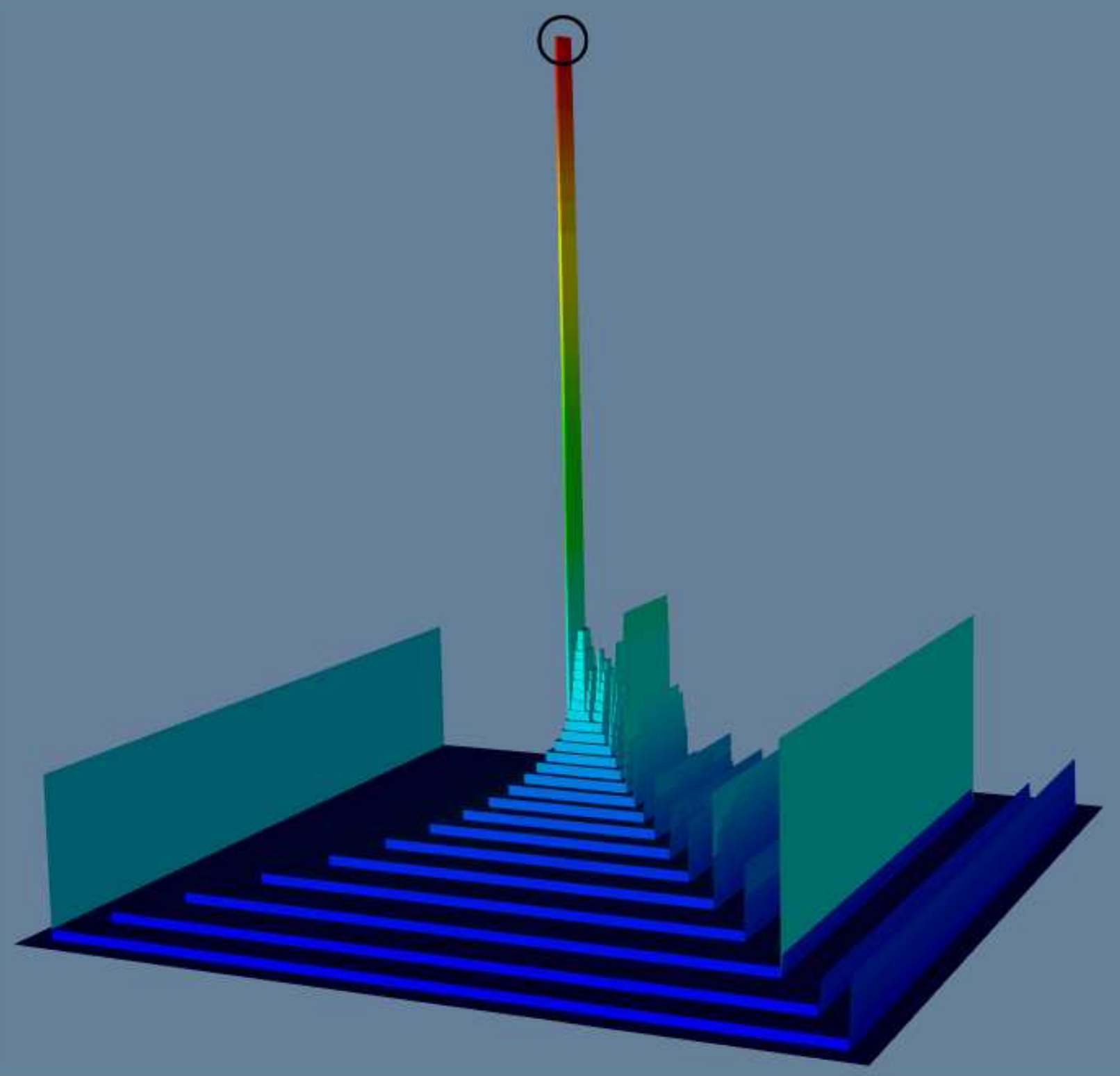}
\label{wiki_tkcore_3d_peak}
}
\subfigure[Cit-Patent (K-Core)]
{
\includegraphics[width=0.20\textwidth]{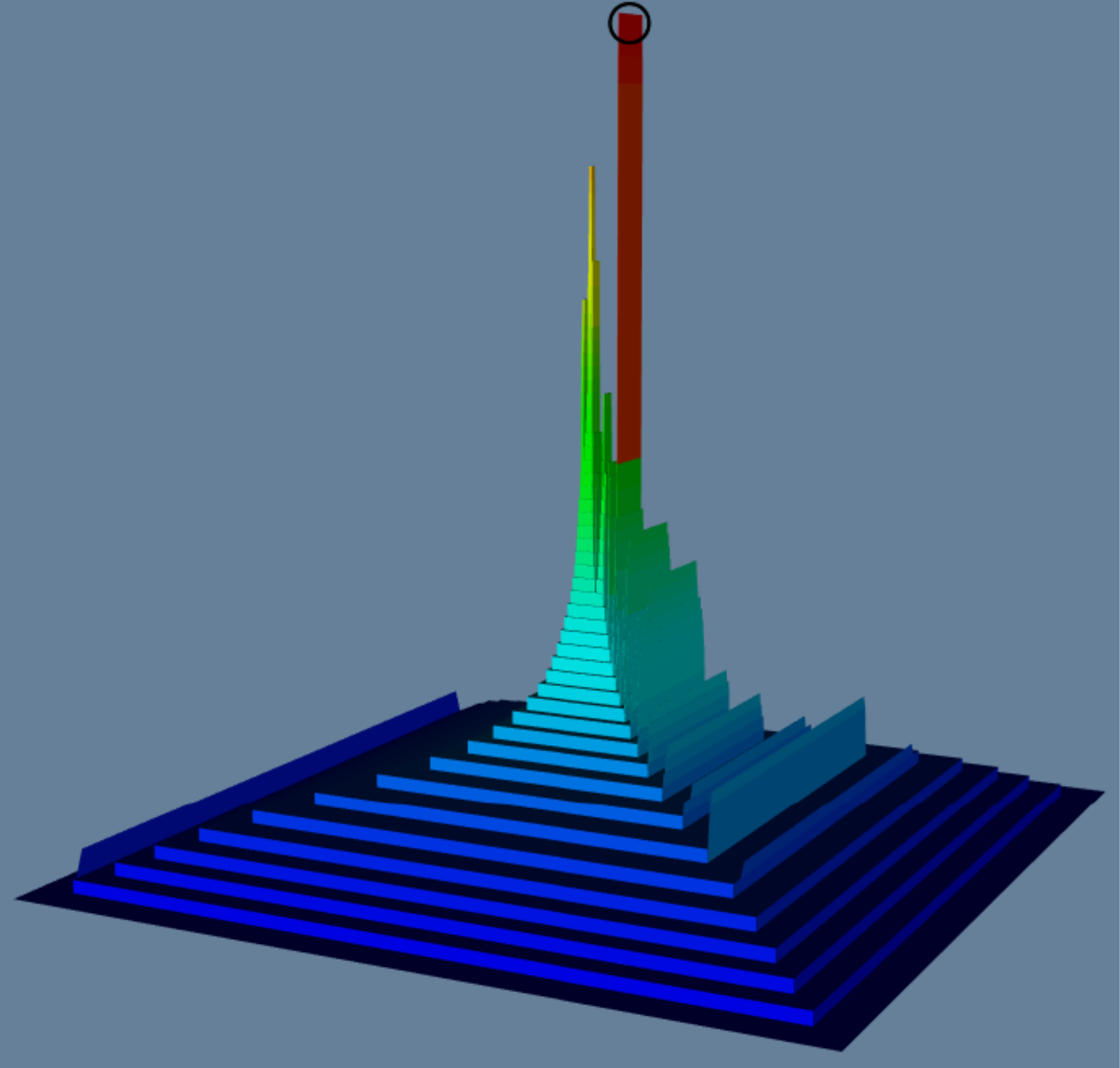}
\label{cit_kcore_3d_peak}
}
 \quad
\subfigure[Cit-Patent (K-Truss)]
{
\includegraphics[width=0.20\textwidth]{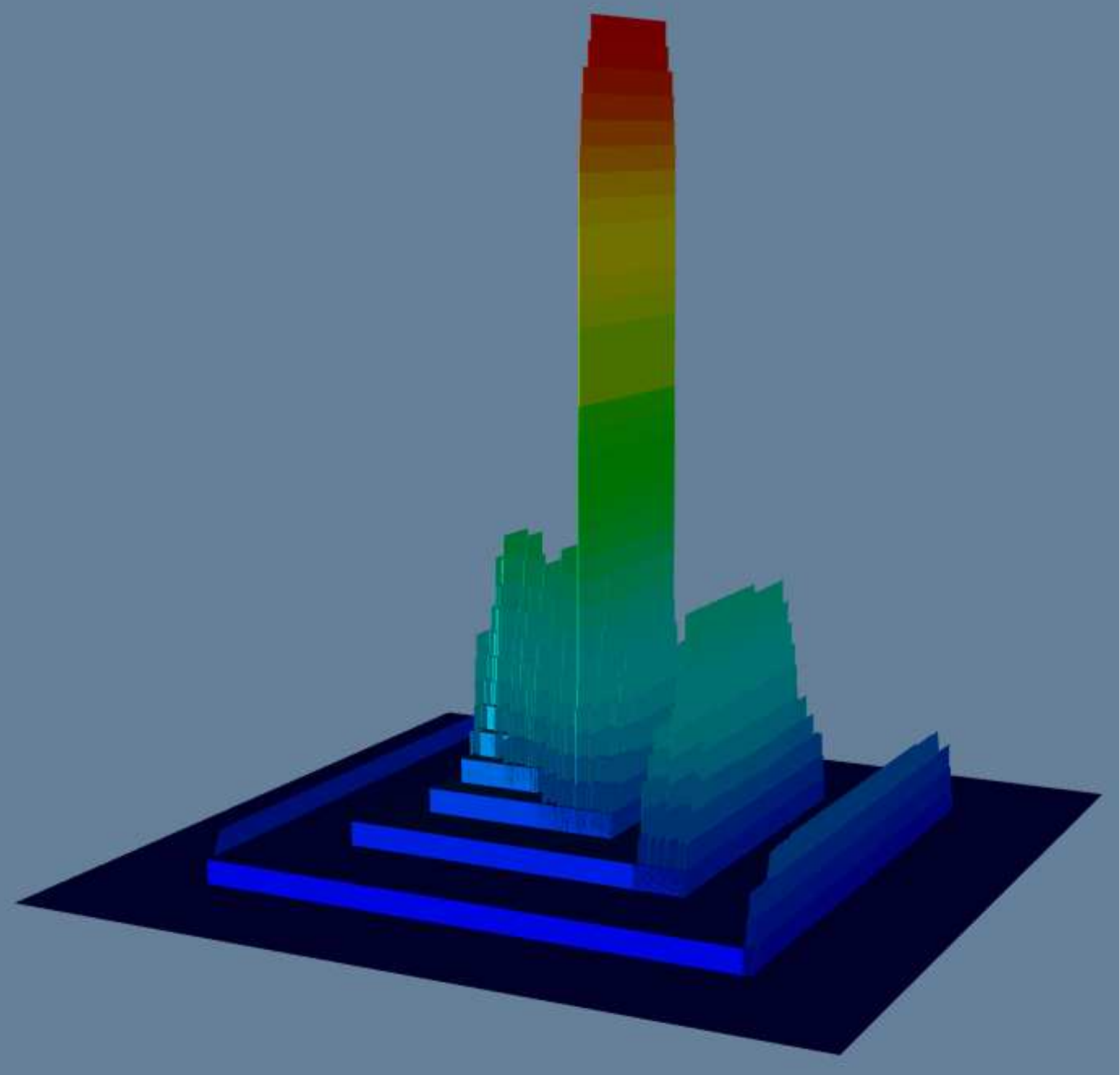}
\label{cit_tkcore_3d}
}
\subfigure[Densest K-Truss of Wiki]
{
\includegraphics[width=0.19\textwidth]{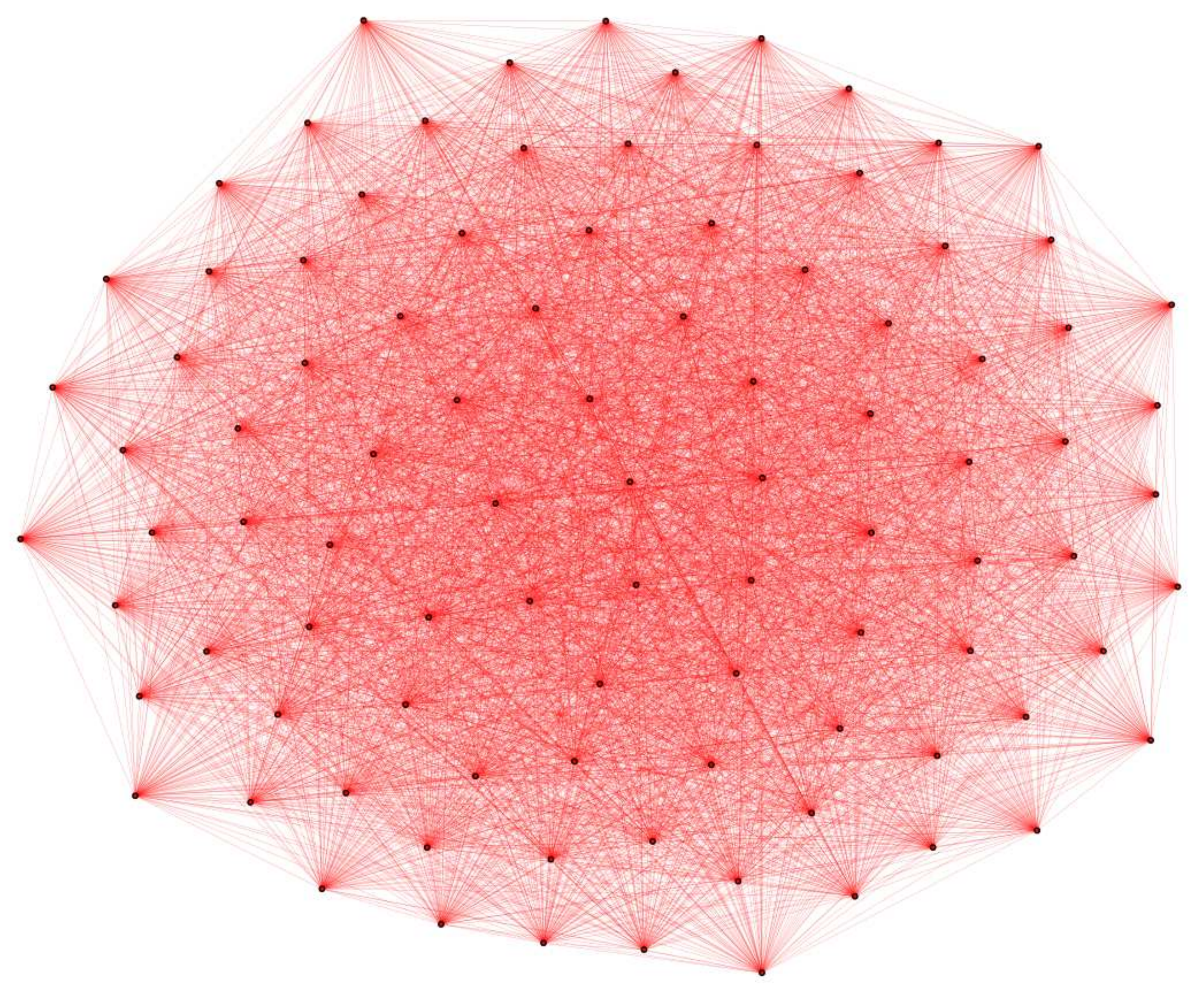}
\label{wiki_tkcore_peak_2d}
}
 \quad
\subfigure[Densest K-Core of Cit]
{
\includegraphics[width=0.19\textwidth]{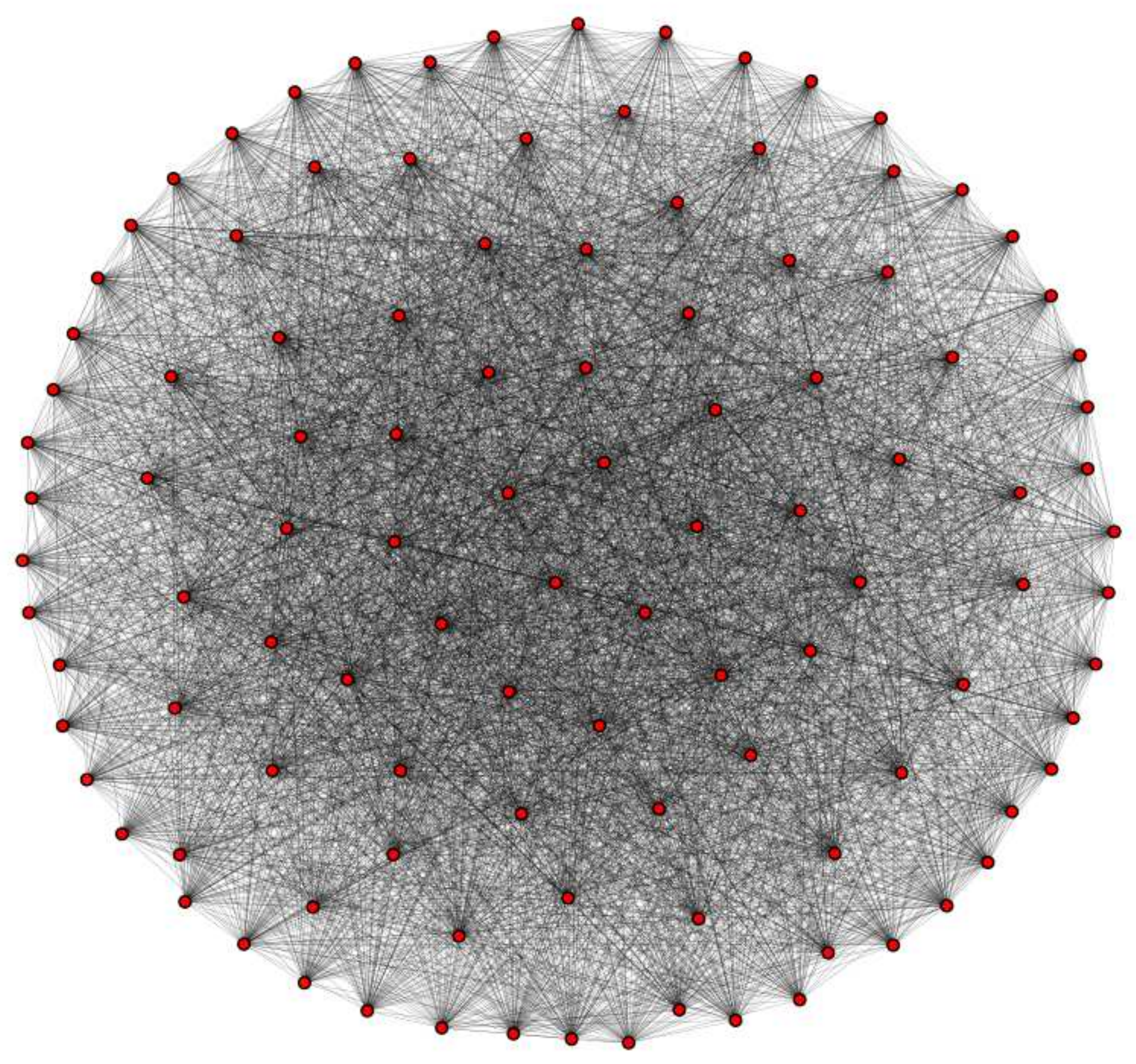}
\label{cit_kcore_peak_2d}
}
\vspace{-0.15in}
\caption[Optional caption for list of figures]{Visualizing K-Cores and K-Trusses}
\label{wiki_cit_vis}
\end{figure}

\vspace{-0.07in}
\subsection{Visualizing Communities and Roles}
\noindent {\bf Visualizing Community Affiliation:}
We next illustrate the flexibility of the terrain visualization scheme
on an important network science task
-- understanding community structures
(see Figure~\ref{dblp4area_communities}).
We use a subset of the DBLP network (DBLP(sub)) for this purpose, comprising authors who publish in
the areas of Machine Learning, Data mining, Databases and  Information Retrieval.
We apply a state-of-the-art overlapping (soft) community detection algorithm~\cite{overlap13} on this
dataset to detect four communities. Each author in the dataset
is affiliated with a community score vector
$(c_{0}, c_{1}, c_{2}, c_{3})$ indicating how much it belongs to each community.
To visualize the affiliation of a particular community $i$, we use $c_{i}$ as the corresponding
scalar measure, and draw the terrain of the network.
$peak_{\alpha}$ in the terrain indicates a connected component in which every vertex has $c_{i} \geq \alpha$.

In Figure~\ref{dblp4areacomm1allcircle}, we visualize community 1,
in which most authors are database researchers.
We highlight two peaks in the circle of Figure~\ref{dblp4areacomm1allcircle},
and zoom in to get a clear picture of the two peaks on the right.
Our tool allows us to easily select authors (vertices) in each peak.
We find that authors in the left peak include researchers
\emph{Donald Kossmann, Divyakant Agrawal, Amr El Abbadi, Michael Stonebraker, Samuel Madden, and
Joseph M. Hellerstein} while authors in the right peak include
\emph{Zheng Chen, Hongjun Lu, Jeffrey Xu Yu, Beng Chin Ooi, Kian-Lee Tan, Qiang Yang and Aoying Zhou}.
Since authors in both peaks have high community scores ($c_{1}$),
they can be seen as core members of the community although from different geographic areas.
The fact that they are in two separate peaks indicates that authors in one peak do not
work with authors in the other peak in the dataset.
Similarly, we also observe subcommunities in another community (Figure~\ref{dblp4areacomm2allcircle}) largely comprising Machine Learning researchers.
We also find two peaks in the terrain,
and authors in the left peak are
\emph{Philip S. Yu, Christos Faloutsos, Michael I. Jordan, Stuart J. Russell, Daphne Koller, Sebastian Thrun,  Wei Fan and Andrew Y. Ng, }
who all work in United States,
while authors in the right peak are \emph{Hang Li, Ji-Rong Wen, Tie-Yan Liu, Lei Zhang, Wei-Ying Ma, Qiang Yang and Yong Yu,} who are researchers in China.

Figure~\ref{dblp4comm_intro} visualizes the four communities together to give an overview of them.

\begin{figure}[!h]
\centering
\subfigure[Community 1]
{
\includegraphics[width=0.14\textwidth]{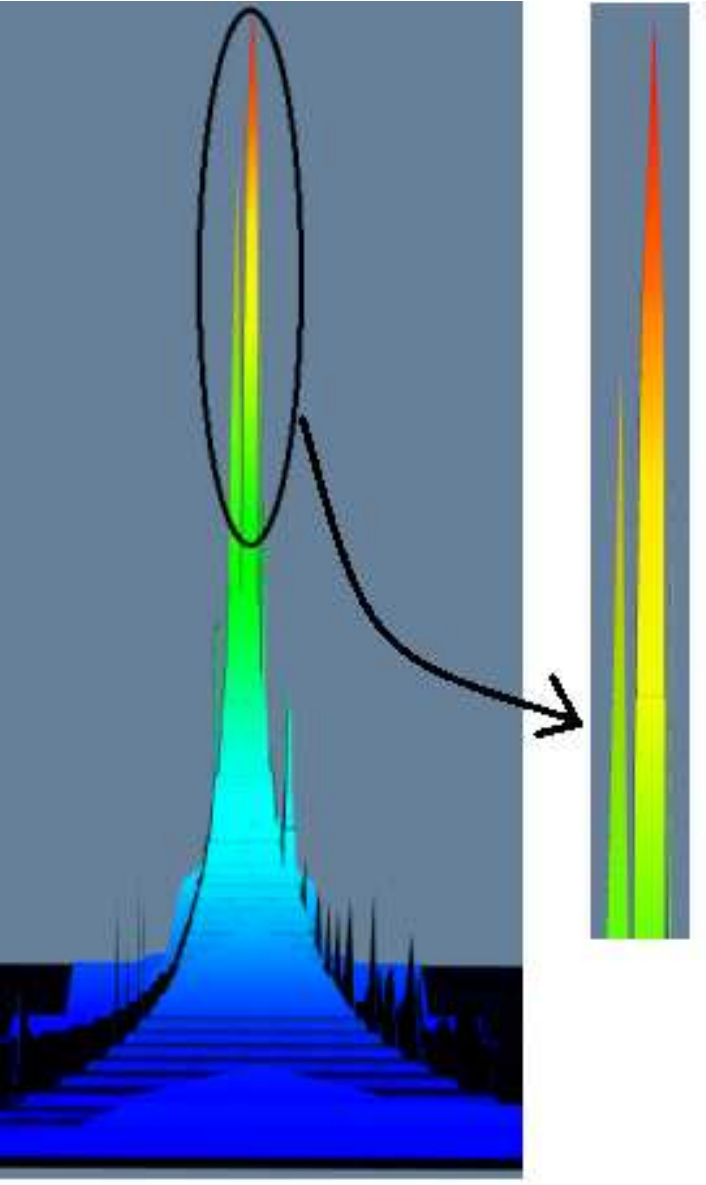}
\label{dblp4areacomm1allcircle}
}
 \quad
\subfigure[Community 2]
{
\includegraphics[width=0.13\textwidth]{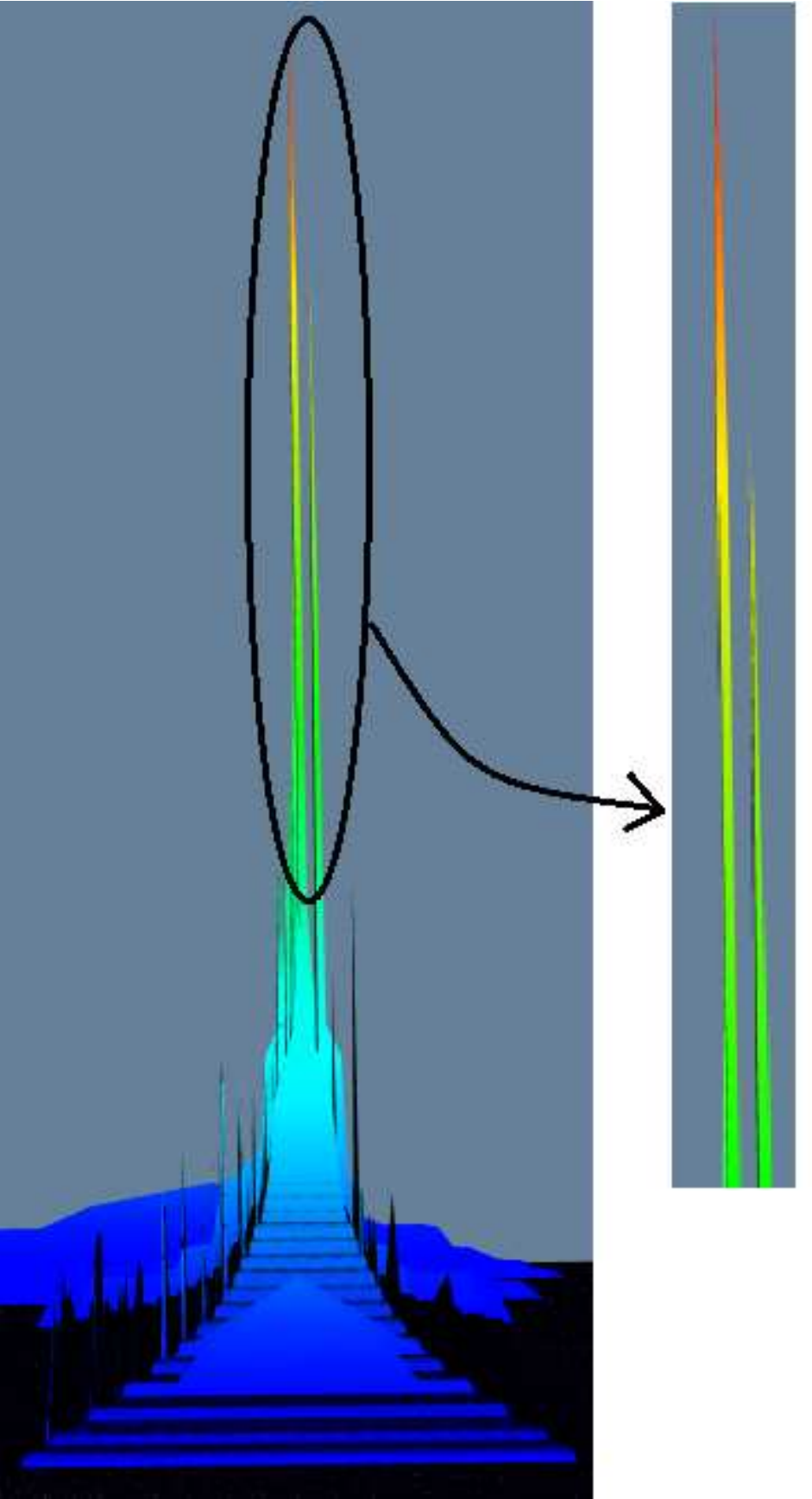}
\label{dblp4areacomm2allcircle}
}

\vspace{-0.10in}
\caption[Optional caption for list of figures]{Visualizing two communities in DBLP network}
\label{dblp4area_communities}
\end{figure}

\noindent {\bf Visualizing Role and Community Affiliation:}
Moving beyond community affiliation, the ability to uncover
the roles of individual nodes (e.g. bridge, hub, periphery and whisker)
within a network or community has received recent interest~\cite{rolx,commrole14}.
Here we examine how one may use the terrain to
visualize the distribution of roles over a community.
We leverage a recent idea to simultaneously detect communities and roles on large scale
networks~\cite{commrole14}. For each vertex in the network the algorithm
outputs a community affinity vector ($c_1, \ldots, c_m$) and a role affinity vector
($r_1, \ldots, r_n$).

As before we focus on a particular community (community $i$) and use the community score ($c_i$) of each vertex
to create terrain visualization.
The peak in Figure~\ref{amazon_3dvis} contains the vertices affiliated with one major community in Amazon co-purchase network.
Instead of re-using the intensity of community score ($c_i$) to color vertices we actually use the dominant role for each
vertex (four roles is typical~\cite{rolx}) to color vertices.
We assign each role a color,
the ``hub vertex'' is green,
the ``dense community vertex'' is blue,
the ``periphery vertex'' is red.
Then we assign the color of roles to the terrain in Figure~\ref{amazon_3dvis}.
From the terrain visualization, we can see that the vertices in the community have 3 roles,
the hub vertex has the highest community score (green top),
and below it is the blue portion, which means the ``hub vertex'' is surrounded by some ``dense community vertices'' in the network.
The red part of the peak indicates that there are some ``peripheral vertices'' attached to the community.
Since the community contains a small number of vertices,
we can draw the details of the community using node-link visualization in Figure~\ref{amazon_2d}.

\begin{figure}[h]
\centering

\subfigure[Roles on a community]
{
\includegraphics[width=0.15\textwidth]{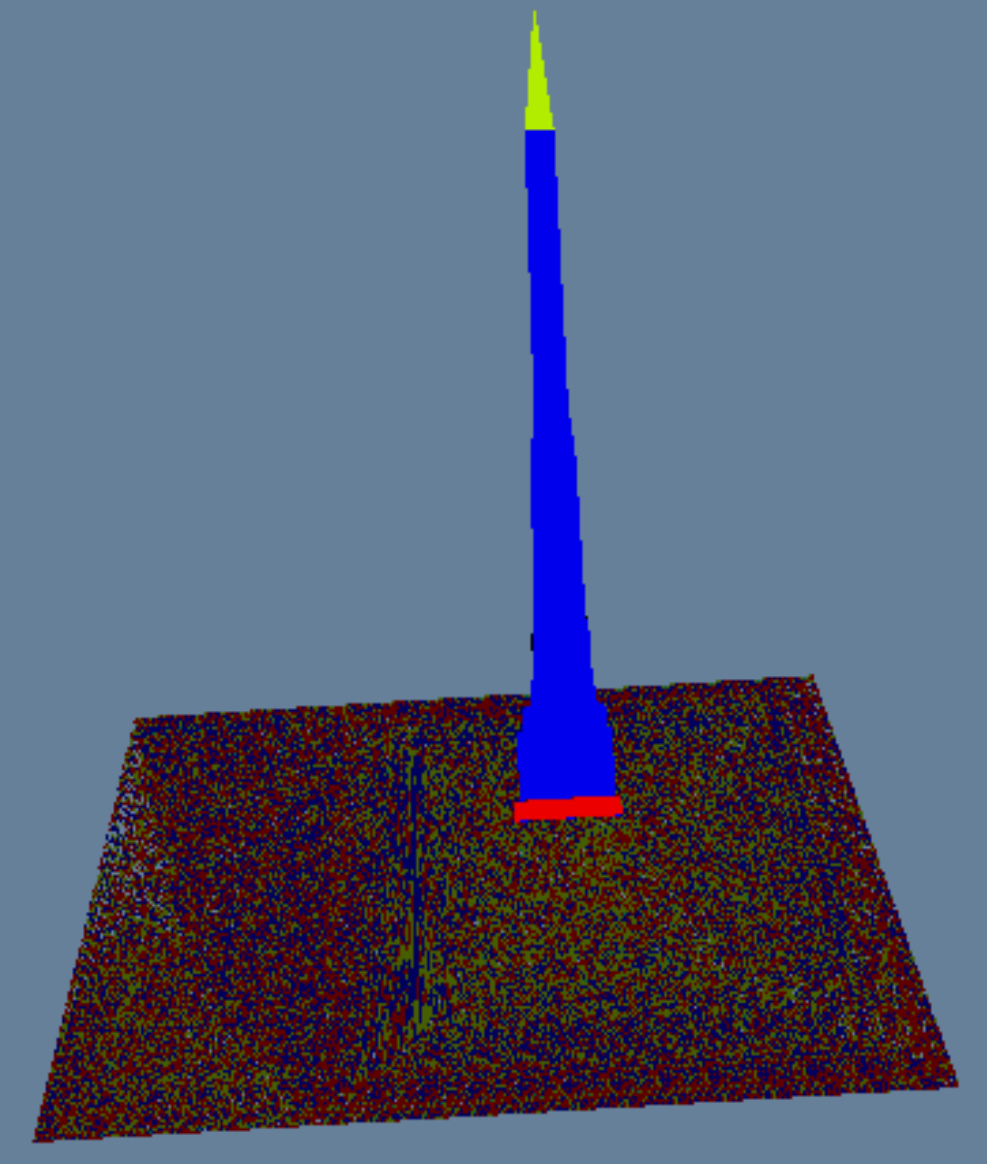}
\label{amazon_3dvis}
}
 \quad
\subfigure[Detail of the community]
{
\includegraphics[width=0.18\textwidth]{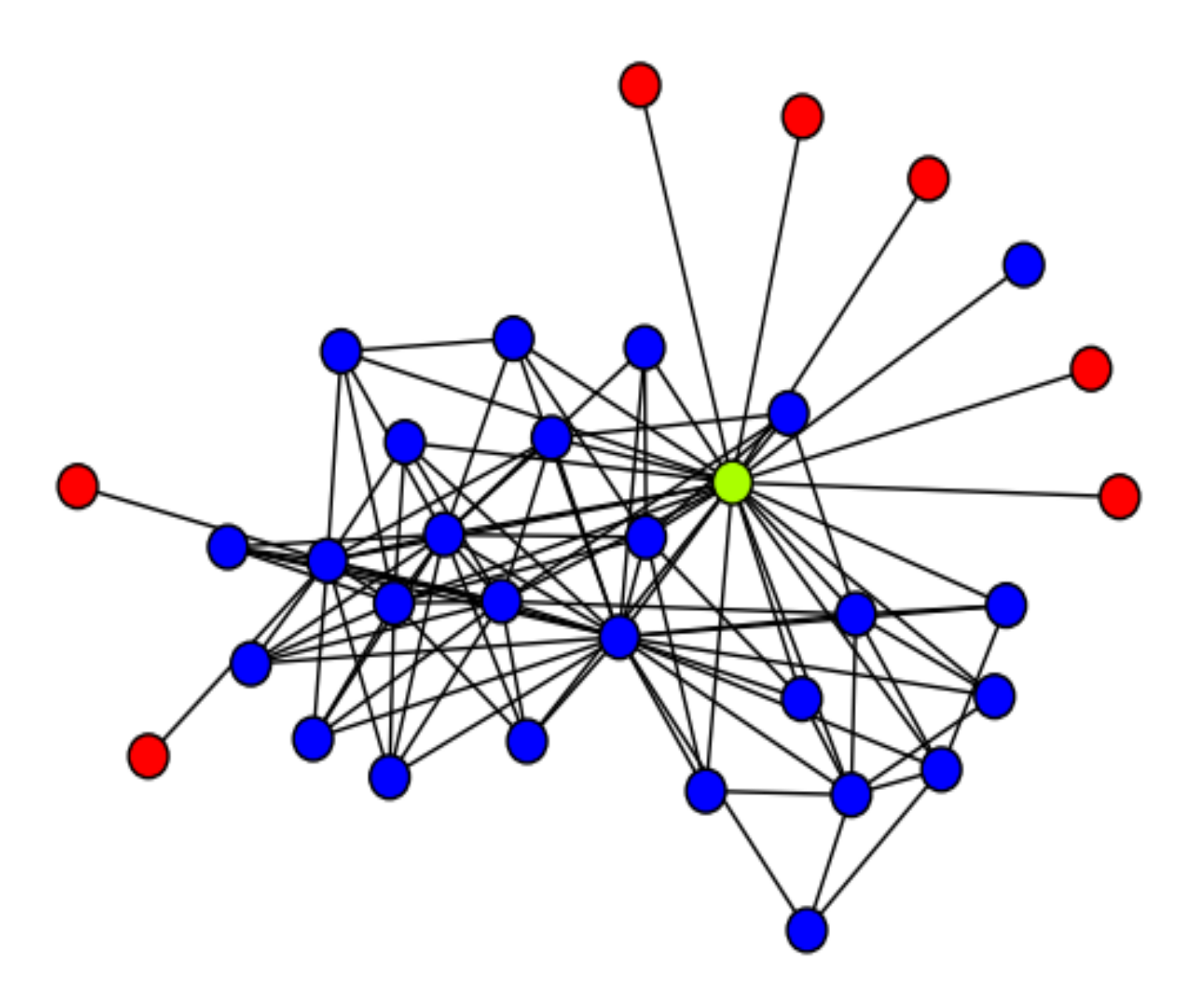}
\label{amazon_2d}
}

\vspace{-0.10in}
\caption[Optional caption for list of figures]{one community of Amazon co-purchase network}
\label{amazon_community}
\end{figure}

All nodes in Figure~\ref{amazon_2d} are books on Amazon, and we list a few of them in Table~\ref{amazon_names}.
The green node is the book which has the highest salesrank and is focused on creativity (hub).
Most of the blue nodes are books about creativity (densely connected),
while the red nodes are books loosely relevant to creativity (periphery).

\vspace{-0.10in}
\begin{table}[!hbtp]
\fontsize{8}{10}\selectfont
\centering
\caption{Book names of some nodes in Figure~\ref{amazon_2d}}
\begin{tabular}{|c|c|}\hline
  Role & Book Name \\\hline
  green & Artist's Way, The PA \\\hline
  blue & Heart Steps: Prayers and Declarations for a Creative Life \\\hline
  blue & Inspirations: Meditations from the Artist's Way \\\hline
  blue & Reflections on the Artist's Way \\\hline
  blue & The Artist's Way Creativity Kit \\\hline
  red & Writing From the Inner Self \\\hline
  red & \pbox{20cm}{Codes of Love : How to Rethink\\ Your Family and Remake Your Life} \\\hline
\end{tabular}
\label{amazon_names}
\end{table}

\subsection{Comparing Different Centralities}
In this section we examine the use of our approach for understanding the relationship of
two different measures of  centrality across various nodes within a network.
We will compare two centralities, degree centrality and betweenness centrality, as two scalar fields, $S_{d}$ and $S_{b}$.

We use the Astro Physics collaboration network, in which each author is a vertex, and each edge indicates a coauthorship between two authors.
We first compute the Local Correlation Index of each vertex (as described earlier),
 and then compute the Global Correlation Index of the network,
$GCI_{S_{d}, S_{b}}$ = 0.89.
This indicates that the overall correlation between degree centrality and  betweenness centrality is highly positive.

In this case, we are interested in those vertices with negative LCI values,
as they could be seen as outliers.
We define an outlier score for each vertex $v$ as follows:
\begin{flalign*}
outlier\_score(v) = - LCI(v)
\end{flalign*}
the vertex with more negative LCI(v) will have higher outlier score.
We use $outlier\_score(v)$ as scalar field to draw the terrain in Figure~\ref{astro_corr2_3d_highlight},
and color the terrain using $S_{d}$ (degree centrality), where red/yellow/blue indicates high/moderate/low degree.
We notice that most high peaks are blue, which indicates that the outlier vertices usually have low degree.

We drill down into the two peaks in the terrain within black and red circles,
and select the vertex at the top of each peak.
Our software allows us to integrate a different visualization method to inspect the selected two vertices.
In this case, we use spring layout to draw the two vertices' 2-hop neighborhoods in Figure~\ref{astro_corr2_peak1_highlight} and Figure~\ref{astro_corr2_peak2_highlight}.
We pick these two vertices specifically because one is in high peak while the other appears to be in a broader but smaller peak.
In both cases the vertices picked have high outlier score, which indicates the correlation between degree centrality and betweenness centrality is
negative in their neighborhood. Actually the two vertices have relatively higher betweenness and lower degree
centrality when compared to many of their neighbors.
From Figure~\ref{astro_corr2_peak1_highlight} and Figure~\ref{astro_corr2_peak2_highlight} we can see the two vertices (in the circles) are like bridge nodes connecting multiple communities.

\begin{figure}[!t]
  \begin{tabular}[c]{cc}
    \subfigure[Terrain of Astro Network based on Outlierness]
    {
    \includegraphics[width=0.20\textwidth]{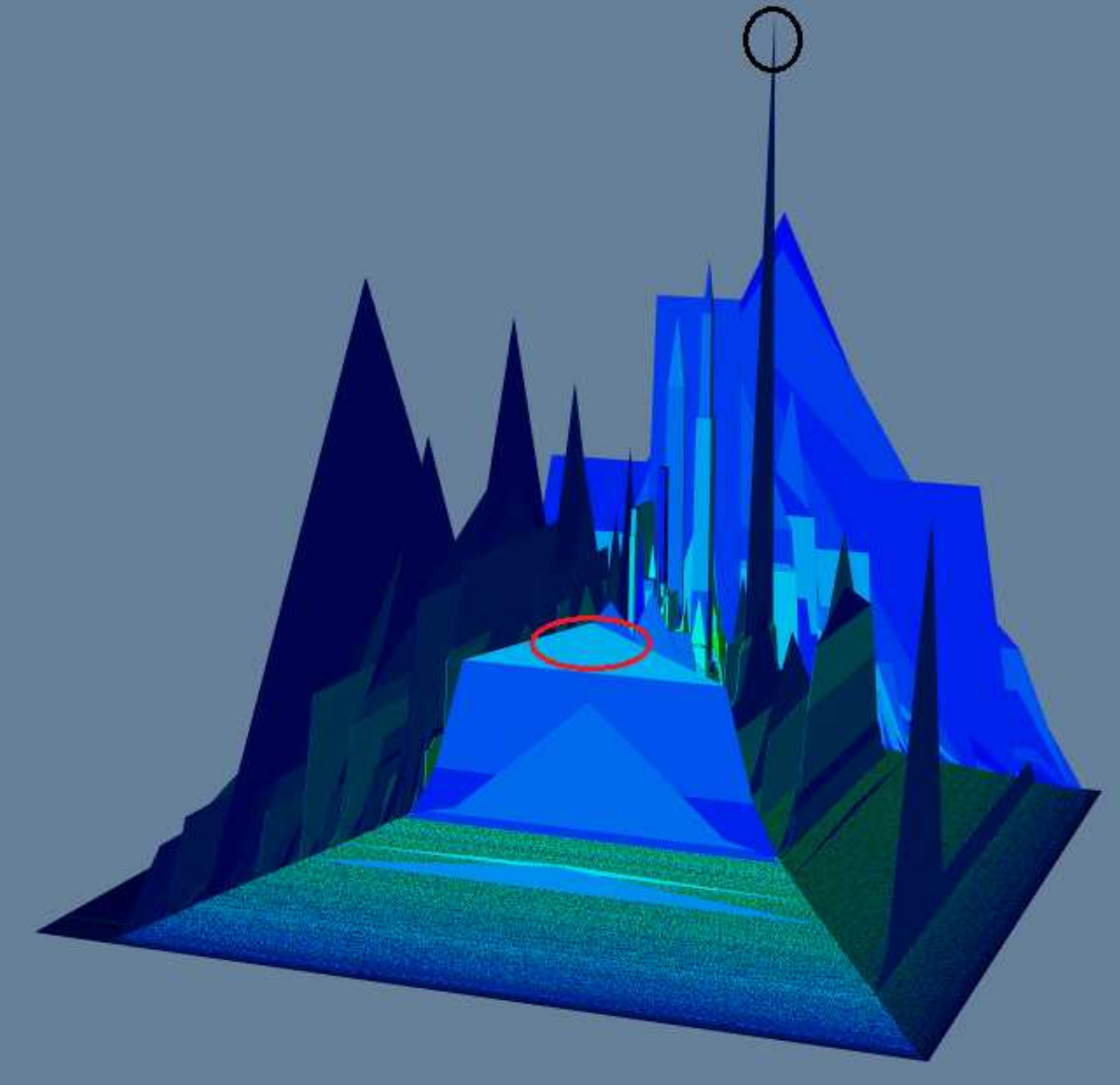}
    \label{astro_corr2_3d_highlight}
    }
    &
      \begin{tabular}[c]{c}
        \subfigure[2-hop neighborhood of vertex in black circle in Figure~\ref{astro_corr2_3d_highlight}]
        {
        \includegraphics[width=0.17\textwidth]{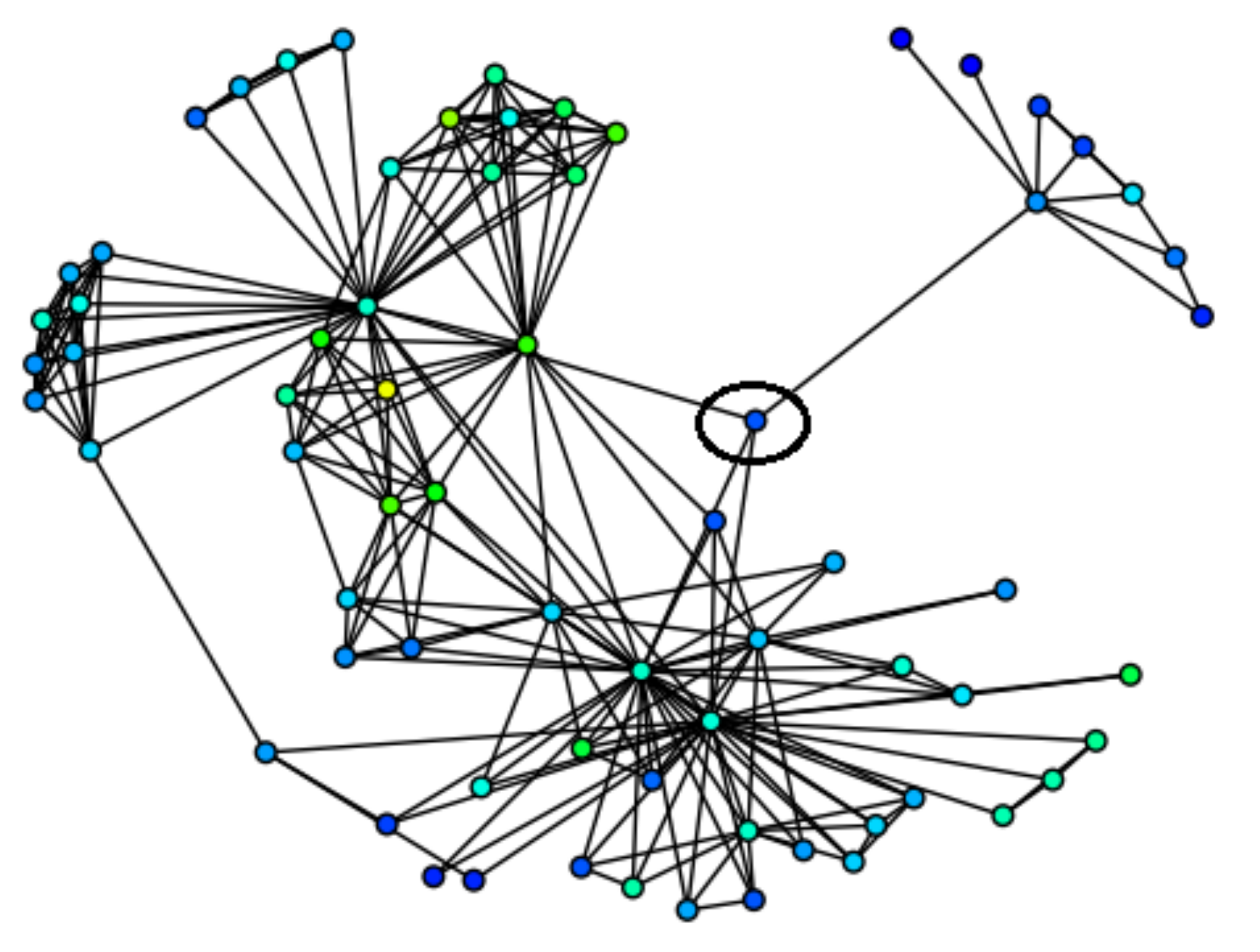}
        \label{astro_corr2_peak1_highlight}
        }\\
        \subfigure[2-hop neighborhood of vertex in red circle in Figure~\ref{astro_corr2_3d_highlight}]
        {
        \includegraphics[width=0.17\textwidth]{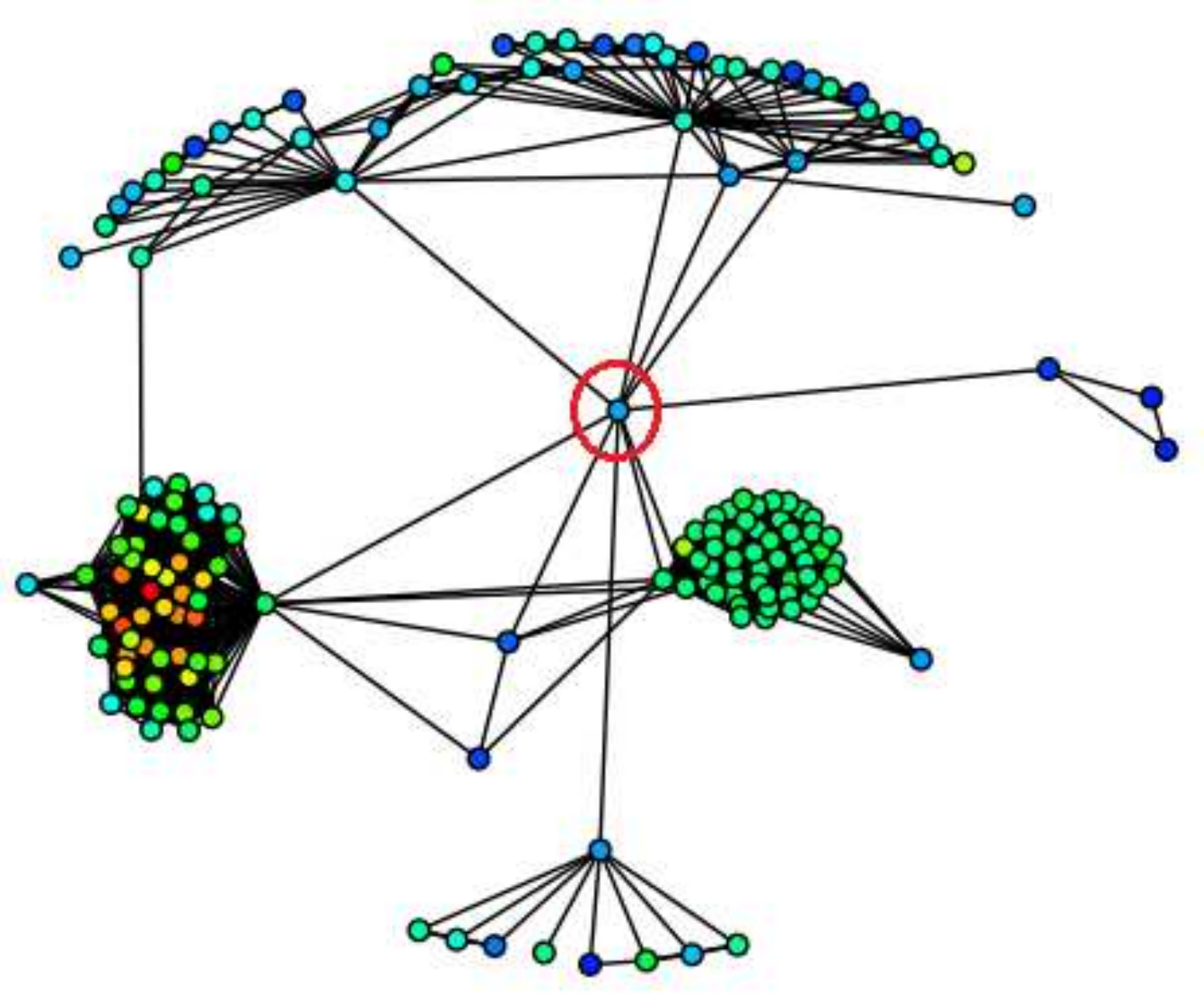}
        \label{astro_corr2_peak2_highlight}
        }
      \end{tabular}
  \end{tabular}

\vspace{-0.15in}
\caption[Optional caption for list of figures]{Compare Degree and Betweenness Centralities}
\label{astro_corr}
\end{figure}

\subsection{Query Result Understanding}
Our terrain visualization can also be extended to visualize the results of SQL queries.
Here we consider a database of plant genus curated by OSU's Horticulture department. In the interests of space, we focus on a common query posed to this dataset, specified by a domain expert,
and model the output (a 5 dimensional materialized table)
as a nearest-neighbor (NN) graph (distance measure and threshold again specified by domain expert) and then visualize the graph
using the terrain visualization (Figure~\ref{plant}).
Color represents different plant genus (3 types in query output), height is a scalar value representing
the values of two of the selected attributes from the query result (attributes 1 and 2).
While details of the genus are omitted for expository simplicity,
the query result visualization clearly conveys the following:
i) the result set from the SQL query contains
three plant genus (red, green and blue) of which the blue genus  is well separated from the other two; ii) It is also clear that the (red) genus is closer to the (green) genus and is in some senses contained within it, i.e. more central,
from a connectivity standpoint (within the NN graph); iii) finally attribute 1 demonstrates greater genus separability (variance in terrain heights across genus) on the subset of data produced by this particular query. While preliminary in nature, such visualizations can potentially allow the domain expert to better understand the coherency of the output w.r.t the selection predicates (attributes) of the query.

\begin{figure}[!h]
\centering
\subfigure[Attribute 1 as Scalar value]
{
\includegraphics[width=0.18\textwidth]{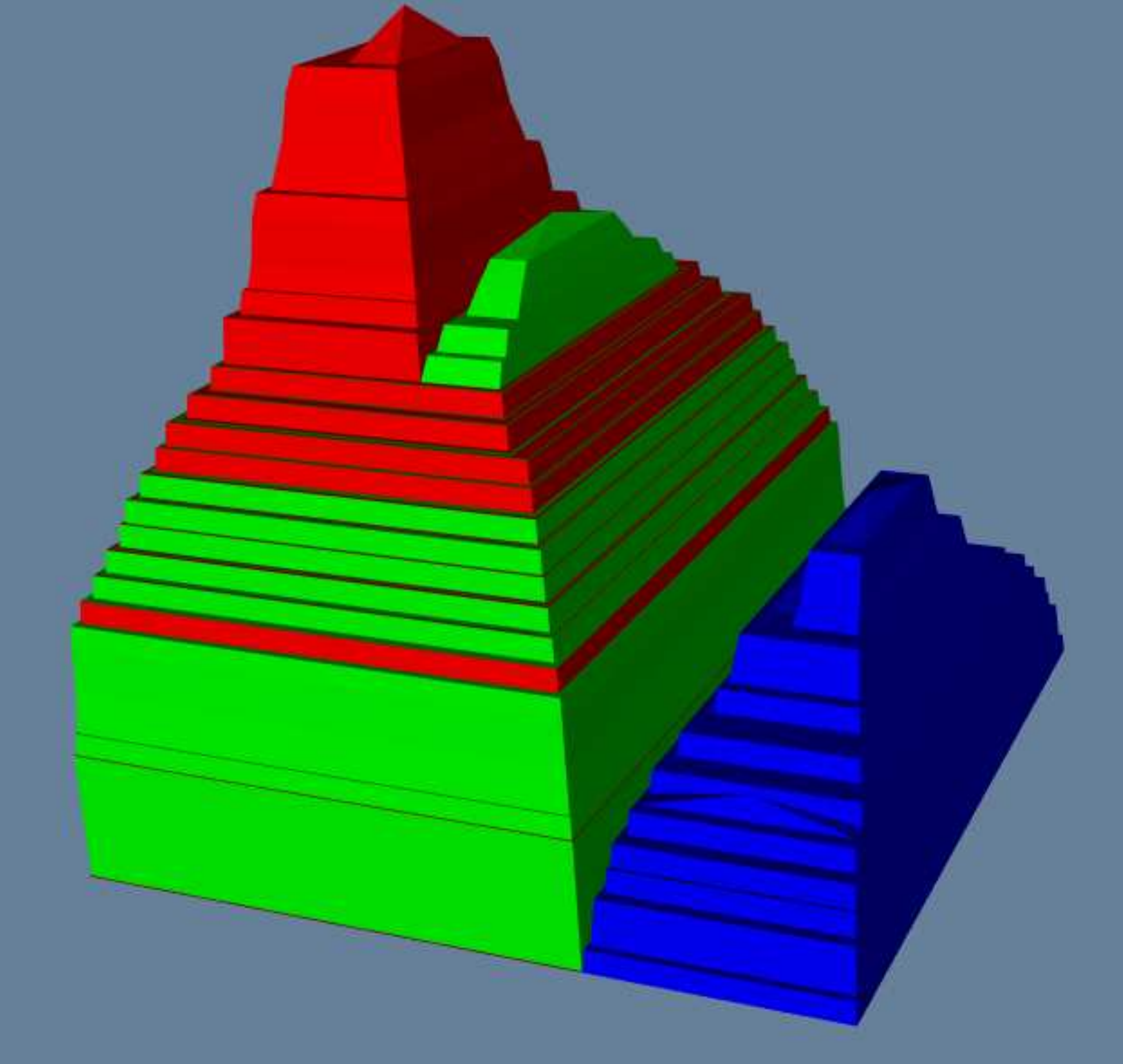}
\label{betweenness_terrain}
}
 \quad
\subfigure[Attribute 2 as Scalar value]
{
\includegraphics[width=0.18\textwidth]{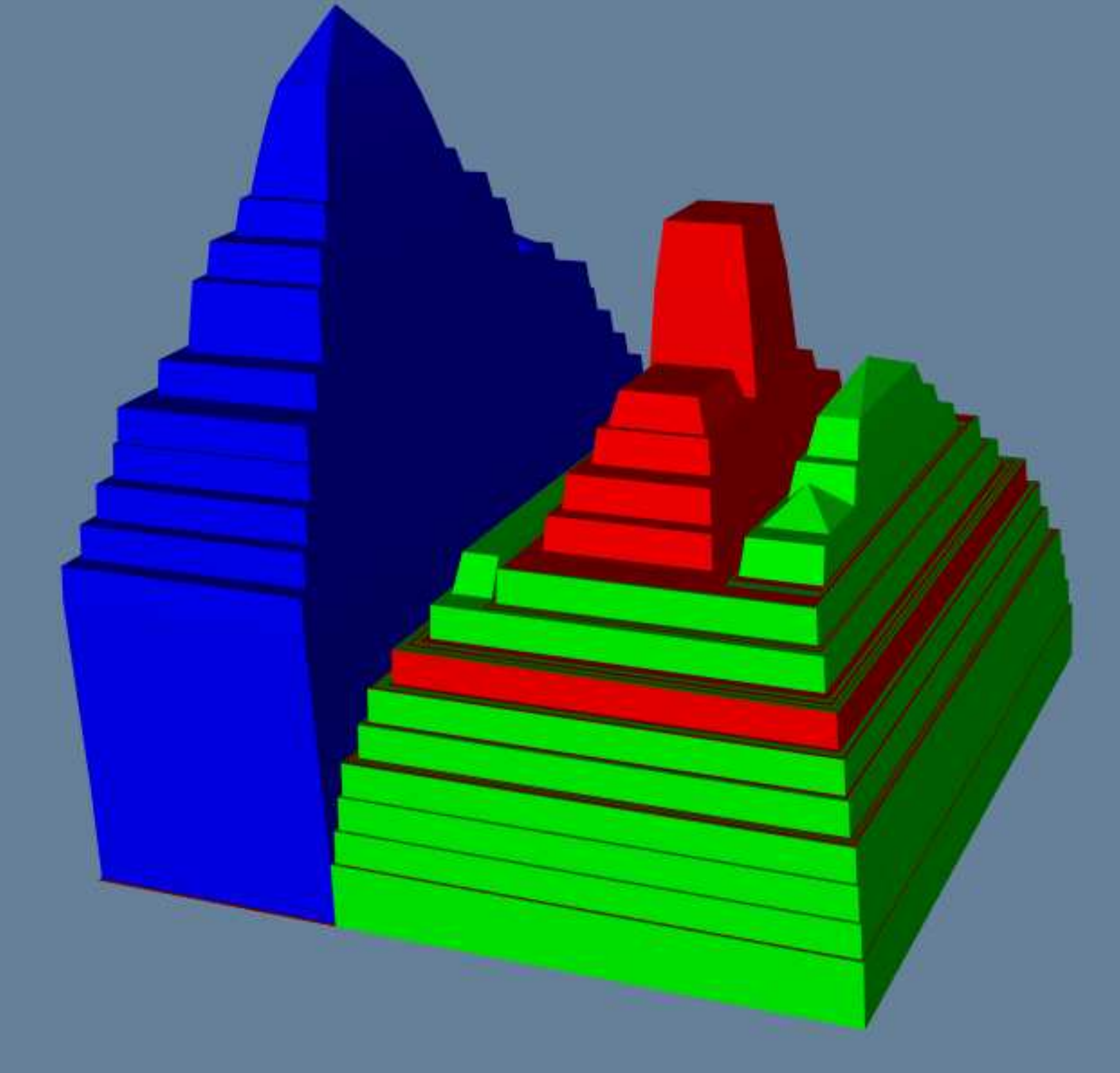}
\label{betweenness_color}
}
\vspace{-0.10in}
\caption[Optional caption for list of figures]{Terrain Visualization of Plant Dataset Query Result}
\label{plant}
\vspace{-0.05in}
\end{figure}

%% file: UserStudy.tex
\section{User Study}
\label{user_sec}

In this section we briefly report on a user study which evaluates the effectiveness of our approach.
Participants were recruited
under OSU IRB (2015B0249) protocol, and follows recommendations
of the Nielsen Norman Group (NNG).  Ten participants, a sufficient number per NNG recommendations,
 were recruited, for {\it each} of the following tasks (no overlap).
Each task was designed to require participants' to solve a real world
problem using multiple visualization platforms.
Participants were recruited from across the university campus, and were largely drawn from quantitatively inclined specializations, due to the nature of the three tasks and as a representation of the domain scientists that may use such tools (e.g. Mathematical and Physical Sciences, Psychology, Engineering, Medicine, Finance). Each participant was given a training tutorial before their prescribed task and were allowed to familiarize themselves with the functionality of the three tools used in the study.

\subsection{Tasks}

\noindent
\emph{\textbf{Task 1:} Identify the densest K-Core in the graph.}\\
The users participating in the first task were
asked  to identify the densest K-Core in each of three graphs (GrQc, PPI, DBLP).
For each graph, we pre-computed the K-Core value (KC(v)) of each vertex, allowing
participants to visualize each graph using following three different visualization methods:\\
(1) Our terrain visualization using K-Core value as scalar value.\\
(2) LaNet-vi which is a K-Core visualization tool~\cite{kcore_vis}.\\
(3) OpenOrd which is a multilevel graph layout method~\cite{open_ord} and uses
  vertex's color to represent its K-Core value.\\
Task 1 is meaningful to data mining and network science researchers who want to explore K-Cores in a graph.
The densest K-Core usually indicates a significant group of closely related nodes in the graph, such as an important community in a social network.

\noindent
\emph{\textbf{Task 2:} Identify the second densest K-Core in the graph that are not connected to the
densest K-Core.}\\
In this task we provided participants with information on the densest K-Core in each dataset, and
asked them to find the next densest K-core which is disconnected from the densest K-core. This is a slightly
nuanced, and more complicated variant of Task 1.
We note that when identifying the second K-Core for analysis
simply choosing the second densest K-Core could be meaningless,
because it might be heavily overlapped with the densest K-Core, and the two K-Cores are actually the same group of closely related nodes.
It is more meaningful to identify the densest K-Core in the graph that are not connected to the previously detected one, because such a K-Core would indicate a separate module-of-interest.

\noindent
\emph{\textbf{Task 3:}
The third task was the most complicated. Participants were
asked to solve a problem involving multiple scalar field visualization on a single (Astro) dataset.}\\
For the terrain visualization method, we provided the
betweenness centrality of each vertex as the first scalar value to generate the terrain,
and provided the degree centrality of each vertex as the second scalar value to color the terrain (Figure~\ref{betweenness_terrain}).
In the visualization generated by OpenOrd, vertex color was used to indicate betweenness centrality, and vertex size to indicate degree centrality (Figure~\ref{betweenness_color}).
We ask the ten participants to  determine whether the betweenness centrality and degree centrality are positively or negatively correlated in the Astro dataset.
In Task 3 we do not compare with LaNet-vi, because it is specifically designed for visualizing K-Cores,
and is not trivially adaptable to visualize two centralities.

\subsection{Results}

\begin{figure*}[ht]
\centering
\subfigure[GrQc (Terrain)]
{
\hspace*{-0.40cm}\includegraphics[width=0.18\textwidth]{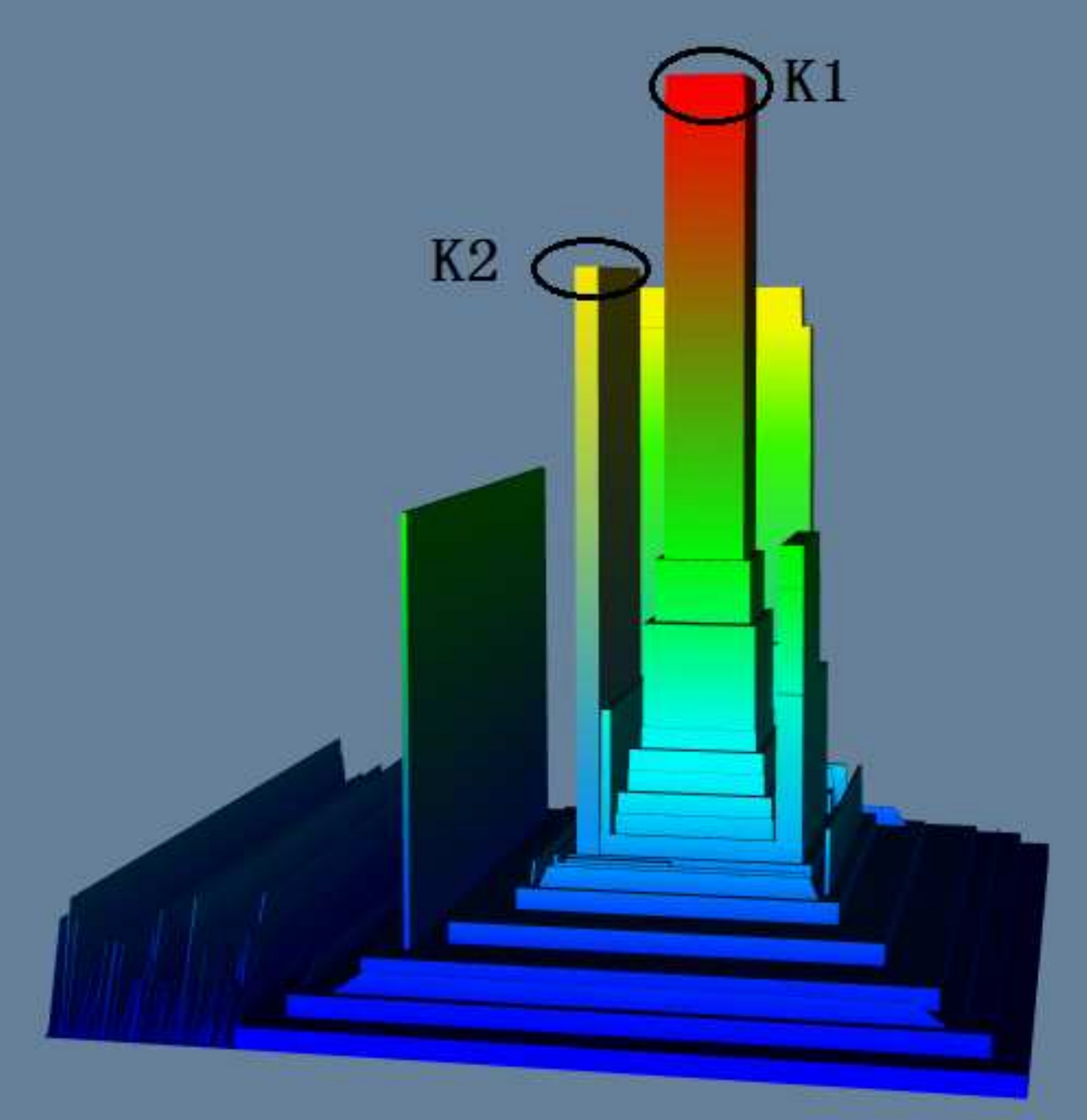}
\label{GrQc_terrain}
}
 \quad
\subfigure[GrQc (LaNet-vi)]
{
\hspace*{-0.20cm}\includegraphics[width=0.21\textwidth]{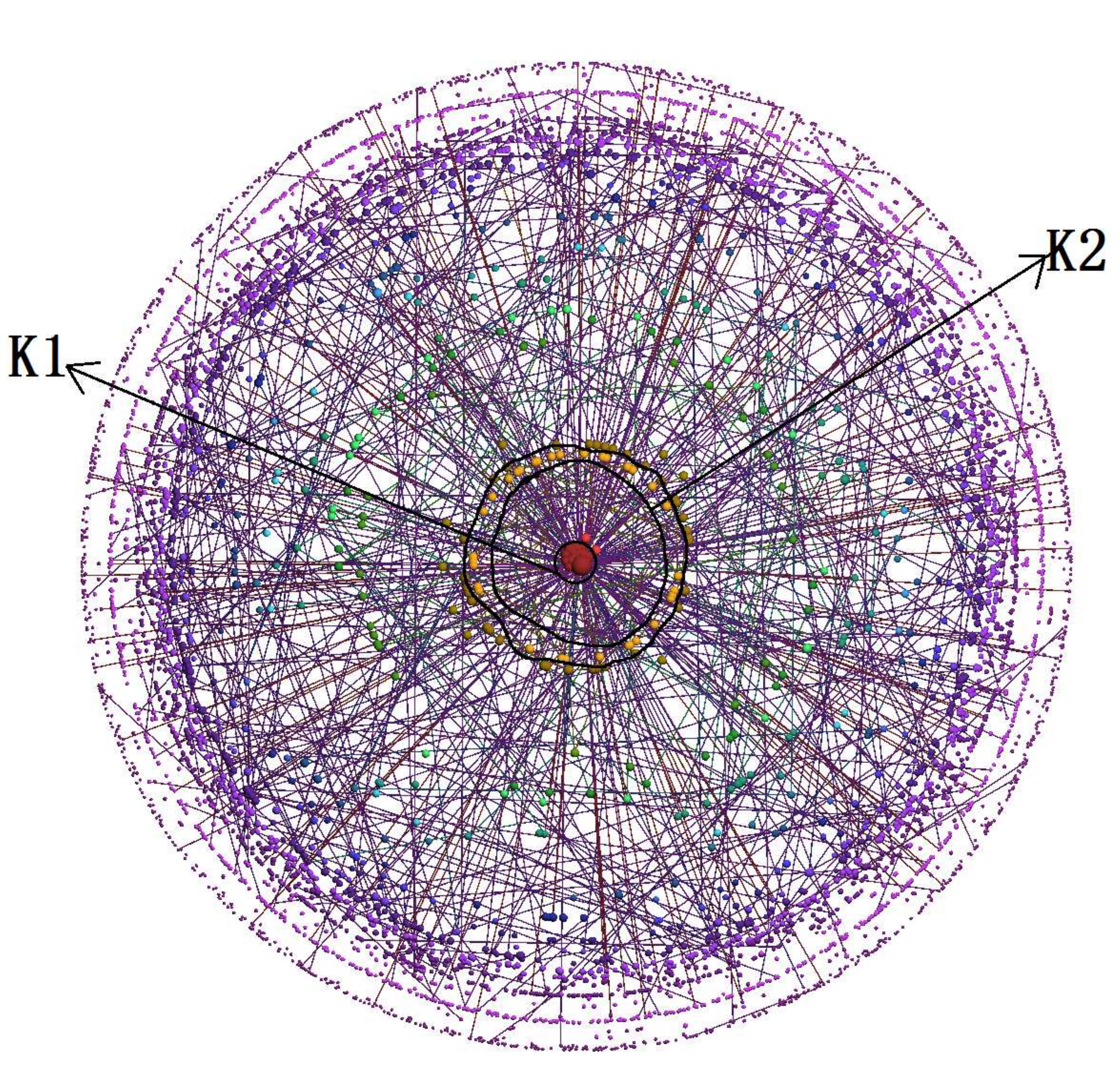}
\label{GrQc_kcore}
}
 \quad
\subfigure[GrQc (OpenOrd)]
{
\hspace*{-0.0cm}\includegraphics[width=0.21\textwidth]{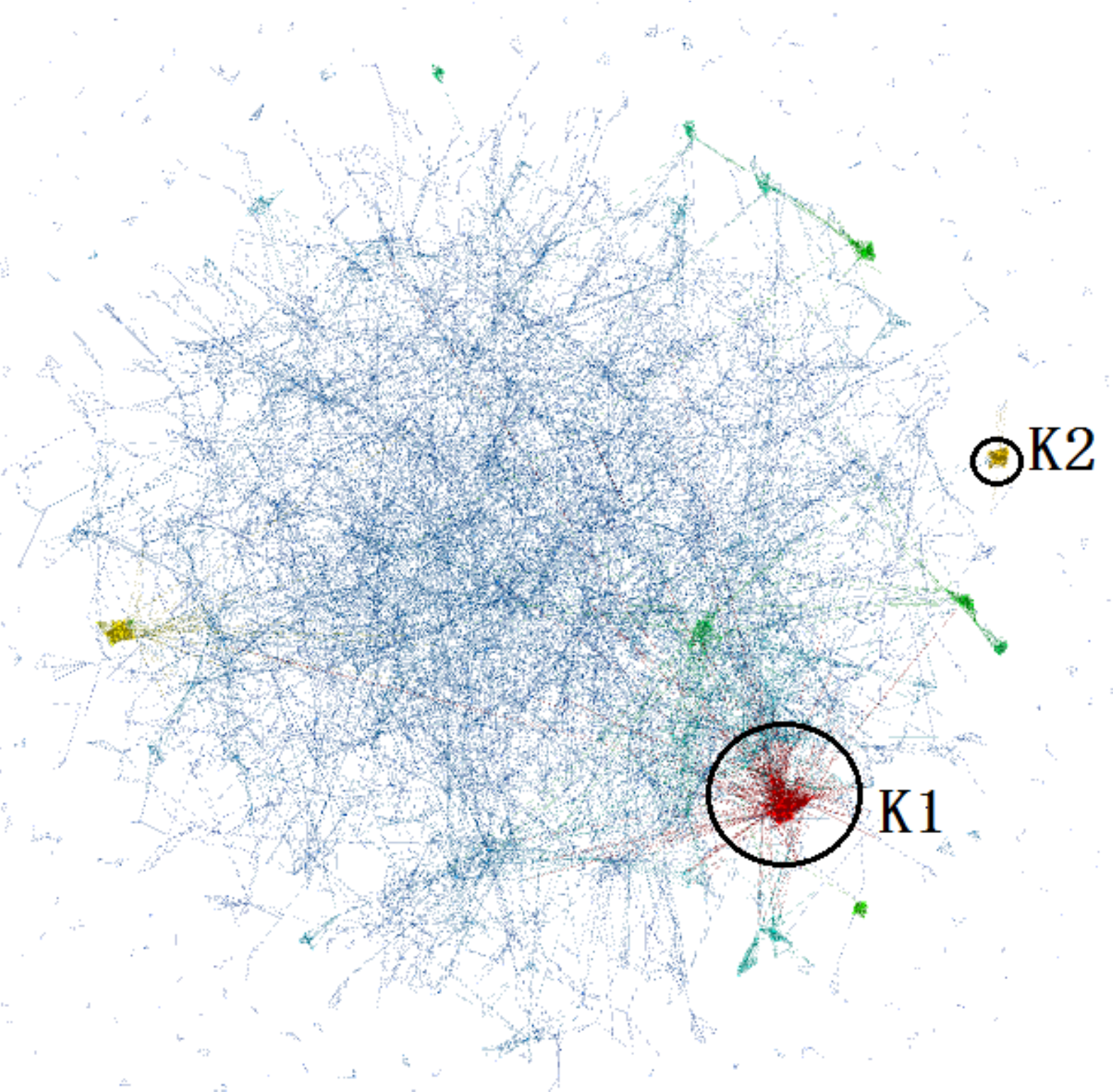}
\label{GrQc_OpenOrd}
}

\subfigure[PPI (Terrain)]
{
\hspace*{-0.0cm}\includegraphics[width=0.18\textwidth]{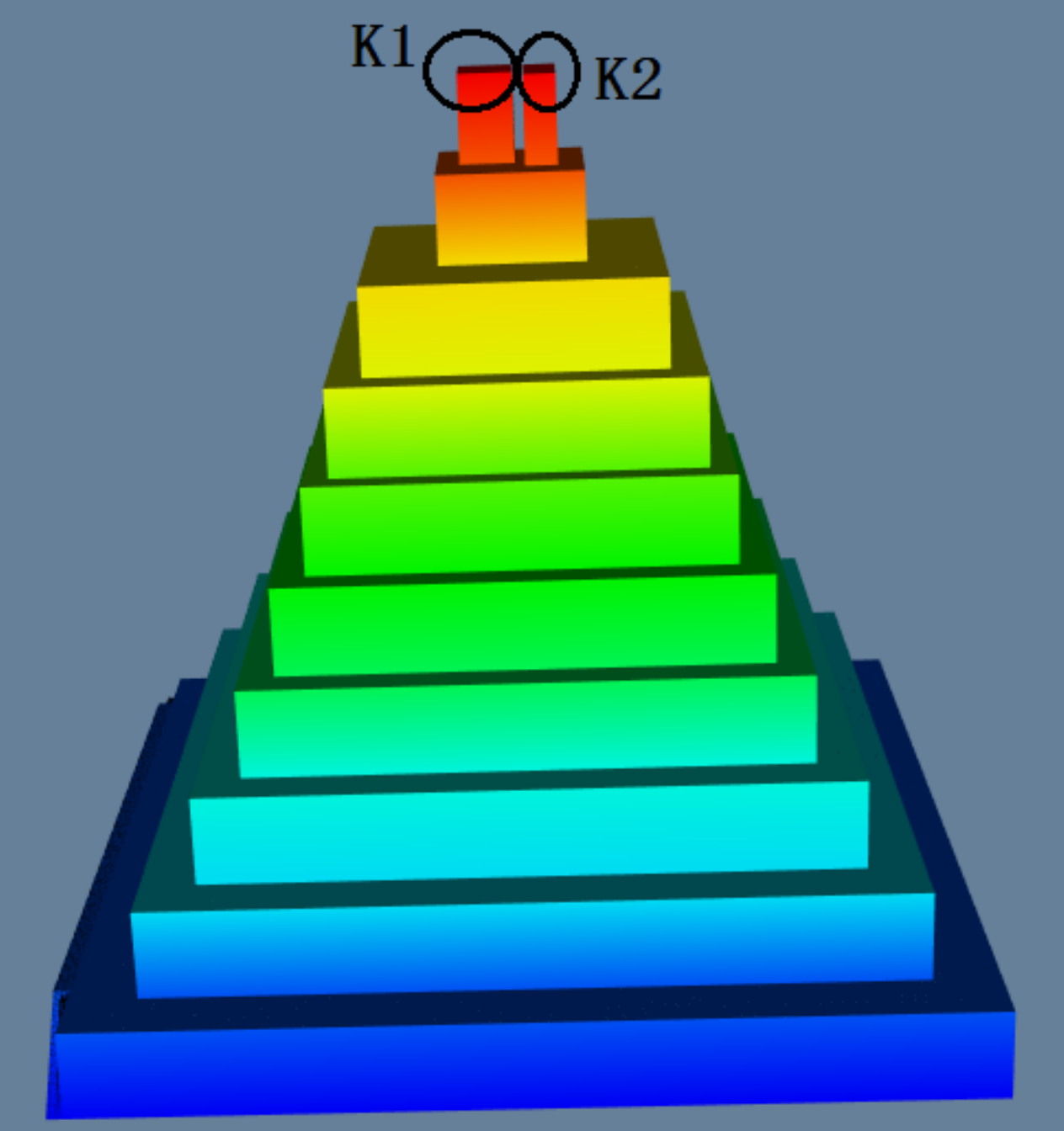}
\label{gene_terrain}
}
 \quad
\subfigure[PPI (LaNet-vi)]
{
\hspace*{-0.0cm}\includegraphics[width=0.22\textwidth]{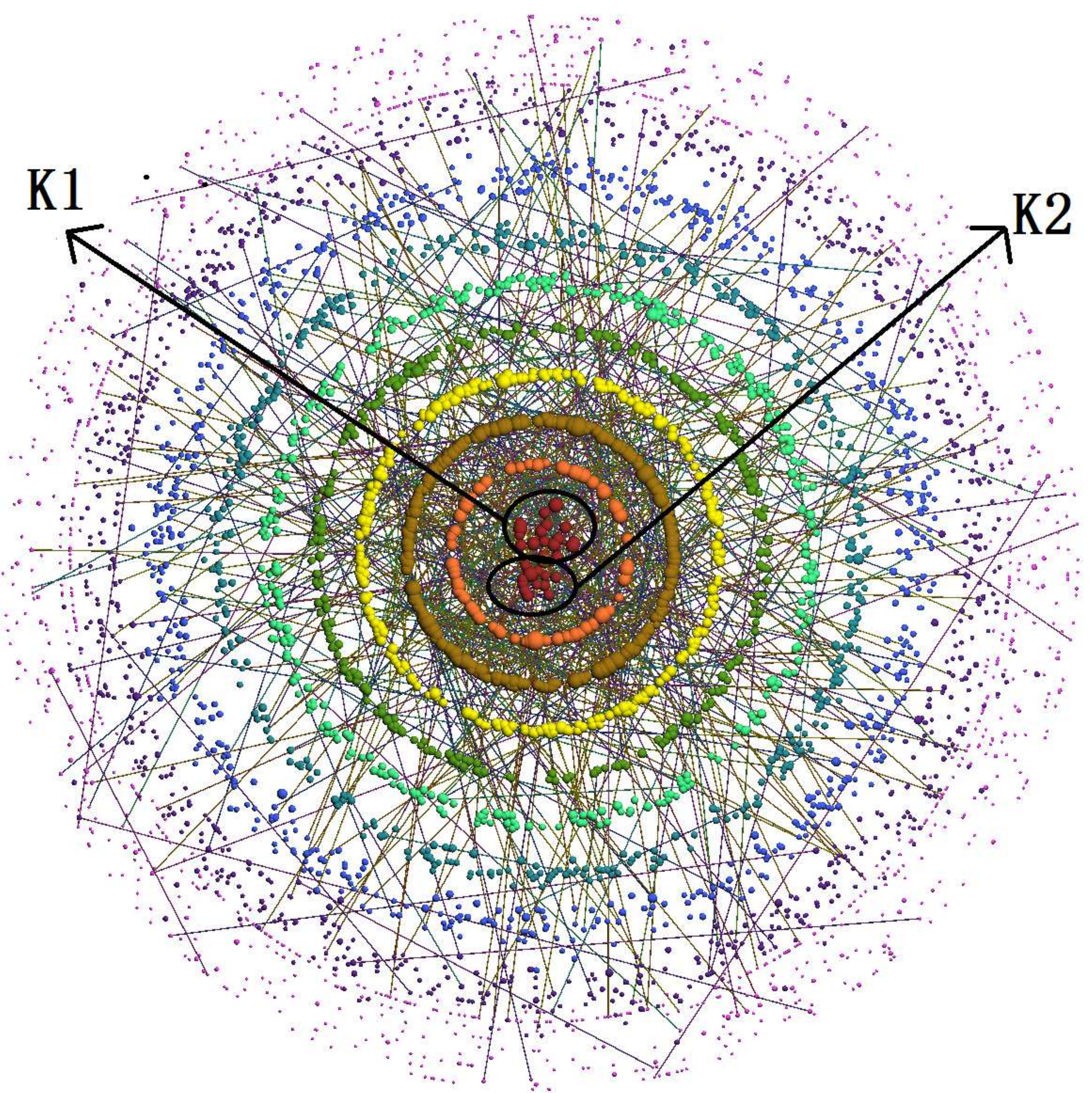}
\label{gene_kcore}
}
 \quad
\subfigure[PPI (OpenOrd)]
{
\hspace*{-0.0cm}\includegraphics[width=0.21\textwidth]{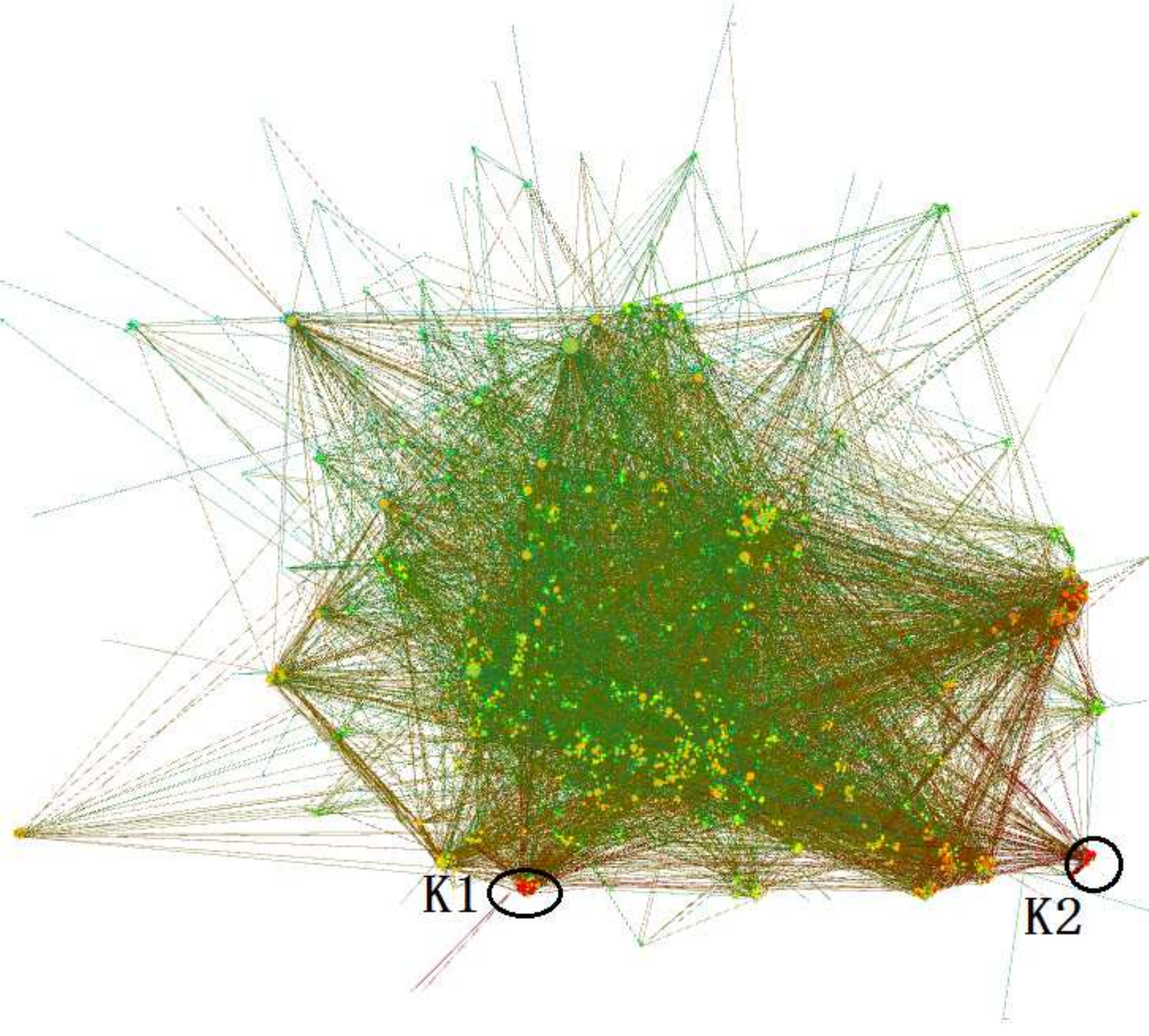}
\label{Gene_OpenOrd}
}

\subfigure[DBLP (Terrain)]
{
\hspace*{-0.0cm}\includegraphics[width=0.18\textwidth]{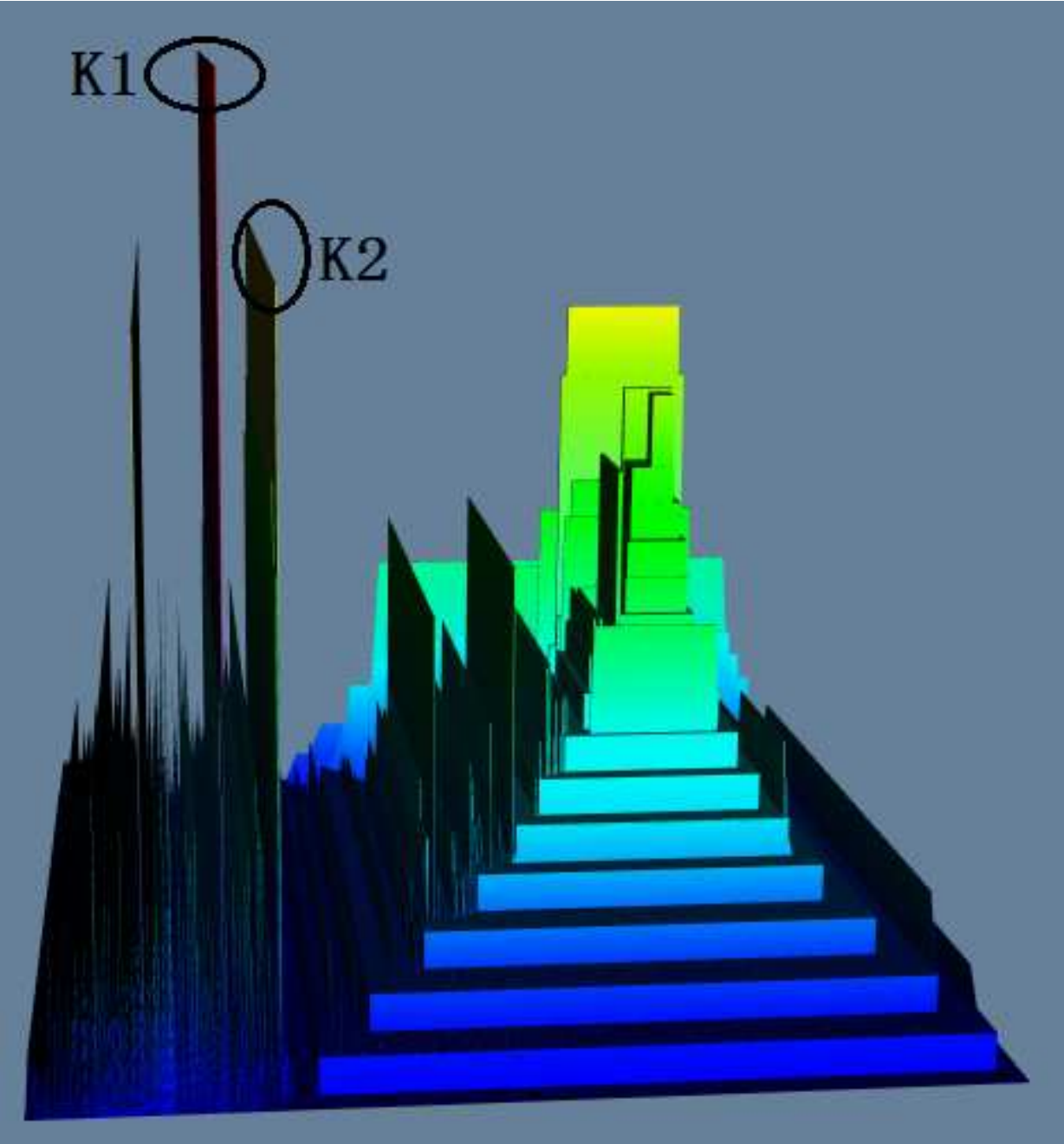}
\label{dblp_terrain}
}
 \quad
\subfigure[DBLP (LaNet-vi)]
{
\hspace*{-0.0cm}\includegraphics[width=0.22\textwidth]{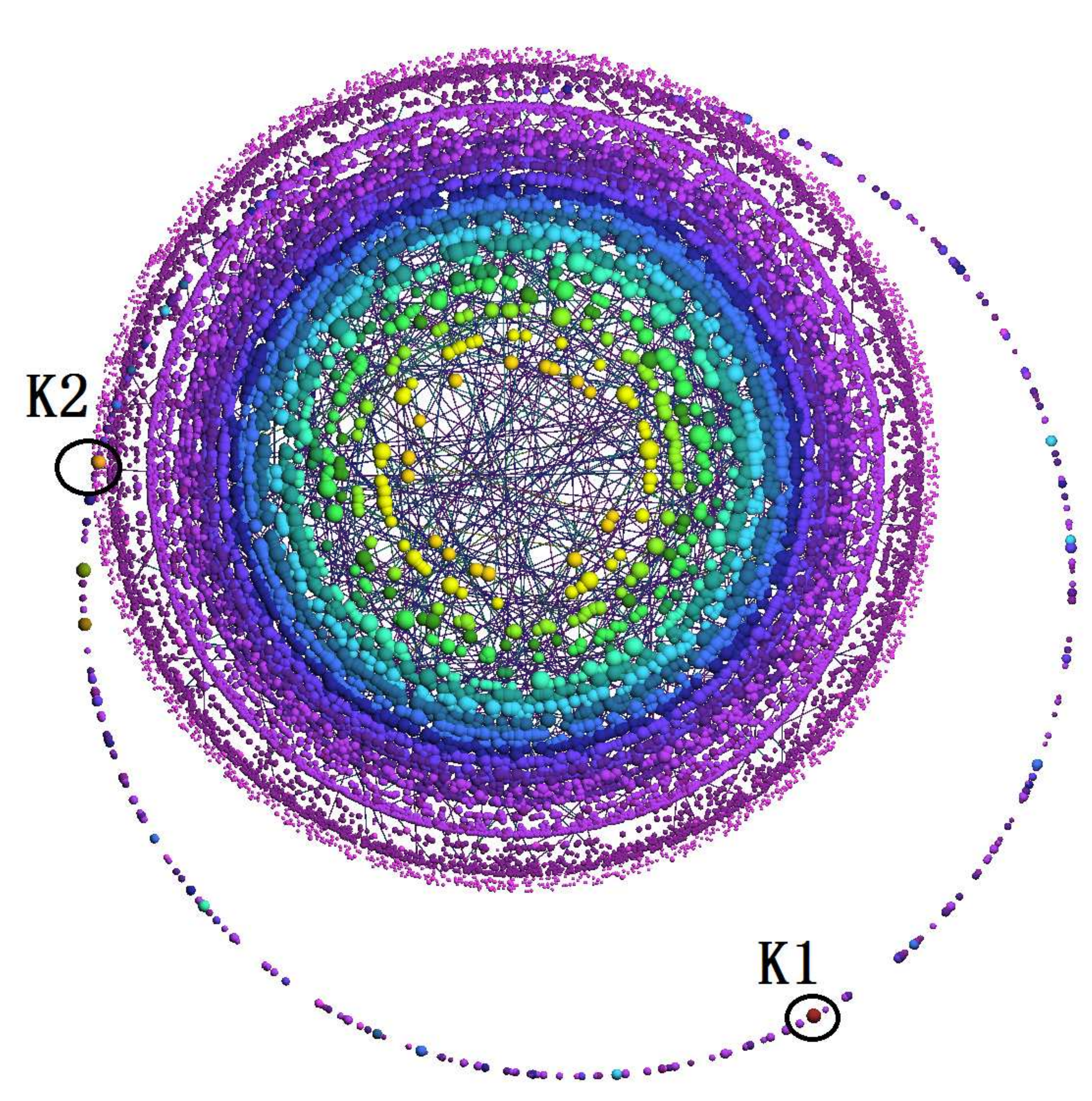}
\label{dblp_kcore}
}
 \quad
\subfigure[DBLP (OpenOrd)]
{
\hspace*{-0.0cm}\includegraphics[width=0.21\textwidth]{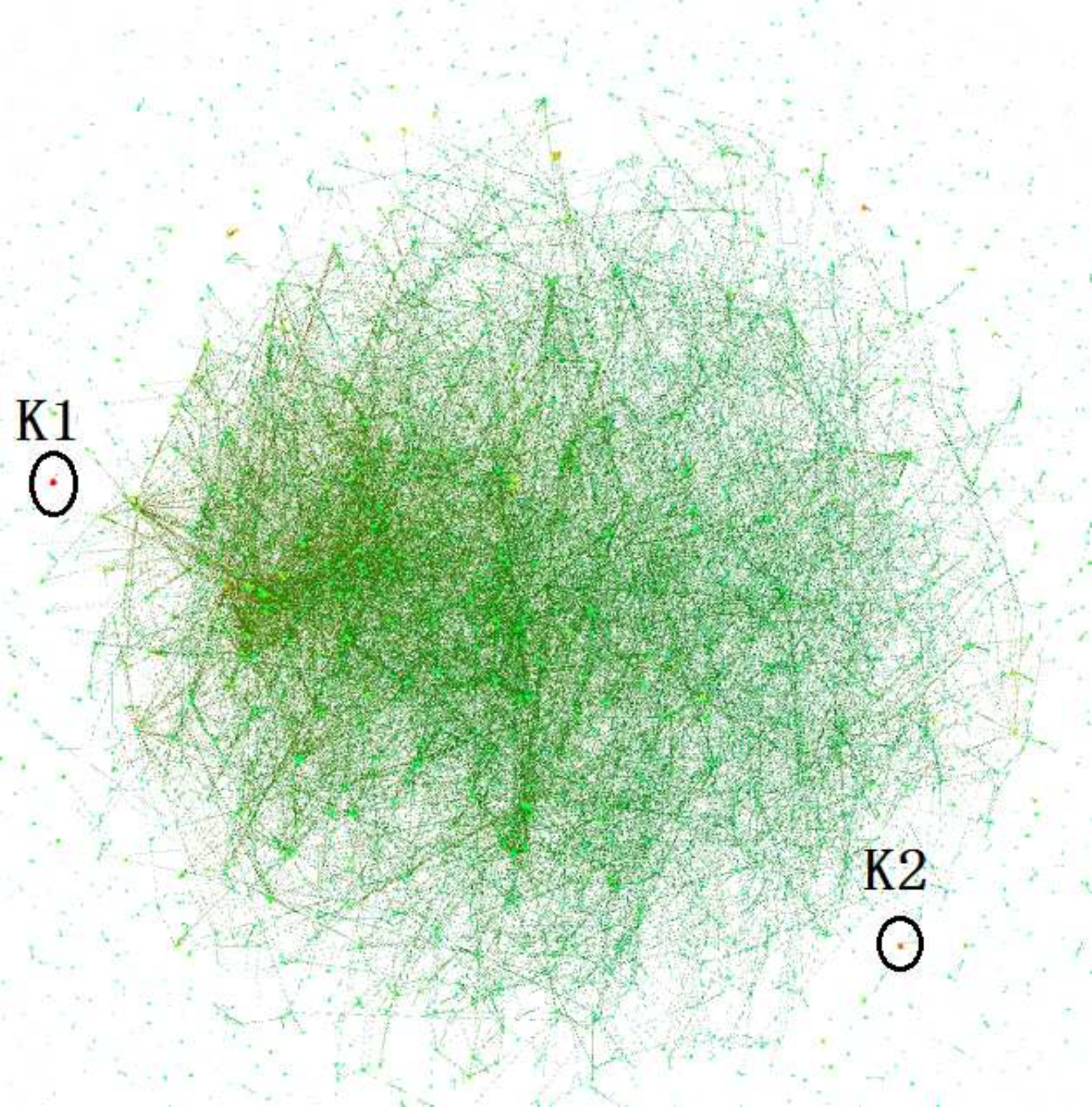}
\label{dblp_OpenOrd}
}
\vspace{-0.05in}
\caption[Optional caption for list of figures]{Visualizations of Single Scalar Field (GrQc, PPI and DBLP)}
\label{user_study}
\end{figure*}

For Task 1 and Task 2, we list the average completion time and
accuracy of all users in Table~\ref{task1} and Table~\ref{task2}.
All the pictures generated by the three visualization methods on the three datasets are from Figure~\ref{GrQc_terrain} to Figure~\ref{dblp_OpenOrd}.
In each picture we label the K-Core to be identified in Task 1/Task 2 as K1/K2.

\begin{table}[!h]
\fontsize{8}{10}\selectfont
\centering
\caption{Average Time(sec) and Accuracy of Task 1}
\begin{tabular}{|c|c|c|c|c|c|c|c|c|}\hline
   & \multicolumn{2}{c|}{Terrain}& \multicolumn{2}{c|}{LaNet-vi} & \multicolumn{2}{c|}{OpenOrd} \\\hline
  Dataset & accuracy & time & accuracy & time & accuracy & time \\\hline
  GrQc &  1 &  2.6  &  1 & 6.7  & 1 & 7.6 \\\hline
  PPI &  1 & 4.9  &  1  & 5.3 & 0.8  & 10.7\\\hline
  DBLP & 1  & 4.6  & 0.8  & 6.6  &  1 & 10.9\\\hline
\end{tabular}
\label{task1}
\end{table}

\vspace{-0.05in}

 \begin{table}[!h]
\fontsize{8}{10}\selectfont
\centering
\caption{Average Time(sec) and Accuracy of Task 2}
\begin{tabular}{|c|c|c|c|c|c|c|c|c|}\hline
   & \multicolumn{2}{c|}{Terrain}& \multicolumn{2}{c|}{LaNet-vi} & \multicolumn{2}{c|}{OpenOrd} \\\hline
  Dataset & accuracy & time & accuracy & time & accuracy & time \\\hline
  GrQc &  1 &  3.6  &  0.7 & 10.3   & 0.8 & 8.2\\\hline
  PPI &  1 & 4.1  &  0.2  & 7.7 & 0.7  & 11.6\\\hline
  DBLP & 1  & 5.1  & 0.8  & 8.5  &  0.9 & 9.8\\\hline
\end{tabular}
\label{task2}
\end{table}

\vspace{-0.05in}

 \begin{table}[!h]
\fontsize{8}{10}\selectfont
\centering
\caption{Average Time(sec) and Accuracy of Task 3}
\begin{tabular}{|c|c|c|c|c|c|c|}\hline
   & \multicolumn{2}{c|}{Terrain}& \multicolumn{2}{c|}{OpenOrd} \\\hline
  Dataset & accuracy & time & accuracy & time \\\hline
  Astro & 0.9 & 9.1   &  0.7 & 11.9 \\\hline
\end{tabular}
\label{task3}
\end{table}

In Table~\ref{task1}, we can see that all users successfully finished Task 1 by using terrain visualization.
Two users failed by using LaNet-vi on DBLP dataset, and two users failed by using OpenOrd on the PPI dataset.
The reason is that the densest K-Core in these two visualizations are small and not obvious,
 and they did not notice it or they couldn't correctly identify the K-Core value through the color.
Anecdotally we observed that
although they zoomed-in the two visualizations to see a larger picture,
they lost the full context and could only see portion of the picture,
leading them to choose incorrect densest K-Cores.

Table~\ref{task2} shows that all users successfully finished Task 2 by terrain visualization,
while some users failed by using LaNet-vi and OpenOrd.
One reason is Task 2 requires users to understand the connectivity to the densest K-Core (to avoid finding
one with significant overlap),
the LaNet-vi and OpenOrd both draw edges to indicate connections,
and users need to check the edges carefully to determine whether two K-Cores are connected,
it is time consuming and led to mistakes being made.
In both tasks, we can see from Table~\ref{task1} and Table~\ref{task2} that users spent least time on terrain visualization.

The results of Task 3 is presented in the Table~\ref{task3}, and the visualization pictures are in Figure~\ref{betweenness_terrain} and Figure~\ref{betweenness_color}.
The result shows that users achieved higher accuracy and spent less time on terrain visualization.
In the visualization generated by OpenOrd, some nodes are blocked by other nodes, which caused  users to make
the incorrect decision.
One participant also pointed out that it is easier to identify the outlier areas (circle in Figure~\ref{betweenness_terrain}) in the terrain visualization than in the OpenOrd visualization. Anecdotally,several (more than 10) users found the ability to rotate, as well as the linked 2D-Maps,
a useful feature in the Terrain Visualization strategy.

\begin{figure}[!h]
\centering

\subfigure[Astro (Terrain)]
{
\includegraphics[width=0.17\textwidth]{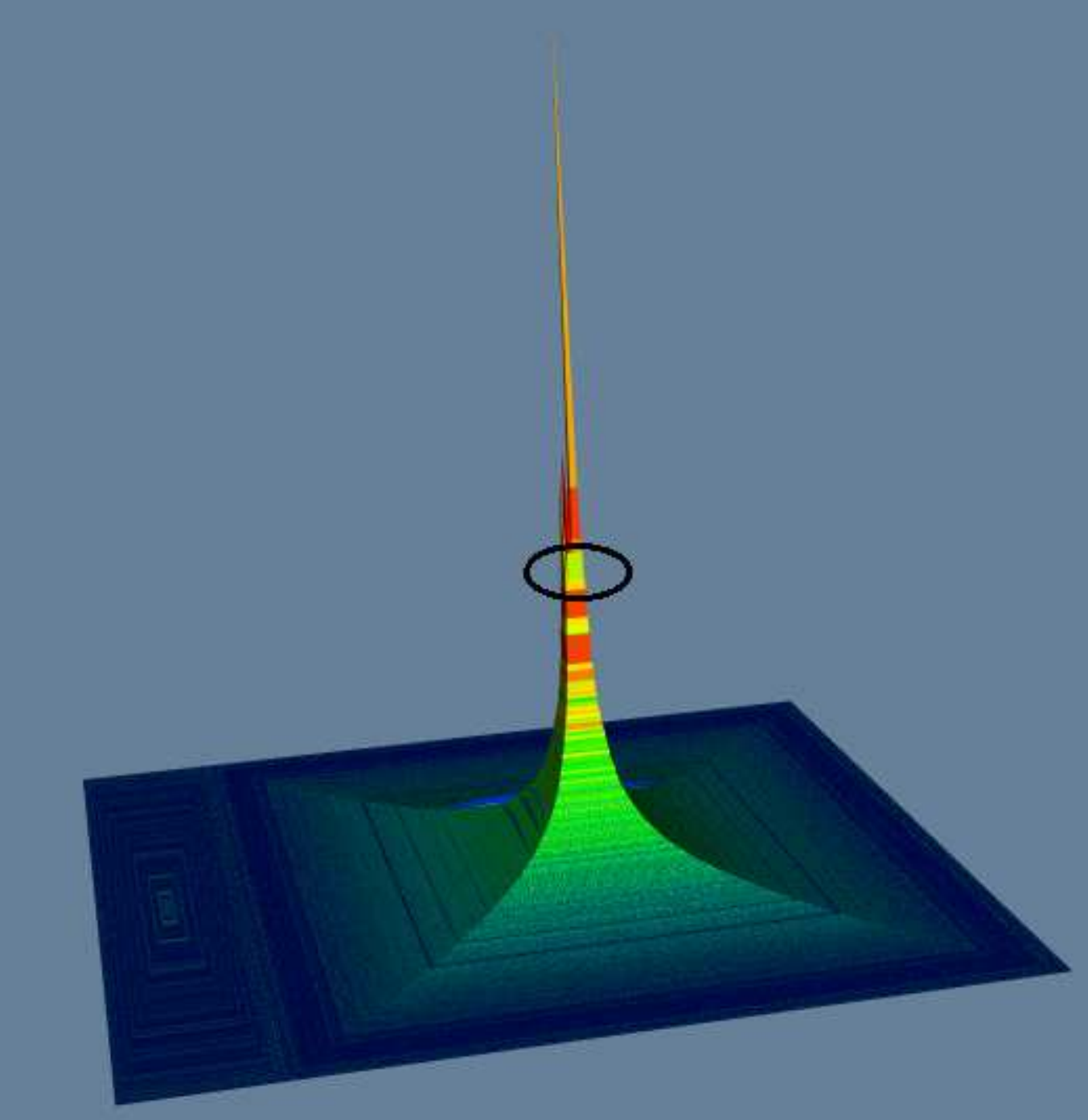}
\label{betweenness_terrain}
}
 \quad
\subfigure[Astro (OpenOrd)]
{
\includegraphics[width=0.18\textwidth]{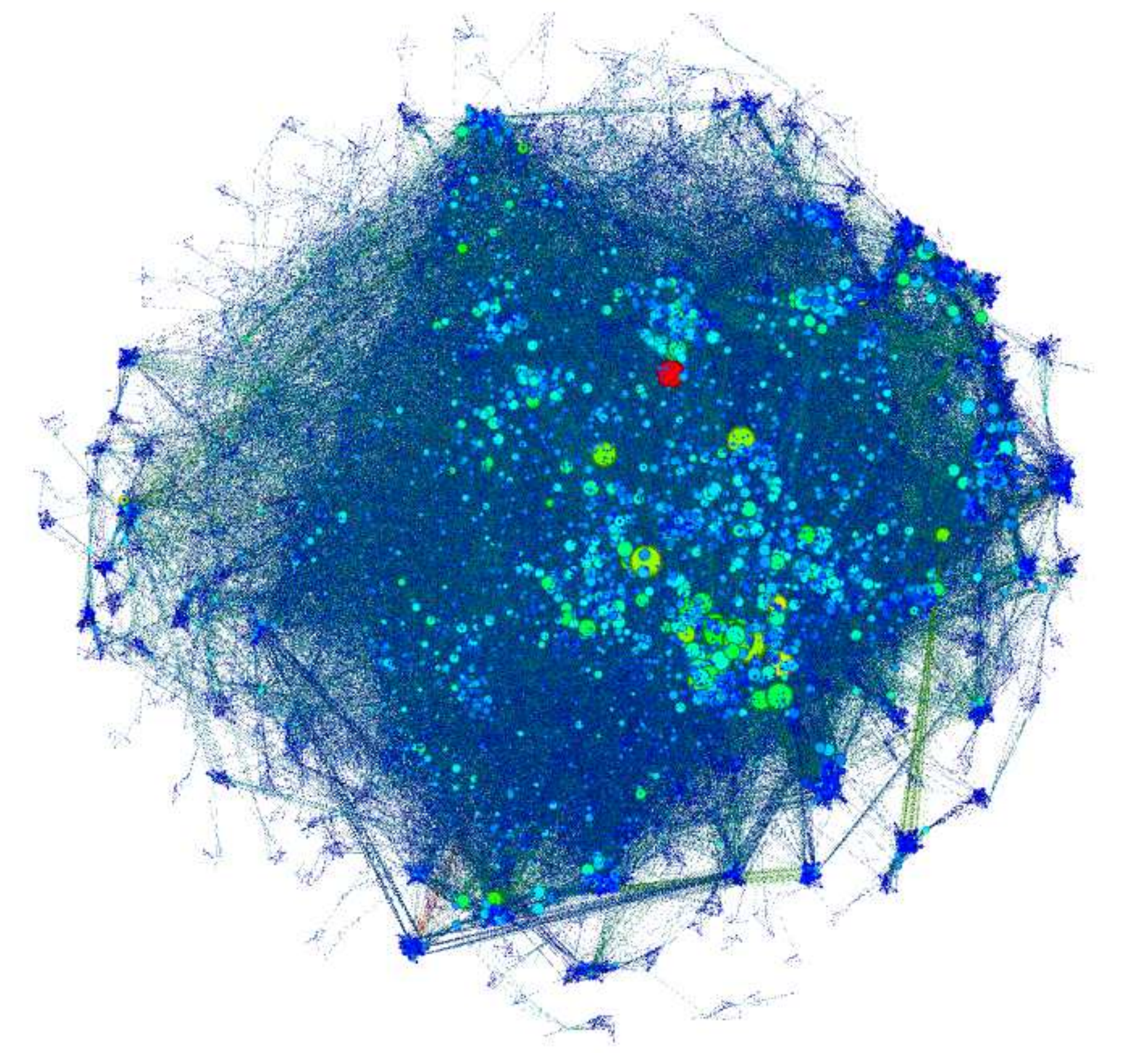}
\label{betweenness_color}
}

\vspace{-0.05in}
\caption[Optional caption for list of figures]{Visualizations of Task 3 in User Study}
\label{user_study2}
\end{figure}